\documentclass[12pt]{article}
\usepackage[utf8]{inputenc}
\usepackage{float}
\usepackage{url}
\usepackage{placeins}
\usepackage[T1]{fontenc}
\usepackage{mathptmx}                 
\usepackage[margin=1.0in,top=1.12in,bottom=1.10in]{geometry}
\usepackage{amsmath,amssymb,amsthm,mathtools}
\usepackage{bm}
\usepackage{setspace}
\usepackage[longnamesfirst,sort&compress]{natbib}
\usepackage{booktabs}
\usepackage{enumitem}
\usepackage{titlesec}
\usepackage{titling}
\usepackage[dvipsnames]{xcolor}
\usepackage{etoolbox}
\usepackage{graphicx}
\usepackage{subcaption}
\usepackage{adjustbox}
\usepackage{rotating}
\usepackage{booktabs}
\usepackage[hidelinks,breaklinks=true]{hyperref}
\hypersetup{colorlinks=true,
            linkcolor=blue!55!black,
            citecolor=blue!55!black,
            urlcolor=blue!55!black}

\DeclareFontFamily{U}{mathx}{\hyphenchar\font45}
\DeclareFontShape{U}{mathx}{m}{n}{
  <5> <6> <7> <8> <9> <10>
  <10.95> <12> <14.4> <17.28> <20.74> <24.88>
  mathx10
}{}

\DeclareSymbolFont{mathx}{U}{mathx}{m}{n}
\DeclareMathAccent{\widecheck}{0}{mathx}{"71}
\bibpunct{(}{)}{;}{a}{,}{,}

\titleformat{\section}
  {\normalfont\Large\scshape\filcenter}{\thesection.}{0.7em}{}
\titlespacing*{\section}{0pt}{2.2ex plus 1ex minus .2ex}{1.4ex plus .2ex}

\titleformat{\subsection}
  {\normalfont\large\itshape}{\thesubsection.}{0.6em}{}
\titlespacing*{\subsection}{0pt}{1.8ex plus .6ex minus .2ex}{1.0ex plus .2ex}

\titleformat{\subsubsection}
  {\normalfont\normalsize\itshape}{\thesubsubsection.}{0.6em}{}
\pretitle{\centering\normalfont\LARGE}
\posttitle{\par\vskip 0.6em}

\preauthor{\centering\normalfont\normalsize
           \textsc{By}\vskip 0.3em}
\postauthor{\par\vskip 2em}

\predate{\centering\normalfont\small}
\postdate{\par}

\theoremstyle{plain}
\newtheorem{proposition}{\scshape Proposition}
\newtheorem{theorem}{\scshape Theorem}
\newtheorem{lemma}{\scshape Lemma}

\theoremstyle{definition}
\newtheorem{assumption}{\scshape Assumption}

\theoremstyle{remark}
\newtheorem{remark}{\bfseries\upshape Remark}

\usepackage{bibunits}

\defaultbibliographystyle{ecta}
\defaultbibliography{references}

\title{Intentionally Overidentifying Local Projections}

\author{%
\begin{tabular}{c@{\hspace{2em}}c@{\hspace{2em}}c}
Junjie Guo\footnotemark[1] & Endong Wang\footnotemark[2] & Ke-Li Xu\footnotemark[3]\\[2pt]
{\small \textit{CUFE}} &
{\small \textit{University of Mannheim}} &
{\small \textit{Indiana University}}
\end{tabular}%
}

\date{\today}
\begin{document}
\maketitle

\footnotetext[1]{School of Finance, Central University of Finance and Economics, Shahe Higher Education Park, Changping District, Beijing, P. R. China, 102206. 
Email: \href{mailto:junjguog@cufe.edu.cn}{junjguog@cufe.edu.cn}.}

\footnotetext[2]{Department of Economics, University of Mannheim, L7, 3--5, Mannheim, Germany, 68161. Email: \href{mailto:endong.wang@uni-mannheim.de}{endong.wang@uni-mannheim.de}.}

\footnotetext[3]{Department of Economics, Indiana University, 100 S. Woodlawn Ave., Bloomington, IN 47405, USA. Email: \href{mailto:kelixu@iu.edu}{kelixu@iu.edu}.}


\begin{bibunit}
\begin{center}
\begin{minipage}{0.88\textwidth}
\small
\noindent\textsc{Abstract.}\quad
Local projections (LPs) have become a standard tool for estimating dynamic effects in macroeconomics and are typically implemented using ordinary least squares (OLS) for a given lag order. We propose an alternative approach that exploits \textit{additional lags} as excluded \textit{instrumental variables} (IVs). Under a finite-order model with martingale-difference shocks, the resulting (optimally weighted) IV estimator is never asymptotically less efficient than OLS and, as we show, can deliver substantial efficiency gains when data are persistent. For the unit-root AR(1) data with i.i.d. errors, our baseline IV estimator (with one overidentifying restriction) has asymptotic variance $2-h^{-1}$, compared with $h$, the asymptotic variance of the OLS lag-augmented estimator in \cite{montiel2021local}, for an integer propagation horizon $h$. We then develop IV-based LP inference that is robust to unknown persistence and moderate local model misspecification. We also characterize the trade-off between the degree of misspecification and the number of additional lag controls. Through simulations and an empirical application, we demonstrate that the proposed method can yield meaningfully shorter confidence intervals for impulse responses than OLS in finite samples.

\end{minipage}
\end{center}

\vspace{1em}


\section{Introduction}
A central task in empirical macroeconomics and finance is to estimate how shocks propagate over time \citep{sims1980macroeconomics,ramey2016macroeconomic}. Since \citet{jorda2005estimation}, local projections (LPs, hereafter) have become a standard tool for this purpose because they estimate horizon-specific impulse responses directly, and can readily accommodate a range of empirically relevant features, including nonlinearities, instability, and rich sets of controls.\footnote{See \cite{jorda2025localprojections,inoue2024unstable,dinh2024rslp,ferreira2025blp, gonccalves2024state}, among many others.} 

In recent work, \cite{montiel2021local} and \cite{olea2024double} formally prove that, under the finite-order model, the lag-augmented LP-based inference (i.e. using a sufficient number of controlled lags) is robust to persistence of the data, long horizons, and a nontrivial degree of model misspecification, strengthening the case for LPs relative to other methods of inference on impulse responses. On the other hand, it is well known that the LP estimator tends to have large variance and provide wide confidence intervals.

The current paper extends this line of research by proposing a new LP estimator which exploits the additional information implied by the model. Concretely, instead of using OLS (ordinary least squares), as is common in practice, we  intentionally overidentify the local projection by adding orthogonality conditions based on extra lags of the data to construct the optimally weighted IV (via the generalized method of moments (GMM), \cite{hansen1982large}) estimator of the LP regression. While such an estimator is guaranteed to be no less efficient than OLS, we ask whether and when the efficiency gain can be nontrivial and whether the robustness properties established in the earlier literature on local projections can be retained.

Our analysis shows that the new method has considerable potential. Additional orthogonal conditions can yield substantial efficiency gains over LP OLS, arising from the specific structure of the LP regression. As a teaser, additional lags are not necessarily informative. For the autoregression (namely, the horizon-one local projection), lag augmented or not, with conditionally homoskedastic martingale-difference errors (e.g., i.i.d. errors), additional lags are generally redundant excluded instrumental variables; the additional moment conditions do not generate an asymptotically more efficient estimator than OLS. However, such redundancy does not extend to a general LP regression at horizons larger than one, even under conditional homoskedasticity. The optimally weighted GMM can produce a strict efficiency gain over OLS, especially when the data are strongly persistent. The efficiency gain arises primarily from the serial correlation of the GMM (residualized) score process, even when the OLS (residualized) score process can be serially uncorrelated \citep{montiel2021local}. Under conditional heteroskedasticity, as our general inferential theory allows, using additional lags as excluded instruments generally yields non-redundant moment conditions (even for the horizon-one LP regression; see \citet{kuer02ET}).

Regarding robustness, we show that controlling for a sufficient number of lags in the LP regression remains crucial for valid inference with our GMM estimator, robust to both data persistence and moderate model misspecification.

\textbf{Literature}. 
The terminology of instruments and GMM has appeared in the VAR/LP literature, but their uses and purposes are all different from ours. The instruments used in the current paper are internal lags that improve efficiency. The method applies to reduced-form impulse responses and, under identification, to structural responses as well. Our use of instruments differs from the recent literature on internal and external instruments for \textit{structural identification} \citep{stock2018identification, lewis2021volatility}. \citet{inoue2026inference} similarly refer to their method as GMM LP, but their approach follows the literature on structural identification. \citet{JordaKoz2011} use GMM and minimum-distance methods to identify and estimate the structural parameter in the solution of a linear rational expectations model. 

Instruments used for structural identification can be subject to the well-known weak-IV problem \citep{Mont-SW}, whereas our use of instruments generally is not. Under the standard regularity assumptions, the extra weakly-correlated lags can at most reduce or eliminate the efficiency gain (over OLS) from overidentification, and do not themselves threaten the identification or inference.

\citet{wang2024} use reduced-form VAR innovations as instruments to conduct two-step estimation and associated inference for the generalized impulse responses and multi-horizon Granger causality in the sense of \citet{dufour1998short}. Their motivation is mostly the simplicity of the (persistence-robust) inference (no HAC variance estimator needed, albeit under an additional assumption on innovations) and the ability to handle a larger class of causal parameters. Their instruments \textit{exactly identify} the model, thus the resulting estimator of the impulse response can be more or less efficient than OLS. The current paper focuses exclusively on impulse responses, and uses extra lagged outcomes to construct overidentification restrictions which guarantee no asymptotic efficiency loss relative to OLS.

\citet{lusompa2023local} also exploits the serial correlation in LP regression errors, by constructing a generated outcome variable, which involves estimates of nuisance parameters
(i.e., the full set of impulse responses). Our method is thus simpler to implement:
it relies on a standard GMM estimator with readily available instrumental variables and can
be carried out using off-the-shelf software. We further develop robust inference and allow
for model misspecification, features that are not addressed in Lusompa’s analysis.\footnote{A recent paper by \cite{DeVos2026} warns that GLS (generalized least squares) local projections may achieve efficiency gains at the expense of robustness.}

Relatedly, recently developed smoothing, shrinkage, and Bayesian methods\footnote{See \cite{barnichon2019smooth}, \cite{ferreira2025blp}, \cite{huber2025gsulp}, among others.} can improve precision of the LP by imposing structure across horizons and borrowing information across horizon-specific estimates. By contrast, our approach leaves the local-projection estimand unrestricted and obtains efficiency solely from valid overidentifying moments.

In this paper we consider the finite-order model, following a strand of the literature (surveyed in \cite{olea2024double}) that views the low-order models commonly used in practice as approximations to an underlying infinite-order DGP. Recent work by \citet{xu2023local,xu2026bound} establishes conditions under which LP OLS attains the efficiency bound as the lag order diverges. Thus, under these conditions, we do not expect our GMM estimator to deliver nontrivial efficiency gains.

\textbf{Outline}. Section \ref{sec:implementation} presents the method and its implementation. Section~\ref{sec:eff} derives the actual efficiency gain in a scalar AR(1) model. Section~\ref{sec:uniform} develops uniform inference theory under correct model specification, and Section \ref{sec:local} extends the analysis to the locally misspecified model. Sections \ref{sec:mc} and \ref{sec:application} report Monte Carlo evidence and an empirical application. Technical details are provided in the appendices.

\textbf{Notation}. For two random (column) vectors $y$ and $x$, define $P(y\mid x)=B'x$ as the population linear projection, where $B=[\mathbb E(xx')]^{-1}\mathbb E(xy')$ is the projection coefficient matrix. For covariance-stationary random vectors $x_t$ and $y_t$, define
$\operatorname{LRV}(x_t)=\sum_{k=-\infty}^{\infty}
\operatorname{Cov}(x_t,x_{t-k})$ and
$\operatorname{LRCov}(y_t,x_t)=\sum_{k=-\infty}^{\infty}
\operatorname{Cov}(y_t,x_{t-k})$. For symmetric matrices, $\lambda_{\min}(\cdot)$ and
$\lambda_{\max}(\cdot)$ denote the extreme eigenvalues. For rectangular
matrices, $\lambda_{\min}(\cdot)$ denotes the smallest singular value. The sum $\sum_{\ell=a}^{b}x_{\ell}$ is defined as zero if $a>b$.

\section{Overidentified GMM local projections}
\label{sec:implementation}
In this section we outline the method in a general setting. The formal analysis is provided in Sections \ref{sec:uniform} and \ref{sec:local}.
\subsection{Model and target parameter}

Let $\{y_t: t=1,\cdots,T\}$ be a $K$-dimensional process following the
VAR$(d)$ model,
\begin{equation}
\label{eq:vard}
y_t
=
A_1 y_{t-1}+\cdots+A_d y_{t-d}+u_t,
\end{equation}
where $\{u_t\}$ is a mean-zero white noise innovation sequence. Let
$\Sigma_u:=\mathbb E[u_tu_t']$. Deterministic terms
are omitted in our analysis to maintain the focus. In practice, we can think of $y_t$ as the data after removing deterministic components (e.g., by demeaning or detrending). The finite-order VAR in \eqref{eq:vard} is the maintained model; Section~\ref{sec:local} allows departures from the model \eqref{eq:vard}.

For any $p\ge d$, define the lag vector $
Y_{t,p}:=
(y_t',y_{t-1}',\cdots,y_{t-p+1}')'
\in\mathbb R^{Kp}.$
For a given forecast horizon $h\ge1$, let $y_{1,t+h}:=\nu_1'y_{t+h}$, where $\nu_1$ is the first column of $I_K$. The target is the coefficient $\beta_h$ on $y_t$ in the following horizon-$h$
regression (known as the local projection (LP)),
\begin{align}
\label{eq:vard_lp}
y_{1,t+h}
&=
(\beta_h',\eta_{h,p}')Y_{t,p}+e_{t,h},
\qquad t=p,\cdots,T-h,
\end{align}
where $e_{t,h}$ is the projection error. Under \eqref{eq:vard}, the local projection residual follows $e_{t,h}=\nu_1'\sum_{\ell=0}^{h-1}\Psi_\ell u_{t+h-\ell}$, with $\Psi_\ell=J_d\mathbf A^\ell J_d'$, $\mathbf A$ the $dK\times dK$ companion matrix of $\{A_1,\cdots,A_d\}$,
and $J_d=[I_K,0,\cdots,0]$ being $K
\times dK$. Here $\beta_h\in\mathbb R^K$ collects the
coefficients on $y_t$ and $\eta_{h,p}\in\mathbb R^{K(p-1)}$ those on the included
lag controls. The two leading cases are $p=d$, the usual local projection
\citep{jorda2005estimation}, and $p=d+1$, the lag-augmented local projection (LALP)
\citep{montiel2021local}.

The coefficient $\beta_h$ in population is invariant to the lag blocks included
beyond the VAR order: once $Y_{t,d}$ is included, further lags have no
predictive content for $y_{1,t+h}$ under \eqref{eq:vard}, so their population
coefficients are zero. Importantly, these lags can nonetheless affect the behavior
of the estimator.

\subsection{Optimal weighting and HAC inference}

The LP regression \eqref{eq:vard_lp} is commonly estimated by OLS. OLS is simple but only uses a subset of
moment restrictions available: for any $q\ge p$, the deeper lag stack $Y_{t,q}$ is also
orthogonal to $e_{t,h}$, and OLS leaves these moments unused. An extended set of valid lag moments can therefore enhance asymptotic efficiency. We
treat \eqref{eq:vard_lp} as a linear GMM problem indexed by the integers $(p,q)$.

\vspace{0.5\baselineskip}

\noindent\textbf{Definition}. The $\operatorname{GMM}(p,q)$ estimator of $\beta_h$ (formally defined in \eqref{gmm_iv} below) is the optimally weighted GMM estimator of the regression
\eqref{eq:vard_lp} using $Y_{t,q}$ as the instruments for the regressors $Y_{t,p}$, where
$q\ge p\ge d$.

\vspace{0.5\baselineskip}

To formalize the estimator, define the sample moment function $ \overline  g_{h,p,q}(\beta_h,\eta_{h,p})
=\\
T^{-1}\sum_{t=q}^{T-h}
Y_{t,q}
\Bigl(
y_{1,t+h}
-
(\beta_h',\eta_{h,p}')Y_{t,p}
\Bigr)$. For an optimal weighting matrix estimate, let $\widehat e_{t,h}$ be a first-step consistent estimate of the projection
error in \eqref{eq:vard_lp}. We use $\widehat e_{t,h}
=
y_{1,t+h}
-
(\widehat\theta_h^{\operatorname{LALP}})'Y_{t,d+1},$ the OLS LALP
residual, formed by plugging in $\widehat\theta_h^{\operatorname{LALP}}$, the
vector of all coefficients in \eqref{eq:vard_lp} with $p=d+1$.\footnote{Other choices of $\widehat e_{t,h}$ are possible, including OLS
local-projection residuals without lag augmentation or GMM residuals, the
latter leading to an iterated GMM procedure. Any first-step choice, for which
the substituted residual does not affect the probability limit of the
feasible long-run variance estimator, leaves the first-order limit of the GMM
local projection estimator unchanged.} Setting $\widehat s_{t,q,h}=Y_{t,q}\widehat e_{t,h}$, we define the HAC (heteroskedasticity and autocorrelation consistent) score variance estimator
$\widehat{\Omega}_{q,h}
=
\sum_{\ell=-b_T}^{b_T}
k\!\left(\frac{\ell}{b_T}\right)
\widehat{\Gamma}_{q,h}(\ell),$
where $\widehat{\Gamma}_{q,h}(\ell)
=
T^{-1}\sum_{t=\ell+q}^{T-h}
\widehat s_{t,q,h}\widehat s_{t-\ell,q,h}'$
for $\ell\ge0$,
$\widehat{\Gamma}_{q,h}(-\ell)
=
\widehat{\Gamma}_{q,h}(\ell)'$, and $k(\cdot)$ is a bounded symmetric kernel, continuous at zero, with
$k(0)=1$, and $b_T$ is a truncation parameter. 

For $q>p$, the population analog of the score process $\widehat s_{t,q,h}$ (or the one after partialling-out (Section \ref{subsec:optimal_iv})) is generally serially correlated (even when the error $u_t$ is assumed to be independent); see Section \ref{subsec:lalp_ar1} for an illustration under the AR(1) model. Thus the HAC-type variance estimator is generally required.\footnote{For the OLS estimator of $\beta_h$, see \citet{montiel2021local} and \citet{xu2023local} for the possibility of avoiding the HAC variance estimator for valid inference.}

The feasible optimally weighted estimator is $\{\widehat{\beta}_h^{\mathrm{GMM}(p,q)},\widehat{\eta}_{h,p}^{\mathrm{GMM}(p,q)}\}
\\=
\arg\min_{\beta_h,\,\eta_{h,p}}
\overline g_{h,p,q}(\beta_h,\eta_{h,p})'
\widehat{\Omega}_{q,h}^{-1}
\overline g_{h,p,q}(\beta_h,\eta_{h,p}).
$
The block of
interest has the closed form
\begin{equation}
\label{gmm_iv}
\widehat{\beta}_h^{\mathrm{GMM}(p,q)}
=
J_p
\left(
\widehat{\Sigma}_{p,q}\widehat{\Omega}_{q,h}^{-1}\widehat{\Sigma}_{p,q}'
\right)^{-1}
\widehat{\Sigma}_{p,q}\widehat{\Omega}_{q,h}^{-1}\widehat m_{q,h},
\end{equation}
where $\widehat m_{q,h}=T^{-1}\sum_{t=q}^{T-h}Y_{t,q}y_{1,t+h}$ and
$\widehat\Sigma_{p,q}=T^{-1}\sum_{t=q}^{T-h}Y_{t,p}Y_{t,q}'$. 
 When $q=p$, GMM($p,q$) reduces to OLS with $p$ regressors. When $q>p$, the excluded lags generate overidentifying moments,
and $\widehat\Omega_{q,h}^{-1}$ enters the estimator itself, selecting the
efficient combination of IVs; the long-run covariance is then relevant for both
estimation and inference. The variance estimator for $\widehat{\beta}_h^{\mathrm{GMM}(p,q)}$ is defined as
\begin{equation}
\label{eq:Vpq_hat}
\widehat V_h^{\mathrm{GMM}(p,q)}
=
J_p
\left(
\widehat{\Sigma}_{p,q}\widehat{\Omega}_{q,h}^{-1}\widehat{\Sigma}_{p,q}'
\right)^{-1}
J_p'.
\end{equation}
Standard errors are the square roots of the diagonal entries of
$T^{-1}\widehat V_h^{\mathrm{GMM}(p,q)}$. The construction of the variance estimator appears to rely on stationarity. As we show in Section \ref{sec:uniform}, however, the resulting inference remains robust to (near-)unit roots.

\subsection{Exactly-identified optimal IV implementation}
\label{subsec:optimal_iv}

It is useful to have an equivalent form of \eqref{gmm_iv} which separates out the parameter vector of interest. It is obtained as an exactly-identified IV representation after Frisch--Waugh--Lovell residualization.

Write $Y_{t,p}=(y_t',w_{t,p}')'$, where $w_{t,p}=(y_{t-1}',\cdots,y_{t-p+1}')'$ collects the controls. Write the instrument vector $Y_{t,q}=(z_{t,p,q}',w_{t,p}')'$, where $z_{t,p,q}=(y_t',y_{t-p}',\cdots,y_{t-q+1}')'$.

Projecting $y_t$ and $z_{t,p,q}$ on $w_{t,p}$ over
$t=q,\cdots,T-h$, gives residuals $\widehat y_t^{\perp}$ and
$\widehat z_{t,p,q}^{\perp}$, respectively. Let
$\widehat \Sigma_{p,q}^{\perp}=T^{-1}\sum_{t=q}^{T-h}
\widehat y_t^{\perp}(\widehat z_{t,p,q}^{\perp})'$. Let
$\widehat\Omega_{p,q,h}^{\perp}$ be defined like $\widehat{\Omega}_{q,h}$ except that $\widehat z_{t,p,q}^{\perp}\widehat e_{t,h}$ is used in place of $Y_{t,q}\widehat e_{t,h}$. The optimal instrument is the $K$-vector
$\widehat z_{t,p,q}^*
=
\widehat \Sigma_{p,q}^{\perp}
\bigl(\widehat\Omega_{p,q,h}^{\perp}\bigr)^{-1}
\widehat z_{t,p,q}^{\perp}.$
Define
$
\widehat\beta_h^{\operatorname{OIV}(p,q)}
=
\Biggl(
\sum_{t=q}^{T-h}
\widehat z_{t,p,q}^*\,\widehat y_t^{\perp\prime}
\Biggr)^{-1}
\sum_{t=q}^{T-h}
\widehat z_{t,p,q}^*\,y_{1,t+h},
$
obtained by solving the \emph{exactly identified} IV regression, and the variance matrix estimator
$\widehat V_h^{\operatorname{OIV}(p,q)}
=
\Bigl(
\widehat \Sigma_{p,q}^{\perp}
\bigl(\widehat\Omega_{p,q,h}^{\perp}\bigr)^{-1}
\widehat \Sigma_{p,q}^{\perp\prime}
\Bigr)^{-1}.$

We can also define the population quantities under the stationarity assumption of $y_t$, which are useful in the efficiency gain calculation (Section \ref{sec:eff}). Under stationarity, the asymptotic variance
of $\widehat{\beta}_h^{\mathrm{GMM}(p,q)}$ is $V_h^{\mathrm{GMM}(p,q)}
=
J_p
\left(
\Sigma_{p,q}\Omega_{q,h}^{-1}\Sigma_{p,q}'
\right)^{-1}
J_p',
$ where $\Omega_{q,h}
:=
\sum_{\ell=-\infty}^{\infty}
\mathbb E
\left[
Y_{t,q}e_{t,h}e_{t-\ell,h}Y_{t-\ell,q}'
\right]$, $\Sigma_{p,q}:=\mathbb E[Y_{t,p}Y_{t,q}']$ and
$Q_{p,q,h}:=\Sigma_{p,q}\Omega_{q,h}^{-1}\Sigma_{p,q}'$. 

To separate out the block of interest, define
$y_t^{\perp}:=y_t-P(y_t\mid w_{t,p})\in\mathbb R^{K}$, the
residualized regressor of interest, and
$z_{t,p,q}^{\perp}:=z_{t,p,q}-P(z_{t,p,q}\mid w_{t,p})\in\mathbb R^{(q-p+1)K}$, the residualized instruments. Both are orthogonal to $w_{t,p}$. Define $\Sigma_{p,q}^{\perp}
:=
\mathbb E\!\left[y_t^{\perp}(z_{t,p,q}^{\perp})'\right],$ and
$
\Omega_{p,q,h}^{\perp}
:=
\sum_{\ell=-\infty}^{\infty}
\mathbb E\!\left[
z_{t,p,q}^{\perp}e_{t,h}\,e_{t-\ell,h}(z_{t-\ell,p,q}^{\perp})'
\right]$. 
Write the $K\times K$ matrix
\begin{equation}
\label{eq:Vpq_pop_2}
V_h^{\operatorname{OIV}(p,q)}
=
\left(
\Sigma_{p,q}^{\perp}
\bigl(\Omega_{p,q,h}^{\perp}\bigr)^{-1}
\Sigma_{p,q}^{\perp\prime}
\right)^{-1}.
\end{equation}

\begin{proposition}[Optimal-IV representation of GMM local projections]
\label{prop:optimal_iv_gmm}
(i) Suppose $\widehat{\Omega}_{q,h}$ and $\widehat\Omega_{p,q,h}^{\perp}$ are nonsingular and are built from the same residuals, kernel, and
bandwidth. Then $\widehat\beta_h^{\mathrm{GMM}(p,q)}
=\widehat\beta_h^{\mathrm{OIV}(p,q)}$ and 
$\widehat V_h^{\mathrm{GMM}(p,q)}
=
\widehat V_h^{\mathrm{OIV}(p,q)}.$ (ii) Under stationarity, $V_h^{\operatorname{GMM}(p,q)}
=
V_h^{\operatorname{OIV}(p,q)}$.
\end{proposition}

Proposition~\ref{prop:optimal_iv_gmm} is proved in
Appendix~\ref{app:proof_optimal_iv}. Part (i) is a finite-sample identity.\footnote{It extends the finite-sample equivalence result on the residualized-IV representation of \citet{basu2024frisch} to a setting
where the weighting is governed by serial dependence in the moment process.} We will implement the GMM in its optimal IV form; the automatic truncation selection in the HAC variance estimator directly targets the (effective) residualized score for the parameter of interest.

\textbf{\textsc{Lag orders}}. What values of $p$ and $q$ should be used in practice? Our recommendations are $p=d+1$ and $q=d+2$, after the researcher has chosen the model order $d$, typically through some data-dependent information criterion. The resulting GMM$(d+1,d+2)$ is therefore an overidentified version of the LALP of \citet{montiel2021local} with a single set of overidentifying restrictions. This recommendation is based on the fact that the lag-augmented GMM estimator shares the robustness properties of the LALP estimator which have been recently established \citep{montiel2021local,olea2024double}, as we will show in Sections \ref{sec:uniform} and \ref{sec:local} in our setting. Although more overidentifying conditions can theoretically produce larger efficiency gains, too many of them can also lead to poor performance in finite samples, especially when the optimal weighting matrix is estimated. In the next section, we will illustrate in a simple model that even one single set of overidentification conditions can in fact generate a large portion of the  efficiency gain that is generated by using all possible lags as instruments.

\subsection{Non-redundancy of the overidentifying restrictions}
\label{subsec:Non-Redundancy}

While the GMM never has a larger asymptotic variance than OLS with the same set of regressors, it is natural to ask when the overidentifying restrictions provide a \textit{strict} efficiency gain. In Appendix \ref{app:gmm_vs_ols}, we provide an analysis of a general time series regression with serially correlated IVs (which accommodates the LP regression \eqref{eq:vard_lp}) to address this question, connecting to the literature on the redundancy of the moment conditions for the GMM estimator (\cite{breusch1999redundancy}, \cite{hall2007information}).

The analysis shows that the redundancy occurs only when two population projection coefficient matrices are equal. The first projection is the ordinary projection of the excluded IVs (i.e. $z_{t,p,q}^{\perp}$, in the context of the LP regression \eqref{eq:vard_lp}, using the notation in Section \ref{subsec:optimal_iv}) on the regressors of interest ($y_t^{\perp}$), after the controls are partialed out. The second projection is the \textit{long-run} projection of the orthogonal moment process ($z_{t,p,q}^{\perp}e_{t,h}$) on the regression score process ($y_t^{\perp}e_{t,h}$), after the controls are partialed out. The  two projection coefficient matrices are generally different (thus the added lags carry first-order information), and the main driving factors of the non-equivalence are the serial correlation of the residualized moment process (e.g. the LP with $h>1$), and the conditional heteroskedasticity of the innovation $u_t$. In the case of LALP as we study in the next subsection, the non-redundancy condition for the extra IVs reduces to a particularly simple one.

\subsection{Lag augmentation}
\label{subsec:lag_aug}
We now focus on the LP with the lag-augmented regressor set, i.e. $p=d+1$. We address two questions: the source of the efficiency gain, and the serial correlation in the GMM moment process (why the HAC inference is generally needed).

Assume stationarity in this subsection for illustrative purposes. With lag augmentation we have $y_t^{\perp}=u_t$. Also,
$z_{t,p,q}^{\perp}=(u_t',v_{t-1}')'$, where $v_{t-1}=
\sum_{j=0}^{\infty}C_j u_{t-1-j}$ is the projection residual of $y_{t-d-1}$ on $y_{t-1},\cdots,y_{t-d}$, with the coefficients $C_j$ being square summable in their norms.\footnote{Note that $P(y_{t-d-1}\mid y_{t-1},\cdots,y_{t-d})$ is a reverse projection; it projects on future variables, not on the past, as usual.} Thus $\mathbb{E}[y_t^\perp z_{t,p,q}^{\perp\prime}]=(\Sigma_u,0_K)$.

With lag augmentation, the off-diagonal long-run term in $\operatorname{LRV}\!\left(z_{t,p,q}^\perp e_{t,h}\right)$ is what drives the efficiency gains from overidentification. If we used a diagonal weighting matrix (like the identity matrix),  GMM would collapse to LALP; the formula \eqref{eq:Vpq_pop_2} would show that they had the same asymptotic variance. The efficiency gain of the GMM lives entirely in the long run covariance between $u_t e_{t,h}$ and $v_{t-1} e_{t,h}$.

\cite{montiel2021local} show that with the lag-augmented regressor set, the OLS (residualized) LP regression score process for the impulse response can be serially uncorrelated, under appropriate assumptions on the higher-order dependence in innovations (i.e. mean independence of its own past and future; their Assumption 1). Thus the standard Eicker-Huber-White HC inference remains asymptotically valid, without the need to correct for the serial correlation in the score as done by the HAC inference.\footnote{See \citet{xu2023local} for the justification of the Eicker-Huber-White HC inference under the weaker assumption of MDS innovations with a zero fourth-order cumulant condition.} By contrast, the GMM moment process $z_{t,p,q}^\perp e_{t,h}$ is generally serially correlated (even with the mean-independent or i.i.d. innovations), because of the reverse-regression residual $v_{t-1}$. We illustrate these points in the simple model in the next section.

\section{Illustration with AR(1) Data}
\label{sec:eff}

This section shows how overidentification can improve the asymptotic efficiency of the
lag-augmented local projection through a simple example.

Consider the simplest dynamic model for the scalar outcome, AR(1),
\begin{align}
\label{eq:ar1}
y_t = \alpha\, y_{t-1} + u_t, \ \ \ \ \ u_t \overset{\text{i.i.d}}{\sim} (0,\sigma^2),
\end{align}
with $|\alpha|\le 1$, $\mathbb{E}u_t=0$, and $\mathbb{E}u_t^2=
\sigma^2>0$. The GMM method proposed in this paper applies to both the standard LP (using $y_t$ as the only regressor) and the lag-augmented LP (LALP) (using $y_t$ and $y_{t-1}$ as the regressors). We focus on our recommended GMM-LALP regression here, and the corresponding results for the GMM-LP under the AR(1) model are provided in Appendix \ref{app:gmm_lp}.

For $h\ge1$, under the AR(1) model, the LALP regression is $y_{t+h}
=
\beta_h y_{t}+\eta_{h}y_{t-1}+e_{t,h}$, where
$\beta_h=\alpha^h$, $\eta_{h}=0$, and $e_{t,h}=\sum_{i=0}^{h-1}\alpha^i u_{t+h-i}$.

\subsection{The first overidentified estimator}
\label{subsec:lalp_ar1}

We first consider the simplest overidentified estimator, $\operatorname{GMM}(2,3)$. As the benchmark, the LALP estimator of \citet{montiel2021local} partials out $y_{t-1}$ and reduces the regressor effectively to $y_t-\alpha y_{t-1}=u_t$, so the
identifying variation is the innovation in $y_t$. LALP uses this single moment, and has the asymptotic variance
$V_h^{\mathrm{LALP}}
=
\sum_{j=0}^{h-1}\alpha^{2j}.$
With one overidentifying restriction, $\operatorname{GMM}(2,3)$ uses $y_{t-2}$ as the single excluded instrument, besides the regressors $y_{t}$ and $y_{t-1}$.
Specializing the analysis of Section \ref{subsec:lag_aug} to the AR(1) model, and using the notation there, we can show (Appendix~\ref{subsec:hac_gmmlalp}) that if $|\alpha|<1$, the residualized excluded instrument is  
\begin{align}
\label{eq:reverse_residual}
v_{t-1}:=y_{t-2}-P(y_{t-2}\mid y_{t-1})
=
-\alpha u_{t-1}
+
(1-\alpha^2)
\sum_{j=0}^{\infty}\alpha^j u_{t-2-j}.
\end{align}

The residualized regressor of interest is $y_t^\perp=y_t-P(y_t\mid y_{t-1})
=u_t$. So the residualized instrument vector is $z_t^\perp=(u_t,v_{t-1})'$. Thus $\mathbb{E}[y_t^\perp z_t^{\perp\prime}]=\sigma^2(1,0)$. Then
\begin{align}
\label{eq:resid_lrv_gmmlalp}
\operatorname{LRV}\!\left(z_t^\perp e_{t,h}\right)
=
\sigma^4
\begin{pmatrix}
\sum_{j=0}^{h-1}\alpha^{2j} & -(h-1)\alpha^{2h-2}\\
-(h-1)\alpha^{2h-2} & \sum_{j=0}^{h-1}\alpha^{2j}
\end{pmatrix},
\end{align}
which is non-diagonal whenever $h>1$ and $\alpha\ne0$. 

Equation \eqref{eq:resid_lrv_gmmlalp} also shows the serial correlation in $z_t^\perp e_{t,h}$. The
contemporaneous covariance matrix of $z_t^\perp e_{t,h}$, $\operatorname{Var}\left(z_t^\perp e_{t,h}\right)=\sigma^4\left(\sum_{j=0}^{h-1}\alpha^{2j}\right)I_2$, is diagonal, and its difference from the LRV matrix arises entirely from the off-diagonal elements.

The residualized
GMM variance formula $\left\{
\mathbb{E}(y_t^\perp z_t^{\perp\prime})
\bigl[\operatorname{LRV}\!\left(z_t^\perp e_{t,h}\right)]^{-1}
\mathbb{E}(y_t^\perp z_t^{\perp\prime})'
\right\}^{-1}$ then gives
\begin{align}
\label{eq:vgmm23}
V_h^{\mathrm{GMM}(2,3)}
&=
\sum_{j=0}^{h-1}\alpha^{2j}
-
\frac{(h-1)^2\alpha^{4h-4}}{\sum_{j=0}^{h-1}\alpha^{2j}}.
\end{align}
The asymptotic variance formula \eqref{eq:vgmm23} in fact extends to the unit root case $\alpha=1$, due to the uniform limit theory we establish in Section \ref{sec:uniform}. Comparison with $V_h^{\mathrm{LALP}}
=
\sum_{j=0}^{h-1}\alpha^{2j}$ shows $V_h^{\mathrm{GMM}(2,3)}\leq V_h^{\mathrm{LALP}}$,
and the strict inequality holds for $h>1$, $\alpha\ne0$.

\begin{figure}[t]
\hspace*{-0.5cm} 
\makebox[\linewidth][c]{
\includegraphics[width=1.4\textwidth]{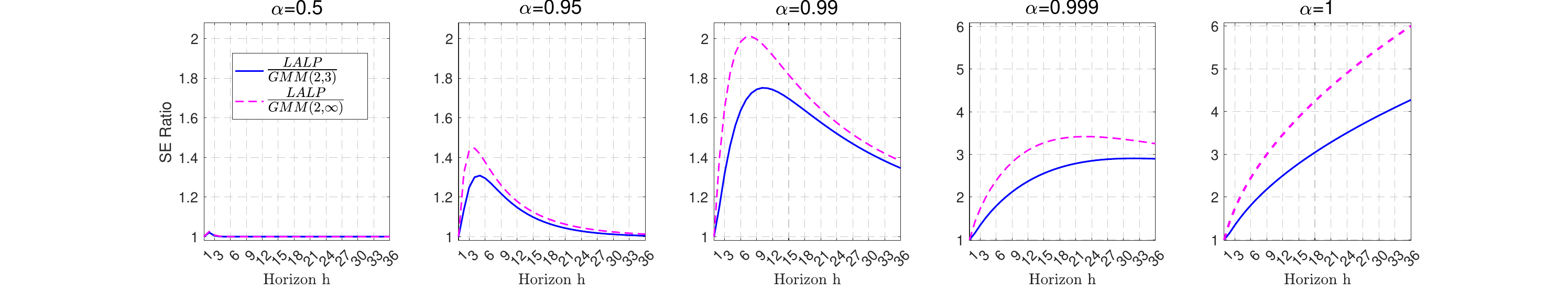}
}
    \caption{Asymptotic standard-error ratio between LALP and GMM$(2,3)$ (solid lines), and between LALP and GMM$(2,\infty)$ (dashed lines).
[AR(1) DGP as in \eqref{eq:ar1}]}
\label{fig:gmm_lalp_ratio}
\begin{minipage}{0.95\linewidth}
\footnotesize
\emph{Note:} The figure plots $\sqrt{V_h^{\mathrm{GMM}(2,2)}/V_h^{\mathrm{GMM}(2,q)}}$
for $q\in\{3,\infty\}$, where $q$ is the number of IVs. $\alpha$ is the slope of the AR(1) DGP. The just-identified $\operatorname{GMM}(2,2)$
coincides with the LALP estimator of \citet{montiel2021local}. We use the formulas in Table~\ref{tab:unit_root_variance_ordering} for the three asymptotic variances for $q\in\{2,3,\infty\}$. 
\end{minipage}
\end{figure}

Figure \ref{fig:gmm_lalp_ratio} visualizes the efficiency comparison between LALP and GMM$(2,3)$ by plotting the ratio of the square roots of their asymptotic variances. The figure shows that the efficiency gain of GMM depends on both the persistence level $\alpha$ and the horizon $h$. When the data are only weakly dependent (like $\alpha=0.5$), there is little efficiency gain. 
The biggest efficiency gain is achieved at the unit root $\alpha=1$, where 
$V_h^{\mathrm{LALP}}=h$, while 
$V_h^{\mathrm{GMM}(2,3)}=(2h-1)/h<2$: a single extra lag as the excluded IV caps the variance below
$2$ at every horizon, whereas LALP grows without bound. 

As indicated in Section \ref{subsec:lag_aug}, the efficiency gain is driven by the long-run score covariance, which specializes under the AR(1) model \eqref{eq:ar1} as $\operatorname{LRCov}\!\left(u_t e_{t,h},\,v_{t-1}e_{t,h}\right)
=
-(h-1)\sigma^4\alpha^{2h-2}$.
Exact redundancy arises only at $h=1$ (no overlap) or
$\alpha=0$ (no propagation). For $|\alpha|$ sufficiently far from unity, $(h-1)^2\alpha^{4h-4}$ is small for moderate and large $h$, leading to the asymptotic variance of GMM close to that of OLS.

\subsection{More instruments}
\label{sec:theory_ar1}

The asymptotic efficiency of the GMM$(2,3)$ estimator can be further improved if we use further lags to overidentify the LALP. For general $q$, the asymptotic variance of GMM$(2,q)$ under the DGP \eqref{eq:ar1} is given in the following proposition. We also consider the particularly interesting case in which $q$ diverges to infinity, corresponding to the efficiency gain that can be achieved by exploiting information from the infinite past.

\begin{proposition}
\label{prop:ar1_ordering}
Let $\{y_t\}$ follow \eqref{eq:ar1}. For $h\ge1$, write
$\gamma_{h,\ell}=\alpha^\ell\sum_{j=0}^{h-\ell-1}\alpha^{2j}$ for $0\le\ell\le h-1$
and $\gamma_{h,\ell}=0$ otherwise. Then:
\begin{enumerate}[label=(\roman*)]
\item For $q\ge3$,
$V_h^{\mathrm{GMM}(2,q)}=\gamma_{h,0}-r_{h,q}'A_{h,q}^{-1}r_{h,q},$
where the $(q-2)$-vector $r_{h,q}$ has entries
$(r_{h,q})_k=-(h-k)\alpha^{2h-k-1}$ for $1\le k\le\min\{h-1,q-2\}$ and zero
otherwise, and $A_{h,q}$ is the $(q-2)\times(q-2)$ Toeplitz matrix with
$(A_{h,q})_{k\ell}=\gamma_{h,|k-\ell|}$. 
\item $\lim_{q\to\infty}\displaystyle V_h^{\mathrm{GMM}(2,q)}=\gamma_{h,0}-\alpha^{2h}\gamma_{h-1,0}
=\sum_{j=0}^{h-1}\alpha^{2j}-\sum_{j=h}^{2h-2}\alpha^{2j}$.
\end{enumerate}
\end{proposition}

The proof is given in
Appendix~\ref{subsec:hac_gmmlalp}. 
In Figure \ref{fig:gmm_lalp_ratio}, we plot the ratio of standard deviations between LALP and GMM$(2,\infty)$, together with the ratio of standard deviations between LALP and GMM$(2,3)$. The relative efficiency for other values of $q$ lies between these two curves. At the unit root $\alpha=1$, the IV estimator optimally weighting the infinite past brings the asymptotic variance down to 1, across all horizons $h$, compared with the previous result for GMM(2,3), $2-h^{-1}$. The results also show that a substantial efficiency gain can be obtained with \textit{only one overidentifying restriction}. Additional results (not reported) indicate that unless $q$ is much larger than 3, the further efficiency gain over GMM$(2,3)$ is quite minor (e.g. the relative efficiency gain at $q=6$ is almost indistinguishable from that at $q=3$).

\begin{table}[t]
\centering
\caption{Asymptotic variances of various estimators of the lag augmentation LP regression. [AR(1) DGP as in \eqref{eq:ar1}]}
\label{tab:unit_root_variance_ordering}
\renewcommand{\arraystretch}{1}
\setlength{\tabcolsep}{14pt}
\begin{tabular}{l|l|l}
\hline\hline
Estimator & Variance for $|\alpha|\le1$ & Variance at $\alpha=1$ \\
\hline
&& \\
$\operatorname{GMM}(2,\infty)$ & 
$\sum_{j=0}^{h-1}\alpha^{2j}-\sum_{j=h}^{2h-2}\alpha^{2j}$ & $1$ \\
$\operatorname{GMM}(2,3)$ 
& 
$\sum_{j=0}^{h-1}\alpha^{2j}
-
 (h-1)^2\alpha^{4h-4}/(\sum_{j=0}^{h-1}\alpha^{2j})$& $2-1/h$ \\
LALP (i.e. GMM(2,2)) & 
$\sum_{j=0}^{h-1}\alpha^{2j}$& $h$ \\
LA-AR & 
$h^2\alpha^{2(h-1)}$ & $h^2$ \\
\hline\hline
\end{tabular}
\par\vspace{0.8em}
\begin{minipage}{0.9\linewidth}
\footnotesize
\emph{Note:} Asymptotic variances of ordinary LALP, the first overidentified
GMM-LALP estimator, its infinite-instrument limit, and the recursive lag-augmented
AR estimator. The infinite-instrument case is the limit of the optimally
weighted $\operatorname{GMM}(2,q)$ as $q\to\infty$.
\end{minipage}
\end{table}

For convenience, we summarize the comparison in Table \ref{tab:unit_root_variance_ordering}, with particular attention to the unit root case $\alpha=1$.
In Table \ref{tab:unit_root_variance_ordering}, we also include the iterated lag-augmented AR (LA-AR) estimator of \citet{inoue2020uniform}, which is analyzed in \cite{montiel2021local} (p.1817).
Compared with LA-AR, the efficiency advantage of the GMM mainly lies in the high-persistence regime. For example, for $\alpha$ between 0.93 and 1, and for horizons up to 36, the GMM(2,3) always has smaller asymptotic variance than LA-AR. At the unit root, the variance of LA-AR is $h^2$, growing quadratically with the horizon, while that of GMM(2,3) remains uniformly bounded by 2. If the data are only weakly dependent, the LA-AR typically has the smaller variance. Nevertheless, as we show in Section \ref{sec:local}, extending the analysis of \citet{olea2024double}, the inference based on the LA-AR estimator is more fragile to model misspecification than our GMM estimator.

\section{Uniform inference}
\label{sec:uniform}

This section establishes that the Wald statistic based on a general GMM($p,q$) LP estimator is asymptotically pivotal \emph{uniformly} over the persistence of the data under the specified conditions. The following assumption is also used in \cite{montiel2021local} (definition 1, p.1804). Following their framework, we assume that the VAR lag order $d$ is fixed. 

\begin{assumption}[VAR parameter space]
\label{ass:persistent} Let $A(L)=I_K-A_1L-\cdots-A_dL^d$ denote the VAR polynomial
in \eqref{eq:vard}. Suppose that
$A(L)=B(L)(I_K-DL)$, where
$B(L)=I_K-B_1L-\cdots-B_{d-1}L^{d-1}$ and
$D=\operatorname{diag}(\rho_1,\cdots,\rho_K)$.
For $k=1,\cdots,K$, $\rho_k\in[-1,1]$. The companion matrix
$\mathbf B$ of $B(L)$ satisfies
$\|\mathbf B^\ell\|_2\le C(1-c)^\ell$ for every $\ell\ge1$,
with $C<\infty$ and $c\in(0,1)$.
\end{assumption}

Let $\xi_t=(I_K-DL)y_t$. Then
\begin{equation}
y_t
=
Dy_{t-1}+\xi_t,
\label{eq:persistent-decomposition}
\end{equation}
or, equivalently, $y_{k,t}=\rho_ky_{k,t-1}+\xi_{k,t}$ for each $k$. Thus each variable carries its own potentially persistent root, while $B(L)$ collects the transitory dynamics.

For each $k$, define $\tau_{k,T}=
\sum_{\ell=0}^{T-1}\rho_k^{2\ell}$ and $\Lambda_T=\operatorname{diag}(\tau_{1,T},\cdots,\tau_{K,T})$. In what follows the parameters may depend on $T$. We suppress the triangular-array index for expositional simplicity, and all constants below are uniform over the parameter space and over $T$. Letting the roots vary with $T$ embeds fixed stationary dynamics, shrinking-to-unity sequences, and exact unit roots of either sign in a single parameter space.\footnote{The diagonal form of $D$ is for expository purposes; the same argument applies after a fixed rotation separating the persistent directions from the transitory block.}

\begin{assumption}[Innovations and transitory process]
\label{ass:stable-score}
(i) $\{u_t,\mathcal F_t\}$ is a martingale difference sequence (MDS), where $\mathcal F_t=\sigma(y_0,u_s:s\le t)$;
(ii) $\{u_t\}$ is stationary and $\alpha$-mixing with coefficients $\alpha(j)$ satisfying
$\sum_{j\ge1} j^{2/(r-2)}\alpha(j)^{1-2/r}<\infty$
for some $r>2$, and
$\mathbb E\|u_t\|^{4r+\delta}\le C$
for some $\delta>0$;
(iii) The initial condition satisfies
$\mathbb E\|\Lambda_T^{-1/2}y_0\|^{4r+\delta}\le C$ uniformly over the parameter space.
\end{assumption}

Write $\mathcal P$ for the set of parameters admitted by
Assumption~\ref{ass:persistent}, a generic
element of which is denoted $(D,B)$. Assumption~\ref{ass:stable-score} (i)-(ii) is imposed only on
the primitive innovation process $\{u_t\}$. The moment and dependence
properties of the transitory process $\xi_t$ are instead derived from
its stable-filter representation. Under correct specification (as in this section),
$\xi_t=B(L)^{-1}u_t$; under local misspecification (in the next section), the corresponding
representation is given in \eqref{eq:aux-local-xi-representation}.

\subsection{A rotation of the lags}
\label{subsec:uniform-rotation}

The following rotation separates quasi-differences from levels of the time series, which is especially useful in the persistent regime. We assume $p\ge2$, the rotation thus places $\beta_h$ on a quasi-difference. We emphasize that the rotation is only used in the proof, and does not change any implementation detailed in Section \ref{sec:implementation}.

Let $\Delta_p$ be block upper-bidiagonal such that
\begin{align}
\underset{pK\times pK}{\Delta_p}
&=
\begin{pmatrix}
I_K & -D & 0 & \cdots & 0\\
0 & I_K & -D & \ddots & \vdots\\
\vdots & \ddots & \ddots & \ddots & 0\\
0 & \cdots & 0 & I_K & -D\\
0 & \cdots & \cdots & 0 & I_K
\end{pmatrix}, \quad \underset{pK\times pK}{\Delta_p^{-1}} =
\begin{pmatrix}
I_K & D & D^2 & \cdots & D^{p-1} \\
0 & I_K & D & \ddots & \vdots \\
\vdots & \ddots & \ddots & \ddots & D^2 \\
0 & \cdots & 0 & I_K & D \\
0 & \cdots & \cdots & 0 & I_K
\end{pmatrix}
,
\label{eq:delta-p}
\end{align}
where the inverse follows immediately.
Let $\Upsilon_p=\operatorname{diag}\left(I_{(p-1)K},\Lambda_T^{-1/2}\right)$, and set $H_p=\Upsilon_p\Delta_p$. Then
\begin{align}
H_pY_{t,p}
=
\left(
\xi_t',\,
\xi_{t-1}',\,
\cdots,\,
\xi_{t-p+2}',\,
(\Lambda_T^{-1/2}y_{t-p+1})'
\right)',
\label{eq:rotated-vector}
\end{align}
which stacks quasi-differences above scaled levels whose second moments are uniformly bounded. Their lower-tail nondegeneracy is imposed in Assumption~\ref{ass:rotated-nonsingularity} below.

Rotate the sample objects by $\widecheck\Sigma_{p,q}=H_p\widehat\Sigma_{p,q}H_q'$, $\widecheck\Omega_{q,h}=H_q\widehat\Omega_{q,h}H_q'$, and $\widecheck m_{q,h}=H_q\widehat m_{q,h}$. For $p\ge2$, we have $J_pH_p'
=
J_p,
H_pJ_p'
=
J_p',$
$J_p(H_p')^{-1}=J_p$, and $H_p^{-1}J_p'=J_p'$.
Hence the estimator \eqref{gmm_iv} is numerically identical in the
rotated and original coordinates,
$\widehat\beta_h^{\operatorname{GMM}(p,q)}
=
J_p
\left(
\widecheck\Sigma_{p,q}\widecheck\Omega_{q,h}^{-1}\widecheck\Sigma_{p,q}'
\right)^{-1}
\widecheck\Sigma_{p,q}\widecheck\Omega_{q,h}^{-1}\widecheck m_{q,h},$
and so is its variance \eqref{eq:Vpq_hat}, 
$\widehat V_h^{\operatorname{GMM}(p,q)}=J_p(\widecheck\Sigma_{p,q}\widecheck\Omega_{q,h}^{-1}\widecheck\Sigma_{p,q}')^{-1}J_p'.$ 
The rotation therefore reorganizes only the nuisance directions.

\subsection{Wald statistic}
\label{subsec:uniform-main}
For the uniformity result, we impose the following high-level condition. A similar condition is also used in \cite{montiel2021local} and \cite{xu2023local}.

\begin{assumption}[Rotated nonsingularity]
\label{ass:rotated-nonsingularity}
For every $\varepsilon>0$ there exists $\underline c>0$ such that
\begin{align}
\liminf_{T\to\infty}
\inf_{(D,B)\in\mathcal P}
\mathbb P_{D,B}
\left\{
\lambda_{\min}(\widecheck\Sigma_{p,p})\ge\underline c,\;
\lambda_{\min}(\widecheck\Omega_{q,h})\ge\underline c
\right\}
\ge1-\varepsilon.
\label{eq:rotated-nonsingularity}
\end{align}
\end{assumption}

The stochastic formulation in Assumption \ref{ass:rotated-nonsingularity} accommodates every persistence regime. When $\max_k|\rho_k|\le1-c$, it reduces to the usual requirement that the stationary design and long-run covariance limits be positive definite. When some $|\rho_k|\to1$, the corresponding scaled-level block may converge to a random quadratic functional. Assumption~\ref{ass:rotated-nonsingularity} requires the lower tails of this block and of the rotated HAC matrix to remain uniformly away from singularity in probability. Appendix~\ref{app:local-rotated-nonsingularity} provides an illustration. Throughout the paper, the truncation parameter for the HAC estimator is such that $b_T\to\infty$ with $b_T=o(T^{1/2})$.

\begin{theorem}[Uniform GMM inference]
\label{thm:persistence-uniform-gmm}
Let $G$ be a known $g\times K$ matrix of rank $g\le K$. Let $q\ge p\ge\max(d,2)$. Define the Wald statistic
$\\ W_{T,h}
=
T
\bigl(G\widehat\beta_h^{\operatorname{GMM}(p,q)}-G\beta_h\bigr)'
\bigl(
G\widehat V_h^{\operatorname{GMM}(p,q)}G'
\bigr)^{-1}
\bigl(G\widehat\beta_h^{\operatorname{GMM}(p,q)}-G\beta_h\bigr),$
where $\widehat V_h^{\operatorname{GMM}(p,q)}$ is defined in Equation~\eqref{eq:Vpq_hat}. Then under Assumptions~\ref{ass:persistent}--\ref{ass:rotated-nonsingularity},
\begin{align}
\limsup_{T\to\infty}
\;
\sup_{(D,B)\in\mathcal P}
\;
\sup_{x\in\mathbb R}
\left|
\mathbb P_{D,B}(W_{T,h}\le x)
-
F_{\chi_g^2}(x)
\right|
=
0,
\label{eq:uniform-wald-convergence}
\end{align}
where $F_{\chi_g^2}(x)$ is the cumulative distribution function of the chi-square variate with $g$ degrees of freedom.
\end{theorem}

Confidence sets that invert $W_{T,h}$ at the $\chi^2_{g,1-\alpha}$ critical value therefore have asymptotic coverage $1-\alpha$ uniformly over $\mathcal P$. In practice, inference can be improved by the bootstrap, which we implement in our numerical studies in Section \ref{sec:mc} and Section \ref{sec:application}.

The proof of Theorem~\ref{thm:persistence-uniform-gmm}, given in Appendix \ref{sec:uniform-theorem-proof}, splits the parameter space in two. When $\max_k|\rho_k|\le1-c$ the rotated system is stationary and a triangular-array GMM argument applies. When $\max_k|\rho_k|\to1$ the cross blocks between quasi-differences and levels vanish once multiplied by $\Lambda_T^{-1/2}$, so the scaled levels survive only in the nuisance block on the diagonal and affect neither the estimator of $\beta_h$ nor its selected variance. This is the mechanism behind the observation of \citet{sims1990inference} that coefficients on stationary linear combinations of the regressors remain asymptotically normal in the presence of unit roots; the rotation shows that, once $p\ge2$, $\beta_h$ is such a coefficient.

Our proof strategy differs from that used for just-identified local projections. \citet{montiel2021local} partial the lagged controls out of $y_{1,t+h}$ and of $y_t$, reducing inference on $\beta_h$ to a regression on a residualized regressor. Our argument instead follows the classical VAR literature: estimate the whole projection coefficient jointly, rotate so that each block is a quasi-difference or a scaled level, and read the estimator of $\beta_h$ off the first block. An advantage of this method is that we only require $p \ge max(d,2)$, while the residualization-based proof typically needs $p \ge d+1$. Additional lags are therefore allowed but not required for persistence robustness.

A simplifying assumption we make is that $h$ is a fixed horizon. \cite{montiel2021local} and \cite{xu2023local} allow $h$ to grow with the sample size, but both rely on the residualization-based arguments. Moreover, the inference procedure of \cite{montiel2021local} relies on the two-sided MDS assumption which we dispense with. Nevertheless, we find in simulations that the GMM-based LP inference is quite robust over a wide range of horizons.

\section{GMM-LALP under local misspecification}
\label{sec:local}

A finite-order VAR is rarely the true data-generating process. Temporal
aggregation, marginalization of omitted variables, and the equilibrium
dynamics of structural models generally induce moving-average components,
so the reduced form is generically VARMA.\footnote{See
\cite{lutkepohl2005,ravenna2007,fernandez2007}.} This section establishes
the validity of GMM-LALP under local VARMA departures in the spirit of
\citet{olea2024double}, while allowing unit
or near-unit roots in the data.

\subsection{A locally misspecified VAR}

We extend the model \eqref{eq:vard} to the following specification
\begin{equation}
\label{localmis}
y_t
=
A_1y_{t-1}+\cdots+A_dy_{t-d}+\widetilde u_t,
\end{equation}
where the VAR error $\widetilde u_t$ is not exactly a white noise process. Instead, $\widetilde u_t$ is assumed to follow the local VARMA($m,n$) law
\begin{align}
\label{eq:local_arma_error}
\prod_{j=1}^{m}\bigl(I_K-\Phi_{j,T}L\bigr)\widetilde u_t
=
\prod_{k=1}^{n}\bigl(I_K+\Theta_{k,T}L\bigr)u_t,
\end{align}
where the lag orders $m$ and $n$ are fixed but unknown. To regulate the degree of misspecification, assume
$\Phi_{j,T}=T^{-\zeta_{AR}}\Phi_j$ and
$\Theta_{k,T}=T^{-\zeta_{MA}}\Theta_k$, where
$\zeta_{AR}>0$, $\zeta_{MA}>0$, and
$\Phi_j$ and $\Theta_k$ are fixed matrices. Write $\bar\zeta=\min\{\zeta_{AR},\zeta_{MA}\}$.

By \eqref{eq:local_arma_error}, we can write
\begin{align}
\label{eq:local-filter-expansion}
\left\{
\prod_{j=1}^{m}
\bigl(I_K-\Phi_{j,T}L\bigr)
\right\}^{-1}
\prod_{k=1}^{n}
\bigl(I_K+\Theta_{k,T}L\bigr)
=
I_K+\sum_{j\geq1}\Pi_{j,T}L^j.
\end{align}
Each
$\Pi_{j,T}$ is a finite sum of products of exactly $j$ local matrices
and hence $\|\Pi_{j,T}\|=O(T^{-j\bar\zeta})$. Thus, for every $r\geq1$,
$\sum_{j\geq r}\|\Pi_{j,T}\|=O(T^{-r\bar\zeta})$; see Appendix
\ref{subsec:local-aux-lemmas}
(Lemma~\ref{lem:local-filter-approximation}).

\subsection{Lag augmentation and robustness}
\label{subsec:local-asym}

The following proposition gives an exact population representation for the local projection regression under the models \eqref{localmis} and
\eqref{eq:local_arma_error}, using $s$ additional lags as regressors.

\begin{proposition}[Exact lag-augmented representation]
\label{prop:local-lalp-representation}
Let $p=d+s$, for $s\geq0$. Under \eqref{localmis} and
\eqref{eq:local_arma_error}, if $s\geq1$,
\begin{equation}
\label{eq:lalp_local}
y_{1,t+h}
=
(\beta_{h,T}',\eta_{h,p,T}')Y_{t,p}
+
e_{t,h,T},
\end{equation}
for some $\eta_{h,p,T}$, where
\begin{equation}
\label{eq:local_beta_T}
\beta_{h,T}
=
\beta_h
+
\sum_{j=1}^{h}
\Pi_{j,T}'\beta_{h-j}.
\end{equation}
Moreover,
\begin{align}
\label{eq:local-error-decomposition}
e_{t,h,T}
&=
\nu_1'
\sum_{\ell=0}^{h-1}
\Psi_\ell u_{t+h-\ell}
+
\nu_1'
\sum_{j=1}^{\min\{s,h-1\}}
\sum_{\ell=0}^{h-j-1}
\Psi_\ell\Pi_{j,T}
u_{t+h-\ell-j}
+
\mathcal E_{t,h,T},
\end{align}
where
$\mathcal E_{t,h,T}
=
O_p(T^{-(s+1)\bar\zeta})$.
If $s=0$, \eqref{eq:lalp_local} holds with
$(\beta_{h,T},\eta_{h,p,T})=(\beta_h,\eta_{h,d})$, and
\eqref{eq:local-error-decomposition} holds with
$\mathcal E_{t,h,T}=O_p(T^{-\bar\zeta})$.
\end{proposition}

The proof of Proposition~\ref{prop:local-lalp-representation} is given in Appendix \ref{proof:local_mis_lalp_exact}, where
explicit expressions for $\eta_{h,p,T}$ and $\mathcal E_{t,h,T}$ are
also provided. Importantly, the regression form \eqref{eq:lalp_local} shows that under the models \eqref{localmis} and
\eqref{eq:local_arma_error}, the true impulse response is $\beta_{h,T}$. This can be seen from the fact that the error $e_{t,h,T}$ only involves the past ($\mathcal E_{t,h,T}$) and the future (other terms in $e_{t,h,T}$) of $u_t$, but not the contemporaneous innovation $u_t$. Note that $\beta_{h,T}$ generally differs from $\beta_{h}$, the true response under the exact finite-order VAR($d$) as in Section \ref{sec:uniform}, and the difference is of order $T^{-\bar\zeta}$, which can be larger than the order $T^{-1/2}$.

Proposition~\ref{prop:local-lalp-representation} isolates the role of
lag augmentation. With $p=d+s$, the additional $s$ lags successively
absorb the predetermined components through order $s$, leaving terms
of order $T^{-(s+1)\bar\zeta}$. The goal of this section is to study inference based on the same
estimator $\widehat\beta_h^{\operatorname{GMM}(p,q)}$ and its variance
estimator $\widehat V_h^{\operatorname{GMM}(p,q)}$ as in previous
sections under the models \eqref{localmis} and
\eqref{eq:local_arma_error}. The following theorem gives the result.

\begin{theorem}
\label{theorem:local_mis_lalp}
Let $s\geq0$ and
$q\geq p=d+s>1$ be fixed. Suppose $\{y_t\}$ is generated by
\eqref{localmis} and \eqref{eq:local_arma_error}, and
Assumptions~\ref{ass:persistent}--\ref{ass:rotated-nonsingularity}
hold. Assume
\begin{equation}
    \bar\zeta=\min\{\zeta_{AR},\zeta_{MA}\}>1/(2s+2).
\end{equation}
 Then the Wald statistic
$W_{T,h}^{\operatorname{local}}
=
T
\bigl(G\widehat\beta_h^{\operatorname{GMM}(p,q)}-G\beta_{h,T}\bigr)'
\bigl(G\widehat V_h^{\operatorname{GMM}(p,q)}G'\bigr)^{-1}
\bigl(G\widehat\beta_h^{\operatorname{GMM}(p,q)}-G\beta_{h,T}\bigr)$, where $G$ is as in
Theorem~\ref{thm:persistence-uniform-gmm},
satisfies
\begin{equation}
\label{eq:local_mis_uniform}
\limsup_{T\to\infty}
\;
\sup_{(D,B)\in\mathcal P}
\;
\sup_{x\in\mathbb R}
\left|
\mathbb P_{D,B}
\left(
W_{T,h}^{\operatorname{local}}\le x
\right)
-
F_{\chi_g^2}(x)
\right|
=
0.
\end{equation}
\end{theorem}

Theorem~\ref{theorem:local_mis_lalp} establishes persistence-uniform
inference for GMM-LALP under local VARMA misspecification.
With $p=d+s$, $s\geq0$, valid inference requires the model to be sufficiently close to the VAR($d$) such that
$\bar\zeta>1/(2s+2)$. A larger $s$ allows the inference to remain robust under a greater degree of misspecification. Thus $s=0$ gives the trivial condition
$\bar\zeta>1/2$; without lag augmentation, the inference is robust to misspecification of smaller order than $T^{-1/2}$. $s=1$ recovers the $T^{-1/4}$ boundary
considered by \citet{olea2024double}. With two-lag augmentation ($s=2$), the inference can be robust to misspecification of order arbitrarily close to $T^{-1/6}$.


The proof of Theorem \ref{theorem:local_mis_lalp}, given in Appendix~\ref{proof:local_mis_lalp}, combines
Proposition~\ref{prop:local-lalp-representation} with the
persistence-uniform argument in Section~\ref{sec:uniform}. After
rotation, the moments generated by $\mathcal E_{t,h,T}$ are uniformly
negligible because $T^{1/2-(s+1)\bar\zeta}\to0$. When $s=0$, the same
condition gives
$\sqrt T\|\beta_{h,T}-\beta_h\|=o(1)$.
The same argument establishes first-order equivalence between the
feasible HAC matrix based on $\widehat e_{t,h}$ and its infeasible
counterpart based on $e_{t,h,T}$.

The nonsingularity requirement in Assumption~\ref{ass:rotated-nonsingularity} needed for Theorem \ref{theorem:local_mis_lalp} is similar to the one needed for Theorem \ref{thm:persistence-uniform-gmm}. The local VARMA filter perturbs both
the rotated sample second moments and the corresponding long-run
covariance matrix by asymptotically negligible terms, so that their
smallest eigenvalues inherit the lower bounds of the reference VAR.
Appendix~\ref{app:local-rotated-nonsingularity} illustrates how this
condition can be verified under local misspecification in an AR$(1)$
model with local ARMA$(1,1)$ errors.

We conclude this section by an illustration of the failure of the \emph{lag-augmented} VAR-based impulse response estimator of
\citet{inoue2020uniform} (which we refer to as the LA-AR estimator in this paper), in a simple misspecified model. Consider the scalar AR$(1)$ model with local ARMA(1,1) innovations; a special case of \eqref{localmis} and
\eqref{eq:local_arma_error} with $d=m=n=1$. Assume stationarity, $|\alpha|<1$, where $\alpha$ is the AR(1) slope. Let $\zeta_{AR}=\zeta_{MA}=1/2$, so that
$\Pi_{1,T}=T^{-1/2}\pi$, where
$\pi=\Phi_1+\Theta_1$ is scalar. 
The construction proposed by
\citet{inoue2020uniform} first fits an AR$(2)$ to obtain two slopes $\widehat\phi_1$ and $\widehat\phi_2$, and then reports $\widehat\phi_1^{\,h}$ as the impulse response estimate at horizon $h$. The analysis by \cite{olea2024double} focuses on the \emph{non-augmented} VAR-based response estimator, thus does not cover the lag-augmented estimator. 
Since $\beta_h=\alpha^h$, \eqref{eq:local_beta_T} gives
$\beta_{h,T}
=
\alpha^h
+
T^{-1/2}\pi\alpha^{h-1}
+
\sum_{j=2}^{h}
T^{-j/2}\pi\Phi_1^{\,j-1}\alpha^{h-j}$.
The last sum is $O(T^{-1})$ for given $h$.  
In Appendix \ref{sec:uniform-lemmas} (Remark \ref{rem:var-illustration}), we show that
\begin{equation}
\sqrt T\bigl(\widehat\phi_1^{\,h}-\beta_{h,T}\bigr)
\xrightarrow{d}
N\Bigl(\alpha^{h-1}\pi(h-1),\;h^2\alpha^{2(h-1)}\Bigr).
\label{eq:var-illustration-aug}
\end{equation}
Compared with the non-augmented estimator, the lag augmentation eliminates the bias (due to misspecification) only at $h=1$ (unless $\alpha=0$).\footnote{The \emph{non-augmented} VAR-based response estimator fits an AR$(1)$ and reports
$\widehat\alpha^{\,h}$ as the impulse response estimate for given $h$. In Appendix \ref{proof:local_mis_lalp}, we show that, in our setting,
\begin{equation}
\sqrt T\bigl(\widehat\alpha^{\,h}-\beta_{h,T}\bigr)
\xrightarrow{d}
N\Bigl(\alpha^{h-1}\pi\bigl[h(1-\alpha^2)-1\bigr],\;
h^2\alpha^{2(h-1)}(1-\alpha^2)\Bigr).
\label{eq:var-illustration-nonaug}
\end{equation}}
The LA-AR estimator is generally asymptotically biased for $h\ge2$.

Because of the bias, the Wald statistic formed from the LA-AR response estimator $\widehat\phi_1^{\,h}$ (the square of \eqref{eq:var-illustration-aug} divided by a variance estimate) is not asymptotically centered chi-squared, so the inference is asymptotically distorted. If the local
misspecification is larger than order $T^{-1/2}$, equivalently allowing
$\pi\to\infty$, the noncentrality diverges and the LA-AR-based
Wald statistic explodes. By contrast, under the simple illustrative scalar model with local misspecification, the GMM-LALP Wald statistic (including the one based on OLS LALP as a special case) is asymptotically central $\chi_1^2$ by Theorem~\ref{theorem:local_mis_lalp} even when the local
misspecification is larger than order $T^{-1/2}$.

\section{Monte Carlo simulation}
\label{sec:mc}

Sections \ref{sec:mc} and \ref{sec:application} examine the finite-sample performance of inference based on our recommended GMM-LALP procedure, GMM($d+1,d+2$) (overidentification by one), and compare it with the LALP method of \cite{montiel2021local}. Given the popularity of the LALP and the similar robust properties our GMM approach shares, we focus on the LALP as the comparison target. For practical purposes, we mainly consider the bootstrap implementation of both methods following the algorithm given in \cite{montiel2021local}.

\begin{figure}[t]
\centering
    \begin{minipage}{1.05\textwidth}
       \noindent\hspace*{-2cm}%
        \includegraphics[
            width=1.2\linewidth,
            height=1.99\textheight,
            keepaspectratio
        ]{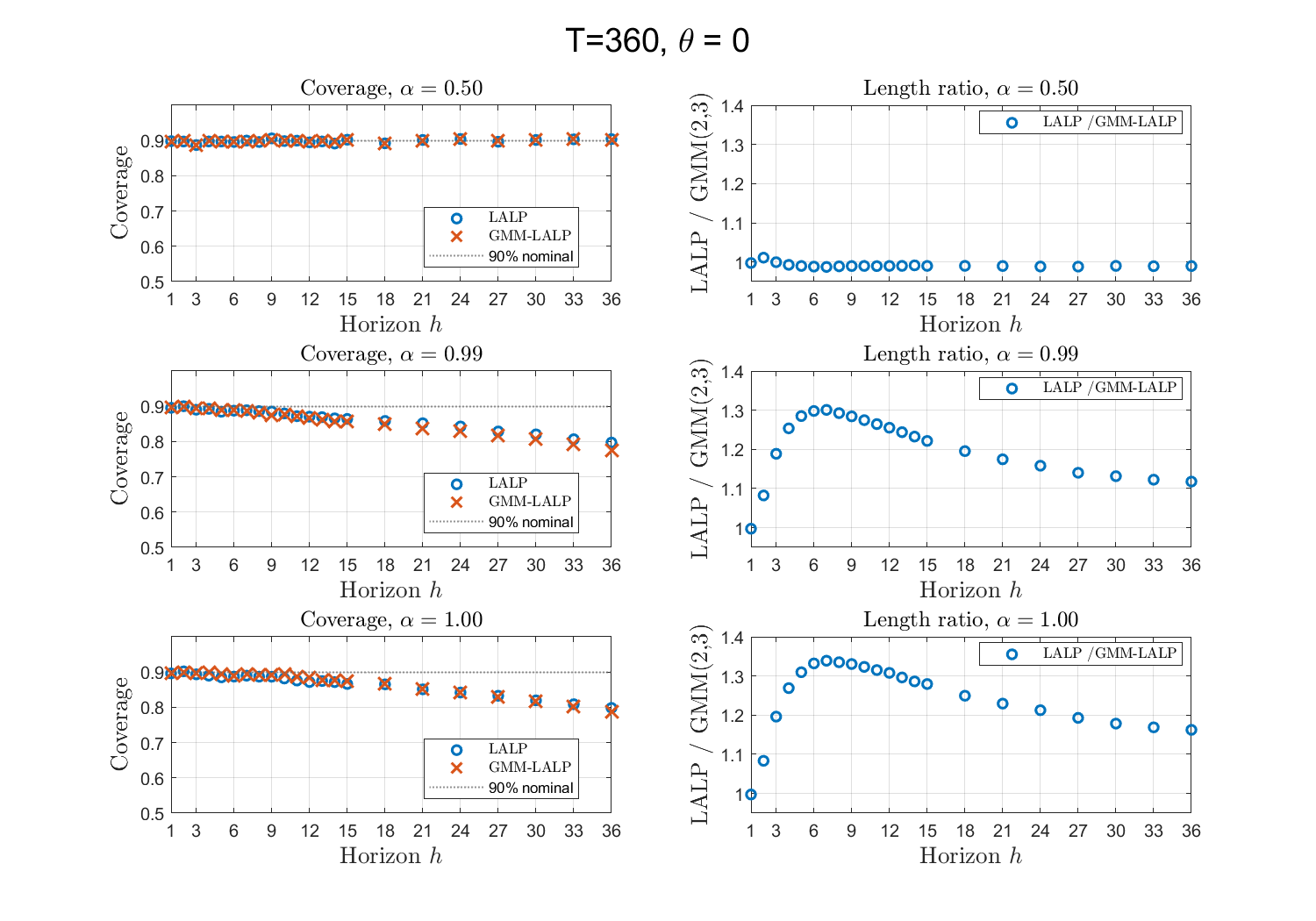}
\caption{Bootstrap confidence intervals (CIs) based on GMM(2,3) and LALP methods, under the AR(1) data generating process with AR coefficient $\alpha \in \{0.5, 0.99, 1\}$, $T = 360$.}
\label{fig:coverage_bandwidth_theta1_360}
        \vspace{-0.0em}
        \footnotesize
        \textit{Notes:}
Left panels plot the coverage probabilities of bootstrap CIs (Nominal level: 90\%). Right panels plot the ratio of the length of LALP to GMM-LALP CIs; values above one indicate that GMM-LALP CIs are shorter.
    \end{minipage}
\end{figure}

In the Monte Carlo simulations, the data are
generated from a scalar ARMA$(1,1)$ process with i.i.d.\ standard normal innovations, the AR coefficient $\alpha \in \{0.5, 0.99, 1\}$ and MA coefficient
$\theta \in \{0, 0.2\}$. For each replication, we demean the series before applying the estimators and bootstrap procedures, and the subsequent local-projection and GMM-LALP regressions are estimated without a constant term. The case $\theta=0$ corresponds to the correctly specified AR
design (with $d=1$), while $\theta=0.2$ introduces misspecification with a moving-average component. We first use the sample
size $T=360$ to match that of our empirical application. We use 5{,}000 Monte Carlo replications and 1{,}000
bootstrap draws. We consider horizons $h = 1, \cdots, 36$.


\begin{figure}[ht]
    \centering
 \begin{minipage}{1.05\textwidth}
       \noindent\hspace*{-2cm}%
        \includegraphics[
            width=1.2\linewidth,
            height=1.99\textheight,
            keepaspectratio
        ]{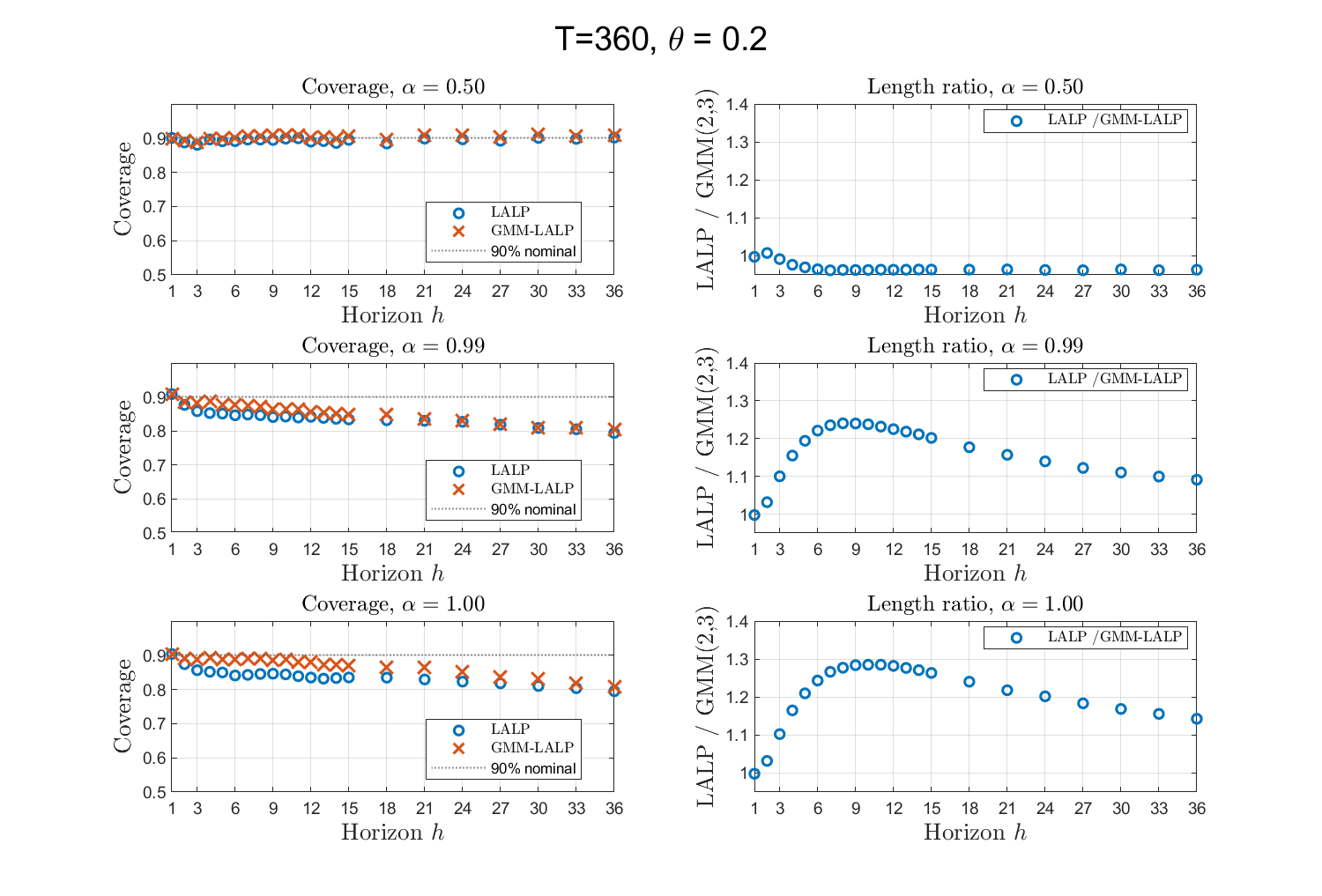}
\caption{Bootstrap confidence intervals (CIs) based on GMM(2,3) and LALP, under the ARMA(1,1) data generating process with AR coefficient $\alpha \in \{0.5, 0.99, 1\}$ and MA coefficient $\theta=0.2$, $T = 360$.}
\label{fig:coverage_bandwidth_theta2_360}
        \vspace{-0.0em}
        \footnotesize
        \textit{Notes:}
Left panels plot the coverage probabilities of bootstrap CIs (Nominal level: 90\%). Right panels plot the ratio of the length of LALP to GMM-LALP CIs; values above one indicate that GMM-LALP CIs are shorter.
    \end{minipage}
\end{figure}

Bootstrap coverage
probabilities and confidence-interval (CI) length ratios for LALP and GMM-LALP are reported in Figure~\ref{fig:coverage_bandwidth_theta1_360} ($\theta=0$) and Figure~\ref{fig:coverage_bandwidth_theta2_360} ($\theta=0.2$). The results show that GMM-LALP CIs generally have similar coverage as LALP CIs, but are remarkably shorter. LALP CIs can be up to 30\% wider than GMM-LALP CIs, and the largest differences appear for the persistent data ($\alpha=0.99$ or 1) and at moderately long horizons. The two CIs perform similarly across all horizons when the data are only weakly dependent ($\alpha=0.5$). We also find that the GMM-LALP can have noticeably better coverage under the misspecified case.

\begin{figure}[ht]
    \centering
 \begin{minipage}{1.05\textwidth}
       \noindent\hspace*{-2cm}%
        \includegraphics[
            width=1.2\linewidth,
            height=1.99\textheight,
            keepaspectratio
        ]{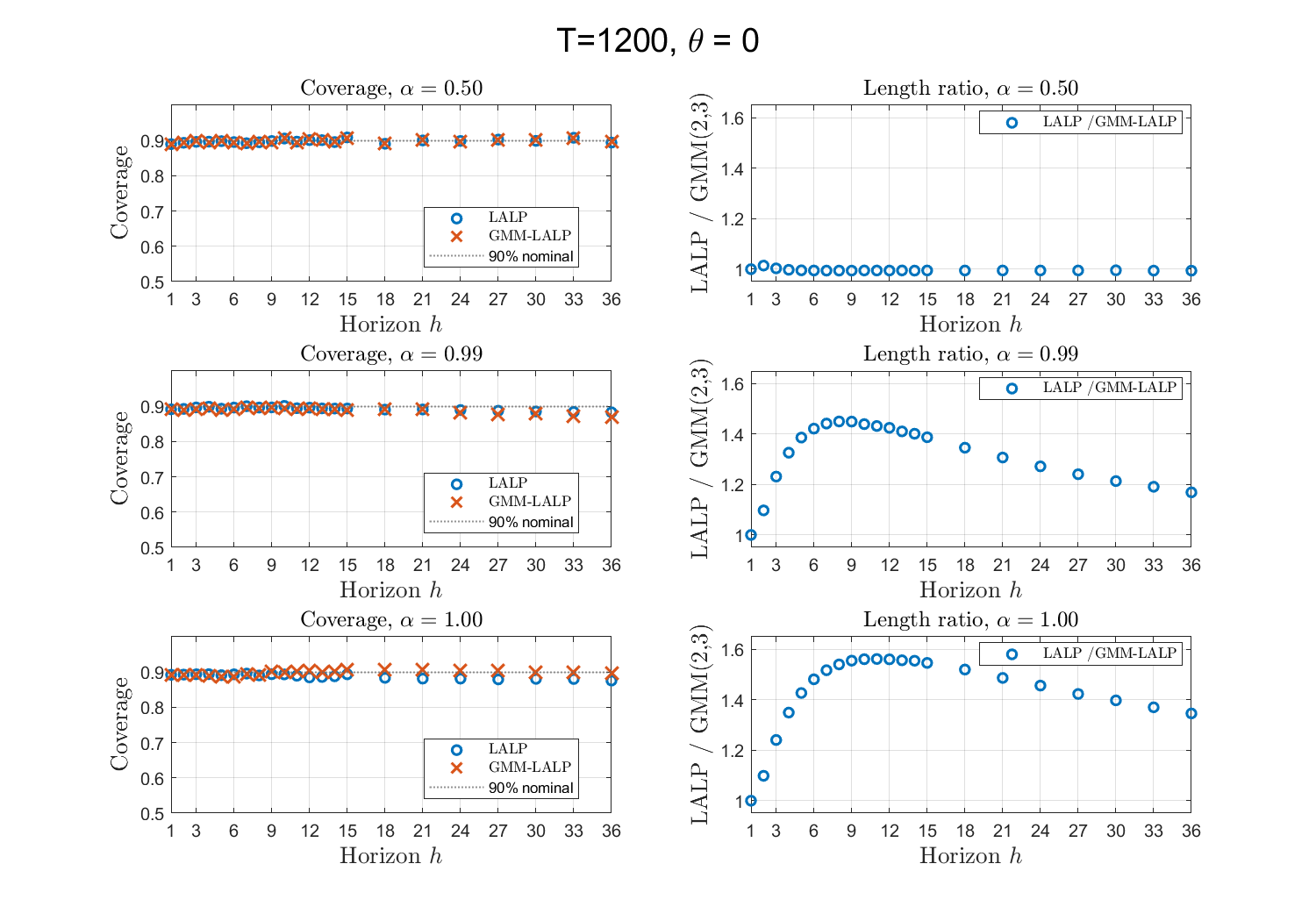}
\caption{Bootstrap confidence intervals (CIs) based on GMM(2,3) and LALP methods, under the AR(1) data generating process with AR coefficient $\alpha \in \{0.5, 0.99, 1\}$, $T = 1200$.}
\label{fig:coverage_bandwidth_theta1_1200}
        \vspace{-0.0em}
        \footnotesize
        \textit{Notes:}
Left panels plot the coverage probabilities of bootstrap CIs (Nominal level: 90\%). Right panels plot the ratio of the length of LALP to GMM-LALP CIs; values above one indicate that GMM-LALP CIs are shorter.
    \end{minipage}
\end{figure}

\begin{figure}[t]
    \centering
 \begin{minipage}{1.05\textwidth}
       \noindent\hspace*{-2cm}%
        \includegraphics[
            width=1.2\linewidth,
            height=1.99\textheight,
            keepaspectratio
        ]{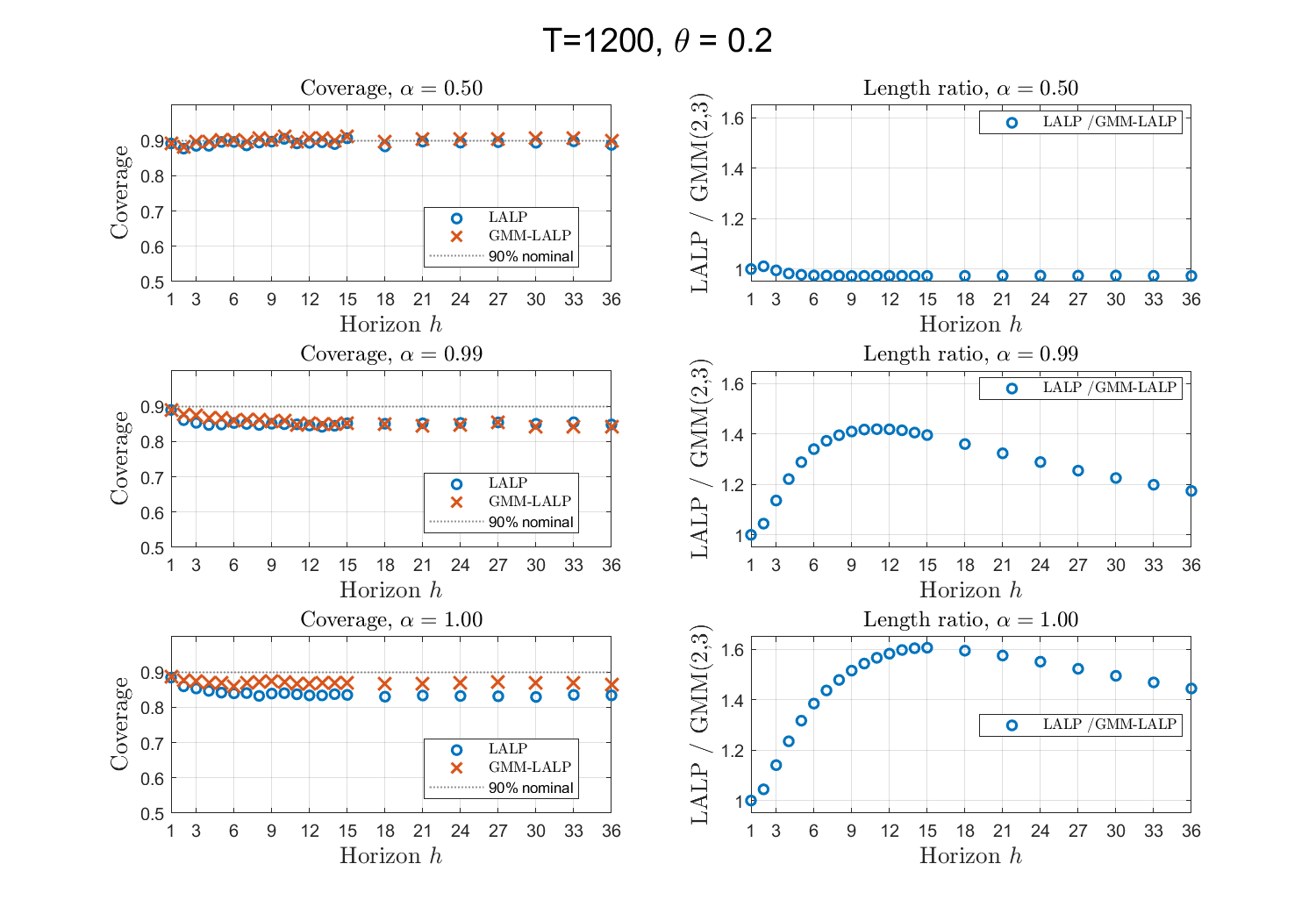}
\caption{Bootstrap confidence intervals (CIs) based on GMM(2,3) and LALP, under the ARMA(1,1) data generating process with AR coefficient $\alpha \in \{0.5, 0.99, 1\}$ and MA coefficient $\theta=0.2$, $T = 1200$.}
\label{fig:coverage_bandwidth_theta2_1200}
        \vspace{-0.0em}
        \footnotesize
        \textit{Notes:}
Left panels plot the coverage probabilities of bootstrap CIs (Nominal level: 90\%). Right panels plot the ratio of the length of LALP to GMM-LALP CIs; values above one indicate that GMM-LALP CIs are shorter.
    \end{minipage}
\end{figure}

We also consider a much larger sample size, $T=1200$, albeit one that may be unrealistic in applications, to better approximate the asymptotic setting underlying our theoretical results. The results are reported in Figures \ref{fig:coverage_bandwidth_theta1_1200} and \ref{fig:coverage_bandwidth_theta2_1200}. Both CIs now have the coverage rate quite close to the nominal level across all cases. The differences in CI lengths are even larger; LALP CIs can be up to 60\% wider than GMM-LALP CIs when the data are persistent.

\section{Empirical Application}
\label{sec:application}

We apply the proposed GMM-LALP estimator to the monthly five-variable VAR model of \cite{baker2016measuring}, which has become a standard setting for studying the macroeconomic
effects of uncertainty shocks. The data contain $360$ monthly
observations from January 1985 through December 2014.\footnote{The data are downloaded from 
Nicholas Bloom's personal website, \text{https://nbloom.people.stanford.edu/research}.} Table \ref{tab:variable_definitions} shows that all variables except EPU are very highly persistent,
which is precisely the environment in which efficiency losses from standard OLS
local-projection estimation can be empirically relevant.



\begin{table}[b]
\centering
\caption{Variable definitions and first-order autocorrelations}
\label{tab:variable_definitions}
\begin{tabular}{c|c|c}
\hline\hline
Variable & Definition & \begin{tabular}[c]{@{}c@{}}First-order  autocorrelation\end{tabular} \\
\midrule
EPU        & Economic policy uncertainty index       & 0.6707 \\
SP500      & Logarithm of the S\&P 500 index         & 0.9984 \\
FFR        & Federal funds rate                       & 0.9973 \\
EMPLOYMENT & Logarithm of employment                  & 0.9996 \\
IP         & Logarithm of industrial production       & 0.9995 \\
\hline\hline
\end{tabular}
\end{table}

For each response variable $i$ and each horizon $h=1,\cdots,36$, we estimate a
local projection using the full current vector
$y_t=(\mathrm{EPU}_t,\mathrm{SP500}_t,\mathrm{FFR}_t,\mathrm{EMPLOYMENT}_t,
\mathrm{IP}_t)'$.\footnote{Prior to computing the impulse responses, each series is demeaned. Accordingly, the local-projection regression and GMM-LALP are estimated without the intercept.} The coefficient on the $j$th
element of $y_t$ is interpreted as the horizon-$h$ (reduced-form) response of variable $i$ to
shock variable $j$, conditional on the other current variables and the
lag-augmented controls. We compute both LALP and GMM-LALP estimates and their bootstrap CIs using the same procedures as the previous simulation section. We use 5{,}000 bootstrap replications.

Figure~\ref{fig:application-epu-shock68} reports the responses of all five variables to the EPU shock with 68\% confidence intervals.\footnote{68\% confidence intervals are quite common in empirical macroeconomics, especially for graphical presentation of impulse responses. We also compute 90\% confidence intervals for these responses. See Appendix \ref{app:IRF}.} To examine whether the efficiency gain of GMM-LALP survives a richer dynamic specification, the figure reports two lag choices: the upper panel uses the AIC-selected (Akaike information
criterion) baseline lag order, $d=3$, and the lower panel uses a substantially larger lag order, $d=12$. With lag augmentation, both
LALP and GMM-LALP use $p=d+1$ (thus 4 or 13) monthly lags as regressors, and GMM uses an additional lag as the excluded instrumental variable.
Thus, the comparison is not driven by a parsimonious lag specification alone; it asks whether optimal GMM weighting continues to improve precision even when the lag-augmented projection includes a large number of controls.

\begin{figure}[t]
\makebox[\linewidth][c]{
\includegraphics[width=1.1\textwidth]{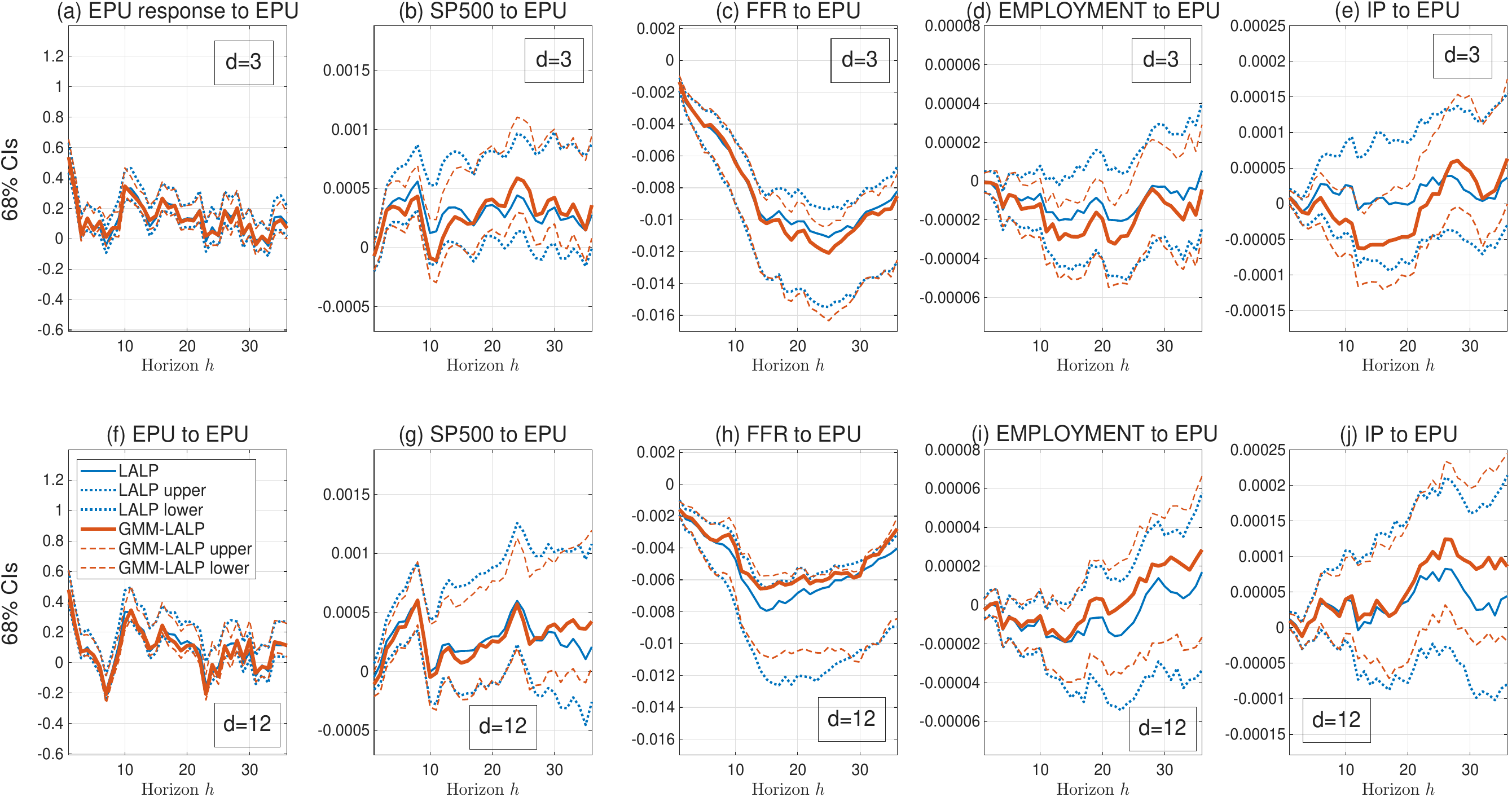}
}
\caption{Responses to an EPU innovation.}
\label{fig:application-epu-shock68}
\begin{minipage}{1\textwidth}
\vspace{2pt}
\footnotesize
\textit{Notes:} The figure plots the impulse responses for EPU, SP500,
FFR, EMPLOYMENT, and IP to an EPU innovation with 68\% confidence intervals. Blue lines are LALP estimates and
bootstrap confidence intervals; orange lines are GMM-LALP estimates and bootstrap confidence intervals. The
model lag order is $d=3$ for the upper panels, and is $d=12$ for the lower panels. The horizons run from 1 to 36 months.
\end{minipage}
\end{figure}

Figure~\ref{fig:application-epu-shock68} can be viewed as a non-structural \textit{local-projection} counterpart to the baseline monthly VAR evidence in \cite{baker2016measuring}. Our exercise differs from \cite{baker2016measuring} in an important respect: we do not impose a recursive structural identification scheme (which can be controversial), and hence the responses in Figure~\ref{fig:application-epu-shock68} should not be read as structural IRFs or compared to their estimates in magnitude.\footnote{Their Figure VIII uses the same broad five-variable monthly system---EPU, the stock market, the federal funds rate, employment, and industrial production---with three lags, and reports Cholesky-identified responses of employment and industrial production to an upward EPU innovation.} However, similar to Figure VIII of \cite{baker2016measuring}, the responses of employment and industrial production move downward after an EPU innovation, indicating that the same empirical system points to weaker real activity following increases in policy uncertainty.

 Overall, except for a few cases, LALP and GMM-LALP deliver similar point estimates for the impulse responses. The GMM-LALP confidence intervals are systematically shorter than the LALP confidence intervals, reflecting the efficiency gain from using valid lagged instruments and optimally weighting the resulting moment conditions. Using $d=3$, the LALP CI for the response of IP can be up to 55\% wider than the GMM CI (on average 22\% wider). Importantly, this pattern remains visible in the lower panel with a larger lag order, albeit to a lesser extent.  Using $d=12$, the LALP CI for the response of IP can be up to almost 30\% wider than the GMM CI (on average 15\% wider).

 The tighter CIs by GMM are also economically meaningful. Using $d=3$, for the response of IP, the LALP CIs cover the zero line (meaning the effects are insignificant) across all horizons, while this is not true for GMM CIs. For the response of EMPLOYMENT, GMM-LALP CIs indicate statistical significance at almost half of the horizons, whereas LALP CIs do so at only a couple of horizons.
In summary, the empirical IRFs reinforce the paper's theoretical and simulation results. The optimal GMM weighting provides additional precision conditional on that same lag-augmented specification. 

\section{Conclusion}
\label{sec:conclu}

This paper proposes a new method for estimating impulse responses and conducting inference based on local projections. Rather than estimating a local projection by OLS for a given set of lag controls, we exploit additional lags of the data as excluded instruments and optimally combine the resulting moment conditions using GMM. We show that these additional moments can deliver substantial efficiency gains relative to OLS, particularly when the data are persistent. We formally establish the robustness of the proposed inference to data persistence and model misspecification.

\renewcommand{\refname}{References}
\putbib

\end{bibunit}

\newpage
\appendix
\titleformat{\section}
  {\normalfont\normalsize\scshape\filcenter}{Appendix \thesection.}{0.7em}{}

\numberwithin{equation}{section}

\begin{bibunit}

\begin{center}

{\Large\bfseries Supplementary Appendix to\\[0.4em]
``Intentionally Overidentifying Local Projections''}

\vspace{1.2em}

{\large
Junjie Guo\footnotemark[1]
\hspace{2em}
Endong Wang\footnotemark[2]
\hspace{2em}
Ke-Li Xu\footnotemark[3]}

\vspace{0.4em}

{\small
\textit{CUFE}
\hspace{4.2em}
\textit{University of Mannheim}
\hspace{4.2em}
\textit{Indiana University}}

\vspace{0.8em}

{\small \today}

\end{center}

\footnotetext[1]{School of Finance, Central University of Finance and Economics,
Shahe Higher Education Park, Changping District, Beijing, P.R. China, 102206.
Email: \href{mailto:junjguog@cufe.edu.cn}{junjguog@cufe.edu.cn}.}

\footnotetext[2]{Department of Economics, University of Mannheim,
L7, 3--5, 68161 Mannheim, Germany.
Tel.: +49 621-181-1879;
Email: \href{mailto:endong.wang@uni-mannheim.de}{endong.wang@uni-mannheim.de}.}

\footnotetext[3]{Department of Economics, Indiana University,
100 S. Woodlawn Ave., Bloomington, IN 47405, USA.
Email: \href{mailto:kelixu@iu.edu}{kelixu@iu.edu}.}

\begin{center}
\textbf{Contents}

\

\begin{tabular}{@{}p{0.16\textwidth}p{0.76\textwidth}@{}}
Appendix~A & AR(1) example \\
\quad A.1 & GMM-LP \\
\quad A.2 & GMM-LALP (Proof of Proposition~\ref{prop:ar1_ordering}) \\
Appendix~B & Proof of Proposition~\ref{prop:optimal_iv_gmm} \\
Appendix~C & Proof of Proposition~\ref{prop:local-lalp-representation} \\
Appendix~D & Proofs of Theorems~\ref{thm:persistence-uniform-gmm} and~\ref{theorem:local_mis_lalp} \\
& \quad Lemmas~\ref{lem:rho-tau-dichotomy}--\ref{lem:stationary-regime}, Remarks \ref{rem:local-intuition}-\ref{rem:mpqw} \\
Appendix~E & Assumption~\ref{ass:rotated-nonsingularity} under local misspecification \\
Appendix~F & Strict efficiency of optimal GMM \\
& \quad Lemma~\ref{lem:gmm_vs_ols} \\
Appendix~G & Bootstrap implementation \\
Appendix~H & 90\% confidence intervals in the empirical application
\end{tabular}
\end{center}

\newpage
Throughout the appendices, $\|\cdot\|$ denotes the Euclidean norm for
vectors and the Frobenius norm for matrices, while $\|\cdot\|_2$
denotes the spectral norm for matrices.

\vspace{0.5\baselineskip}

\section{AR(1) example}
\label{app:ar1_derivations}

This appendix collects the algebra behind the AR(1) efficiency comparisons in Section~\ref{sec:eff}. Throughout, $\{y_t\}$ follows AR(1), where the innovation satisfies $\mathbb E u_t=0$ and $\mathbb E u_t^2=\sigma^2$. All variance formulas below are normalized by $\sigma^2$. For the stationary calculations, let $c_\ell:=\sigma^{-2}\operatorname{Cov}(y_t,y_{t-\ell})$ for $\ell\ge0$, so $c_\ell=\alpha^\ell/(1-\alpha^2)$. For $0\le\ell\le h-1$, define
$\gamma_{h,\ell}
:=
\sigma^{-2}\operatorname{Cov}(e_{t,h},e_{t+\ell,h})
=
\alpha^\ell\sum_{j=0}^{h-\ell-1}\alpha^{2j},
$
with $\gamma_{h,\ell}=0$ for $\ell\ge h$ and $\gamma_{h,-\ell}=\gamma_{h,\ell}$.

\subsection{GMM-LP}
\label{app:gmm_lp}

In Subsection \ref{app:gmm_lp} we assume $|\alpha|<1$. The GMM-LP estimator (i.e. GMM(1,$q$)) uses $y_t$ as the regressor and $Y_{t,q}=(y_t,y_{t-1},\cdots,y_{t-q+1})'$ as the instrument vector. Define
$\omega_{h,r}
:=
\sigma^{-4}\operatorname{LRCov}(y_t e_{t,h},y_{t-r}e_{t,h})
=
\sum_{\ell=1-h}^{h-1} c_{|\ell-r|}\gamma_{h,|\ell|}.$
For $q\ge1$, define the $q \times q$ matrix,
\begin{equation}
\Omega_{q,h}
:=
\sigma^{-4}\operatorname{LRV}(Y_{t,q}e_{t,h})
=
(\omega_{h,|i-j|})_{i,j=0}^{q-1}.
\label{eq:Omega_qh}
\end{equation}
With regressor $Y_{t,p}$ and instrument $Y_{t,q}$, write $\Sigma_{p,q}:=\sigma^{-2}\mathbb E[Y_{t,p}Y_{t,q}']$.

The case $q=1$ coincides with OLS LP. Since $\Sigma_{1,1}=(1-\alpha^2)^{-1}$ and $\Omega_{1,h}=\omega_{h,0}$,
\begin{equation}
V_h^{\mathrm{LP}}
=
(1-\alpha^2)^2\omega_{h,0},
\label{eq:vlp}
\end{equation}
where direct evaluation gives
$\omega_{h,0}
=
\frac{1+\alpha^2-(2h+1)\alpha^{2h}+(2h-1)\alpha^{2h+2}}
{(1-\alpha^2)^3}.$

\textbf{Two instruments}. With two IVs, $\Sigma_{1,2}=(1-\alpha^2)^{-1}(1,\alpha)$ and
$\omega_{h,1}
=
\frac{2\alpha(1-\alpha^{2h})-h\alpha^{2h-1}(1-\alpha^4)}
{(1-\alpha^2)^3}.$
Then
\begin{align}
V_h^{\mathrm{GMM}(1,2)}
=
\bigl(\Sigma_{1,2}\Omega_{2,h}^{-1}\Sigma_{1,2}'\bigr)^{-1}
=
\frac{(1-\alpha^2)^2(\omega_{h,0}^2-\omega_{h,1}^2)}
{\omega_{h,0}(1+\alpha^2)-2\alpha\omega_{h,1}}.
\label{eq:V_gmm_lp}
\end{align}
With $\rho_{h,1}:=\omega_{h,1}/\omega_{h,0}$, we have $
\frac{V_h^{\mathrm{LP}}}{V_h^{\mathrm{GMM}(1,2)}}
=
\frac{1+\alpha^2-2\alpha\rho_{h,1}}{1-\rho_{h,1}^2}
=
1+\frac{(\alpha-\rho_{h,1})^2}{1-\rho_{h,1}^2}.$
Since $\Omega_{2,h}$ is nonsingular, $|\rho_{h,1}|<1$, so $\mathrm{GMM}(1,2)$ weakly improves on OLS LP, with equality only in the redundant cases $h=1$ or $\alpha=0$. Figure \ref{fig:relative_gmm_lp} visualizes the efficiency gain of $\mathrm{GMM}(1,2)$, which depends on $\alpha$ and $h$.

\begin{figure}[t]
\makebox[\linewidth][c]{
\includegraphics[width=1.2\linewidth]{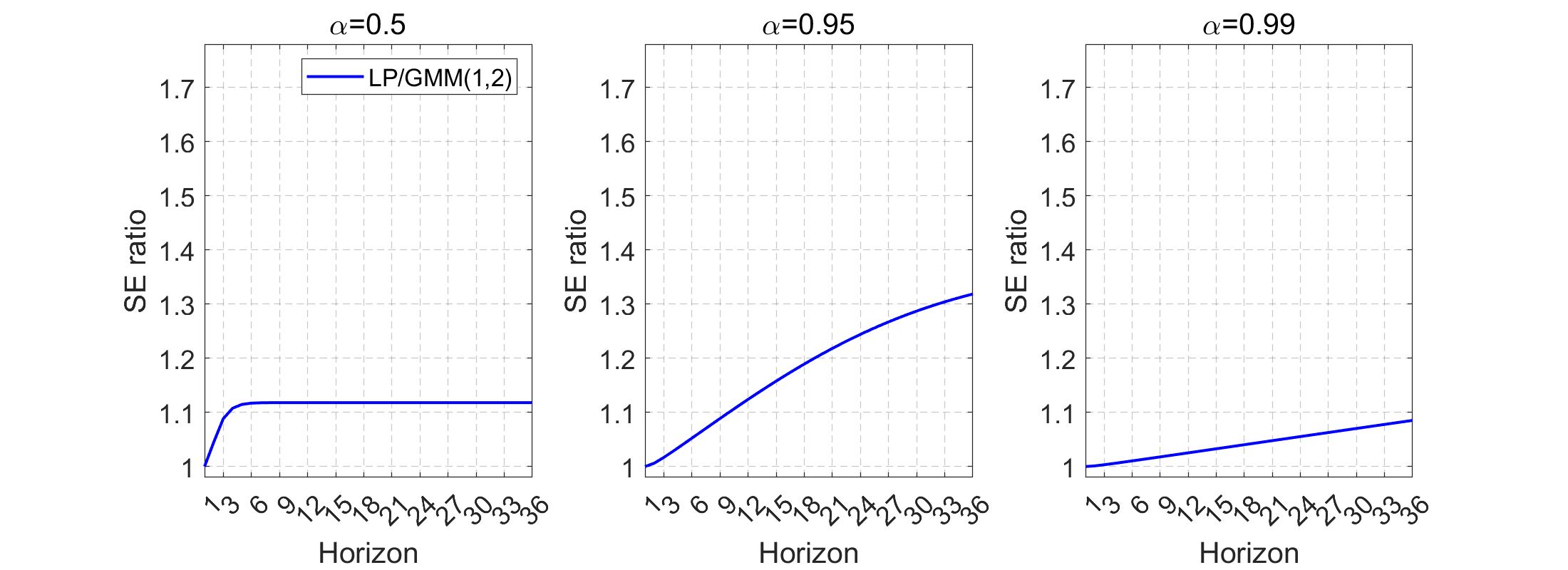}
}
\caption{Asymptotic standard-error ratio for $\mathrm{GMM}(1,2)$.}
\label{fig:relative_gmm_lp}
\vspace{0.4em}
\begin{minipage}{1\linewidth}
\footnotesize
\emph{Note:} The figure plots $\sqrt{V_h^{\mathrm{LP}}/V_h^{\mathrm{GMM}(1,2)}}$. We compute $V_h^{\mathrm{LP}}$ using \eqref{eq:vlp} and compute $V_h^{\mathrm{GMM}(1,2)}$ using \eqref{eq:V_gmm_lp}.
\end{minipage}
\end{figure}

\textbf{More instruments}. For general $q\ge2$, the regressor is $y_t$ and the instrument vector is $Y_{t,q}$. With $\Sigma_{1,q}=(1-\alpha^2)^{-1}(1,\alpha,\cdots,\alpha^{q-1})$ and score long-run covariance $\Omega_{q,h}$,
\begin{align}
V_h^{\mathrm{GMM}(1,q)}
=
\bigl(\Sigma_{1,q}\Omega_{q,h}^{-1}\Sigma_{1,q}'\bigr)^{-1}.
\label{eq:V_gmm_1q}
\end{align}
The sequence $\{V_h^{\mathrm{GMM}(1,q)}\}_{q\ge1}$ is non-increasing, since adding valid moments weakly improves optimally weighted GMM efficiency. Figure \ref{fig:relative_gmm_lp_IV} shows the efficiency gain of $\mathrm{GMM}(1,q)$ over OLS, for a few values of $q$. It is interesting to note that the efficiency gain of our ``intentionally overidentified'' GMM over the LP (Figure \ref{fig:relative_gmm_lp_IV}) is generally less impressive than that over the LALP; c.f. Figure \ref{fig:gmm_lalp_ratio}.

\begin{figure}[t]
\makebox[\linewidth][c]{
\includegraphics[width=1.2\linewidth]{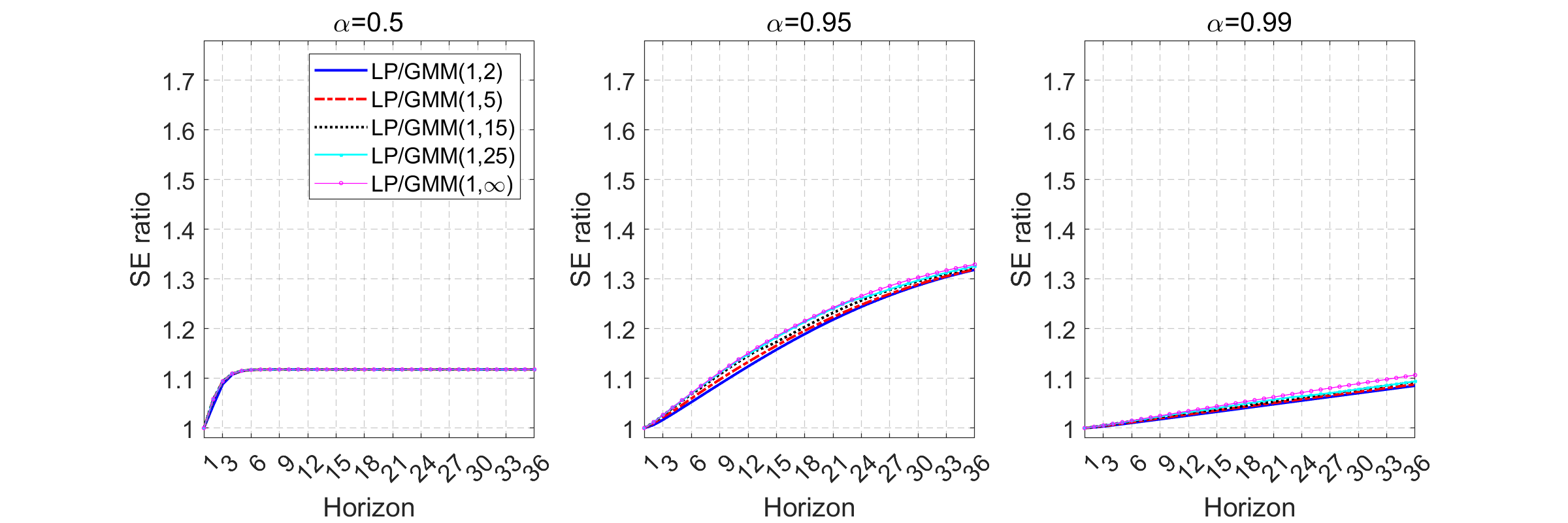}
}
\caption{Asymptotic standard-error ratio for $\mathrm{GMM}(1,q)$, $q \ge 2$.}
\label{fig:relative_gmm_lp_IV}
\vspace{0.4em}
\begin{minipage}{1\linewidth}
\footnotesize
\emph{Note:} The figure plots $\sqrt{V_h^{\mathrm{GMM}(1,1)}/V_h^{\mathrm{GMM}(1,q)}}$ for $q\in\{2,5,15,25,\infty\}$, where $q$ is the number of instruments, and $\mathrm{GMM}(1,1)$ is OLS local projection. We compute
$V_h^{\mathrm{GMM}(1,1)}$ from \eqref{eq:vlp}, the finite-$q$ variances
from \eqref{eq:V_gmm_1q}, and the infinite-instrument limit from
\eqref{eq:V_gmm_lp_infty}.
\end{minipage}
\end{figure}

\textbf{Infinite instruments}.
The closed linear span of
$\{y_{t-j}\}_{j\ge0}$ equals that of $\{u_{t-j}\}_{j\ge0}$.
Since optimal GMM is invariant to nonsingular linear transformations
of the instrument vector, the limit of \eqref{eq:V_gmm_1q} can be
computed using the infinite instrument vector
$(u_t,u_{t-1},\ldots)'$.

Let $G=(G_0,G_1,\ldots)'$ denote the
regressor--instrument moment vector. Since
$y_t=\sum_{j\ge0}\alpha^j u_{t-j}$ and
$\mathbb E(u_t^2)=\sigma^2$, $G_r
:=
\sigma^{-2}\mathbb E[u_{t-r}y_t]
=
\alpha^r, $ $ r\ge0.$
Thus $G=(1,\alpha,\alpha^2,\ldots)'$.

Recall that $\gamma_{h,\ell}
:=
\sigma^{-2}\operatorname{Cov}(e_{t,h},e_{t+\ell,h})$, with $\gamma_{h,\ell}=0$ for $\ell\ge h$. The
infinite-dimensional covariance matrix is
$
\Omega
=
\bigl(\gamma_{h,|r-s|}\bigr)_{r,s\ge0}.
$

Thus, the optimal-GMM asymptotic variance is $(G'\Omega^{-1}G)^{-1}$. We will have
\begin{align}
\Omega^{-1}G
&=
\frac{1}{\gamma_{h,0}}
\left(
1,
\underbrace{0,\ldots,0}_{h-1},
\alpha^h,
\underbrace{0,\ldots,0}_{h-1},
\alpha^{2h},
\underbrace{0,\ldots,0}_{h-1},
\ldots
\right)'.
\label{eq:infinite-gmm-solution}
\end{align}
To verify this, consider the $r$th element after premultiplication by
$\Omega$. If $r=nh$ for positive integer $n$, only the element at position $nh$ contributes,
giving $\alpha^{nh}=G_{nh}$. If $r=nh+\ell$ with
$1\le\ell\le h-1$, only the elements at positions $nh$ and
$(n+1)h$ contribute. Using
$\gamma_{h,\ell}+\alpha^h\gamma_{h,h-\ell}
=\alpha^\ell\gamma_{h,0}$ gives
$\alpha^{nh+\ell}=G_r$.

It follows that
$G'\Omega^{-1}G
=\gamma_{h,0}^{-1}\sum_{n\ge0}\alpha^{2nh}
=\{\gamma_{h,0}(1-\alpha^{2h})\}^{-1}$.
Therefore,
\begin{align}
\lim_{q\to\infty}V_h^{\mathrm{GMM}(1,q)}
&=
\left(G'\Omega^{-1}G\right)^{-1}
=
\gamma_{h,0}(1-\alpha^{2h})
=
\frac{(1-\alpha^{2h})^2}{1-\alpha^2}.
\label{eq:V_gmm_lp_infty}
\end{align}

\textbf{Comparison with the iterated AR estimator}. The iterated AR impulse response estimator applies the nonlinear transformation $\alpha$ as $\alpha^h$ to the OLS estimate of $\alpha$. Since $\sqrt T(\widehat\alpha-\alpha)\overset{d}{\to}N(0,1-\alpha^2)$, the asymptotic variance of the AR estimator is
$V_h^{\mathrm{AR}}=h^2\alpha^{2h-2}(1-\alpha^2),$ where $\alpha \neq 0.$
Combining this expression with \eqref{eq:V_gmm_lp_infty},
\begin{align}
\frac{\lim_{q\to\infty}V_h^{\mathrm{GMM}(1,q)}}{V_h^{\mathrm{AR}}}
=
\left(
\frac{h^{-1}\sum_{j=0}^{h-1}\alpha^{2j}}
{|\alpha|^{h-1}}
\right)^2.
\label{eq:lp_vs_ar_ratio}
\end{align}
The numerator inside the square is the arithmetic mean of $\{\alpha^{2j}\}_{j=0}^{h-1}$ and the denominator is their geometric mean. By AM--GM (arithmetic mean–geometric mean inequality), the ratio is weakly above one, strictly so whenever $h>1$ and $0<|\alpha|<1$. Infinite-instrument GMM-LP therefore does not attain the iterated AR variance; the gap reflects the unimposed cross-horizon restriction $\beta_h=\alpha^h$.

\subsection{GMM-LALP (Proof of Proposition~\ref{prop:ar1_ordering})}
\label{subsec:hac_gmmlalp}
The first overidentified GMM-LALP estimator (i.e. GMM(2,3)) estimates LA-LP using one excluded lag as the IV: the regressor is $Y_{t,2}$ and the instrument is $Y_{t,3}$. For $q\ge3$, GMM(2,$q$) uses the instrument $Y_{t,q}$.

\textit{Proof of Proposition~\ref{prop:ar1_ordering}}.

\emph{(i)} For $|\alpha|<1$,
$P(y_{t-2}\mid y_{t-1})=\alpha y_{t-1}$, since
$\mathbb E[y_{t-2}y_{t-1}]
=\alpha\mathbb E[y_{t-2}^2]
=\alpha\mathbb E[y_{t-1}^2]$.
For $k\geq0$, define
$v_{t-k}:=y_{t-k-1}-\alpha y_{t-k}$.
Since
$\operatorname{Cov}(y_t,y_{t-j})
=\sigma^2\alpha^j/(1-\alpha^2)$ for $j\geq0$,
$\mathbb E v_{t-k}=0$ and
$\mathbb E v_{t-k}^2=\sigma^2$.
For $k\geq1$,
$\operatorname{Cov}(v_t,v_{t-k})
=\sigma^2\bigl(\alpha^k-\alpha^k-\alpha^{k+2}+\alpha^{k+2}\bigr)/(1-\alpha^2)
=0$.

Since $y_{t-k-1}=v_{t-k}+\alpha y_{t-k}$, it gives
$\operatorname{span}(y_{t-1},\ldots,y_{t-q+1})
=\operatorname{span}(y_{t-1},v_{t-1},\ldots,v_{t-q+2})$ by repeated substitution.
For $1\leq j\leq k$,
$\operatorname{Cov}(v_{t-k},y_{t-j})
=
\sigma^2
\bigl(
\alpha^{k+1-j}-\alpha^{k+1-j}
\bigr)/(1-\alpha^2)=0$.
Thus
$v_{t-k}\perp\operatorname{span}(y_{t-1},\ldots,y_{t-k})$.

After partialling out $y_{t-1}$, $y_t$ becomes $u_t$ and the
$v_{t-k}$ remain unchanged. Hence, we let
$\Omega_{h,q}^{\perp}
:=\sigma^{-4}\operatorname{LRV}(z_{t,q}^{\perp}e_{t,h})$ where $ z_{t,q}^{\perp}=
\bigl(u_t,v_{t-1},v_{t-2},\ldots,v_{t-q+2}\bigr)'.$
Using $\operatorname{Cov}(v_t,v_{t-k})=0$ for $k\geq1$ gives
\begin{align}
\Omega_{h,q}^{\perp}
&=
\begin{pmatrix}
\gamma_{h,0} & r_{h,q}' \\
r_{h,q} & A_{h,q}
\end{pmatrix},
\label{eq:Omega_perp_main}
\end{align}
where, for $k=1,\ldots,q-2$,
\begin{align}
(r_{h,q})_k
&=
\begin{cases}
-(h-k)\alpha^{2h-k-1}, & 1\leq k\leq h-1,\\
0, & k\geq h,
\end{cases}
\label{eq:r_hq}
\end{align}
and, for $k,\ell=1,\ldots,q-2$,
\begin{align}
(A_{h,q})_{k\ell}
&=
\gamma_{h,|k-\ell|}
=
\begin{cases}
\alpha^{|k-\ell|}
\displaystyle\sum_{j=0}^{h-|k-\ell|-1}\alpha^{2j},
& |k-\ell|\leq h-1,\\
0,
& |k-\ell|\geq h.
\end{cases}
\label{eq:A_hq}
\end{align}

Since
$\sigma^{-2}\mathbb E[y_t^{\perp}z_{t,q}^{\perp\prime}]
=(1,0,\ldots,0)$, the Schur-complement formula\footnote{For a partitioned matrix
$M=\begin{pmatrix}A&B\\ C&D\end{pmatrix}$ with $D$ nonsingular,
the Schur complement of $D$ in $M$ is
$A-BD^{-1}C$. If $A-BD^{-1}C$ is nonsingular, then
$(M^{-1})_{11}=(A-BD^{-1}C)^{-1}$.
See \citet[Eq.~(0.7.3.1)]{horn2013matrix}.} gives
\begin{align}
V_h^{\mathrm{GMM}(2,q)}
&=
\gamma_{h,0}
-
r_{h,q}'A_{h,q}^{-1}r_{h,q}.
\label{eq:V_gmmlalp_q_main}
\end{align}

For $q=3$,
$A_{h,3}=\gamma_{h,0}$ and
$r_{h,3}=-(h-1)\alpha^{2h-2}$, so
\begin{align}
V_h^{\mathrm{GMM}(2,3)}
&=
\gamma_{h,0}
-
\frac{(h-1)^2\alpha^{4h-4}}{\gamma_{h,0}}=
\sum_{j=0}^{h-1}\alpha^{2j}
-
\frac{(h-1)^2\alpha^{4h-4}}
{\sum_{j=0}^{h-1}\alpha^{2j}}.
\label{eq:V_gmm23_compact}
\end{align}

The derivation above is for $|\alpha|<1$. The uniform asymptotic
theory in Section~\ref{sec:uniform} implies that the same variance
formulas (continuous in $\alpha$ and bounded from below) apply at $\alpha=1$. In particular,
$V_h^{\mathrm{GMM}(2,3)}=(2h-1)/h$, when $\alpha=1$.

\emph{(ii)}
Let $A_{h,\infty}$ and $r_{h,\infty}$ be given by
\eqref{eq:A_hq} and \eqref{eq:r_hq}, respectively, with
$k,\ell\geq1$. Since $\gamma_{h,\ell}=0$ for $\ell\geq h$,
$(A_{h,\infty})_{k\ell}=0$ whenever $|k-\ell|\geq h$.

For $h\geq2$,
\begin{align}
\begin{split}
\left(A_{h,\infty}^{-1}r_{h,\infty}\right)_{1:h}
&=
\left(
-\alpha^{2h-2},
-(1-\alpha^2)\alpha^{2h-3},
\ldots,
-(1-\alpha^2)\alpha^h,
\alpha^{h+1}
\right)',\\
\left(A_{h,\infty}^{-1}r_{h,\infty}\right)_{k+h}
&=
\alpha^h
\left(A_{h,\infty}^{-1}r_{h,\infty}\right)_k,
\qquad k\geq1.
\end{split}
\label{eq:infinite-gmmlalp-solution}
\end{align}
To verify \eqref{eq:infinite-gmmlalp-solution}, multiplication by
$A_{h,\infty}$ gives $(r_{h,\infty})_k$ for
$k=1,\ldots,h-1$. For $k\geq h$, the second line of
\eqref{eq:infinite-gmmlalp-solution} and
$(A_{h,\infty})_{k\ell}=0$ for $|k-\ell|\geq h$ give zero.
Hence
$A_{h,\infty}(A_{h,\infty}^{-1}r_{h,\infty})=r_{h,\infty}$.

Since $(r_{h,\infty})_k=0$ for $k\geq h$,
\begin{align}
r_{h,\infty}'A_{h,\infty}^{-1}r_{h,\infty}
&=
(h-1)\alpha^{4h-4}
+
(1-\alpha^2)
\sum_{k=2}^{h-1}
(h-k)\alpha^{4h-2k-2}=
\alpha^{2h}\gamma_{h-1,0}.
\label{eq:infinite-gmmlalp-quadratic}
\end{align}
The last equality follows from
$\gamma_{h-1,0}=\sum_{j=0}^{h-2}\alpha^{2j}$.

Substituting \eqref{eq:infinite-gmmlalp-quadratic} into
\eqref{eq:V_gmmlalp_q_main} gives
\begin{align}
V_h^{\mathrm{GMM}(2,\infty)}
&=
\gamma_{h,0}
-
\alpha^{2h}\gamma_{h-1,0}
=
\sum_{j=0}^{h-1}\alpha^{2j}
-
\sum_{j=h}^{2h-2}\alpha^{2j}.
\label{asym:var:inf}
\end{align}
For $h=1$, $r_{1,\infty}=0$, so \eqref{asym:var:inf} follows directly.

\vspace{0.5\baselineskip}

\section{Proof of Proposition \ref{prop:optimal_iv_gmm}}
\label{app:proof_optimal_iv}

The result is algebraic. Let $\widetilde Y_{t,q}=(z_{t,p,q}',w_{t,p}')'$ be the
reordered instrument vector, with $z_{t,p,q}=(y_t',y_{t-p}',\cdots,y_{t-q+1}')'$
and $Y_{t,p}=(y_t',w_{t,p}')'$. Since $\widetilde Y_{t,q}$ is a permutation of
$Y_{t,q}$, the GMM estimator is unchanged after the corresponding permutation of
the moments and weight matrix.

Let $\widehat\Phi$ and $\widehat\Theta$ denote the sample residualization
matrices such that
\begin{align}
\widehat\Phi Y_{t,p}
&=
\begin{bmatrix}
\widehat y_t^\perp\\
w_{t,p}
\end{bmatrix},
&
\widehat\Theta \widetilde Y_{t,q}
&=
\begin{bmatrix}
\widehat z_{t,p,q}^\perp\\
w_{t,p}
\end{bmatrix}.
\label{eq:app_resid_transforms}
\end{align}
Thus $\widehat\Phi$ partials $w_{t,p}$ out of $y_t$ and $\widehat\Theta$ partials
it out of $z_{t,p,q}$, over $t=q,\cdots,T-h$.

Define
$\widetilde\Sigma_{p,q}=T^{-1}\sum_{t=q}^{T-h}Y_{t,p}\widetilde Y_{t,q}'$,
$\widetilde m_{q,h}=T^{-1}\sum_{t=q}^{T-h}\widetilde Y_{t,q}y_{1,t+h}$, and
let $\widetilde\Omega_{q,h}$ be the HAC matrix based on
$\widetilde Y_{t,q}\widehat e_{t,h}$. Since $J_p\widehat\Phi'=J_p$, the
coefficient block on $y_t$ can be computed after this change of coordinates:
\begin{align}
\widehat\beta_h^{\operatorname{GMM}(p,q)}
&=
J_p
\Bigl[
(\widehat\Phi\widetilde\Sigma_{p,q}\widehat\Theta')
(\widehat\Theta\widetilde\Omega_{q,h}\widehat\Theta')^{-1}
(\widehat\Phi\widetilde\Sigma_{p,q}\widehat\Theta')'
\Bigr]^{-1}
\nonumber\\
&\quad\times
(\widehat\Phi\widetilde\Sigma_{p,q}\widehat\Theta')
(\widehat\Theta\widetilde\Omega_{q,h}\widehat\Theta')^{-1}
\widehat\Theta\widetilde m_{q,h}.
\label{eq:app_gmm_transformed}
\end{align}

By least-squares orthogonality,
$\sum_{t=q}^{T-h}\widehat y_t^\perp w_{t,p}'=0$ and
$\sum_{t=q}^{T-h}\widehat z_{t,p,q}^\perp w_{t,p}'=0$, so
\[
\widehat\Phi\widetilde\Sigma_{p,q}\widehat\Theta'
=
\begin{pmatrix}
\widehat\Sigma_{p,q}^{\perp} & 0\\
0 & \widehat W_p
\end{pmatrix},
\qquad
\widehat\Theta\widetilde m_{q,h}
=
\begin{bmatrix}
\widehat m_{p,q,h}^{\perp}\\
\widehat m_{w,h}
\end{bmatrix},
\]
where $\widehat\Sigma_{p,q}^{\perp}=T^{-1}\sum_{t=q}^{T-h}
\widehat y_t^\perp(\widehat z_{t,p,q}^\perp)'$,
$\widehat W_p=T^{-1}\sum_{t=q}^{T-h}w_{t,p}w_{t,p}'$, and
$\widehat m_{p,q,h}^{\perp}=T^{-1}\sum_{t=q}^{T-h}
\widehat z_{t,p,q}^\perp y_{1,t+h}$. Write
\[
\widehat\Theta\widetilde\Omega_{q,h}\widehat\Theta'
=
\begin{pmatrix}
\widehat\Omega_{p,q,h}^{\perp} & \widehat\Xi\\
\widehat\Xi' & \widehat\Omega_{ww}
\end{pmatrix},
\]
whose upper-left block is the HAC matrix based on
$\widehat z_{t,p,q}^\perp\widehat e_{t,h}$.

We use the following block identity. Let $D=\operatorname{diag}(S,W)$ with $W$ square and nonsingular, and let
$\Omega=\begin{pmatrix}\Omega_z&\Xi\\ \Xi'&\Omega_w\end{pmatrix}$ be positive
definite. Then the first block row of $(D\Omega^{-1}D')^{-1}D\Omega^{-1}$ is
$\left(S\Omega_z^{-1}S'\right)^{-1}[S\Omega_z^{-1},0]$: squareness of $W$ makes
both the cross block $\Xi$ and the second moment block drop out of the first
coefficient block, and the Schur identity for $\Omega^{-1}$ reduces its
$(1,1)$-block contribution to $\Omega_z^{-1}$. The same Schur-complement calculation shows that the first diagonal block of
$(D\Omega^{-1}D')^{-1}$ is
$\left(S\Omega_z^{-1}S'\right)^{-1}$. Since $\widehat W_p$ is square and
nonsingular, applying the identity to \eqref{eq:app_gmm_transformed} gives
\begin{align}
\widehat\beta_h^{\operatorname{GMM}(p,q)}
&=
\left(
\widehat\Sigma_{p,q}^{\perp}
(\widehat\Omega_{p,q,h}^{\perp})^{-1}
\widehat\Sigma_{p,q}^{\perp\prime}
\right)^{-1}
\widehat\Sigma_{p,q}^{\perp}
(\widehat\Omega_{p,q,h}^{\perp})^{-1}
\widehat m_{p,q,h}^{\perp}.
\label{eq:app_oiv_reduced}
\end{align}

Define $\widehat z_{t,p,q}^*=\widehat\Sigma_{p,q}^{\perp}
(\widehat\Omega_{p,q,h}^{\perp})^{-1}\widehat z_{t,p,q}^{\perp}$. Then
$\sum_{t=q}^{T-h}\widehat z_{t,p,q}^*\widehat y_t^{\perp\prime}
=T\widehat\Sigma_{p,q}^{\perp}
(\widehat\Omega_{p,q,h}^{\perp})^{-1}\widehat\Sigma_{p,q}^{\perp\prime}$ and
$\sum_{t=q}^{T-h}\widehat z_{t,p,q}^*y_{1,t+h}
=T\widehat\Sigma_{p,q}^{\perp}
(\widehat\Omega_{p,q,h}^{\perp})^{-1}\widehat m_{p,q,h}^{\perp}$, so
\eqref{eq:app_oiv_reduced} is exactly
\[
\widehat\beta_h^{\operatorname{GMM}(p,q)}
=
\left(
\sum_{t=q}^{T-h}
\widehat z_{t,p,q}^* \widehat y_t^{\perp\prime}
\right)^{-1}
\sum_{t=q}^{T-h}
\widehat z_{t,p,q}^* y_{1,t+h}
=
\widehat\beta_h^{\operatorname{OIV}(p,q)}.
\]


The covariance identity follows from the same block calculation. In the
sample,
\begin{align}
\widehat V_h^{\operatorname{GMM}(p,q)}
&=
J_p
\left(
\widehat\Sigma_{p,q}
\widehat\Omega_{q,h}^{-1}
\widehat\Sigma_{p,q}'
\right)^{-1}
J_p'
=
\left(
\widehat\Sigma_{p,q}^{\perp}
(\widehat\Omega_{p,q,h}^{\perp})^{-1}
\widehat\Sigma_{p,q}^{\perp\prime}
\right)^{-1}
=
\widehat V_h^{\operatorname{OIV}(p,q)}.
\label{eq:app_feasible_variance_identity}
\end{align}

At the population level,
$\mathbb E[y_t^\perp w_{t,p}']=0$ and
$\mathbb E[z_{t,p,q}^\perp w_{t,p}']=0$,
so the residualized regressor--instrument cross moment is
\begin{align}
\mathbb E
\left[
\begin{pmatrix}
y_t^\perp\\
w_{t,p}
\end{pmatrix}
\begin{pmatrix}
z_{t,p,q}^\perp\\
w_{t,p}
\end{pmatrix}'
\right]
&=
\begin{pmatrix}
\Sigma_{p,q}^{\perp} & 0\\
0 & W_p
\end{pmatrix}.
\label{eq:app_population_cross_moment}
\end{align}
The corresponding long-run covariance matrix has the general partition
\begin{align}
\sum_{\ell=-\infty}^{\infty}
\mathbb E
\left[
\begin{pmatrix}
z_{t,p,q}^\perp\\
w_{t,p}
\end{pmatrix}
e_{t,h}e_{t-\ell,h}
\begin{pmatrix}
z_{t-\ell,p,q}^\perp\\
w_{t-\ell,p}
\end{pmatrix}'
\right]
&=
\begin{pmatrix}
\Omega_{p,q,h}^{\perp} & \Xi\\
\Xi' & \Omega_{ww}
\end{pmatrix},
\label{eq:app_population_block_matrices}
\end{align}
where $\Xi$ need not be zero. The permutation and residualization steps are invertible and therefore do
not affect nonsingularity. Since $W_p$ is square and nonsingular, the block
identity above then implies
\begin{align}
V_h^{\operatorname{GMM}(p,q)}
&=
J_p
\left(
\Sigma_{p,q}
\Omega_{q,h}^{-1}
\Sigma_{p,q}'
\right)^{-1}
J_p'
\nonumber\\
&=
J_p
\Biggl[
\begin{pmatrix}
\Sigma_{p,q}^{\perp} & 0\\
0 & W_p
\end{pmatrix}
\begin{pmatrix}
\Omega_{p,q,h}^{\perp} & \Xi\\
\Xi' & \Omega_{ww}
\end{pmatrix}^{-1}
\begin{pmatrix}
\Sigma_{p,q}^{\perp\prime} & 0\\
0 & W_p
\end{pmatrix}
\Biggr]^{-1}
J_p'
\nonumber\\
&=
\left(
\Sigma_{p,q}^{\perp}
(\Omega_{p,q,h}^{\perp})^{-1}
\Sigma_{p,q}^{\perp\prime}
\right)^{-1}
=
V_h^{\operatorname{OIV}(p,q)}.
\label{eq:app_population_variance_identity}
\end{align}

\vspace{0.5\baselineskip}

\section{Proof of Proposition \ref{prop:local-lalp-representation}}
\label{proof:local_mis_lalp_exact}

\begin{proof}
Set $\Pi_{0,T}=I_K$. For all sufficiently large $T$, the
autoregressive factors in \eqref{eq:local_arma_error} are invertible.
By \eqref{eq:local-filter-expansion},
\begin{equation}
\widetilde u_t
=
\sum_{j\geq0}\Pi_{j,T}u_{t-j}.
\label{eq:local-filter-innovation}
\end{equation}

The VAR($d$) recursion gives
\begin{align}
y_{1,t+h}
=
(\beta_h',\eta_{h,d}')Y_{t,d}
+
\nu_1'
\sum_{\ell=0}^{h-1}
\Psi_\ell\widetilde u_{t+h-\ell}.
\label{eq:local-proof-reference-projection}
\end{align}
Substituting \eqref{eq:local-filter-innovation} gives
\begin{align}
y_{1,t+h}
&=
(\beta_h',\eta_{h,d}')Y_{t,d}
+
\nu_1'
\sum_{\ell=0}^{h-1}
\Psi_\ell
\sum_{j=0}^{h-\ell-1}
\Pi_{j,T}u_{t+h-\ell-j}
\nonumber\\
&\quad
+
\sum_{r\geq0}
\left(
\nu_1'
\sum_{\ell=0}^{h-1}
\Psi_\ell\Pi_{h+r-\ell,T}
\right)u_{t-r}.
\label{eq:local-proof-collected}
\end{align}

Suppose first that $s\geq1$. For $r=0,\ldots,s-1$, define recursively
$\delta_{r,h,T}'
=
\nu_1'\sum_{\ell=0}^{h-1}\Psi_\ell\Pi_{h+r-\ell,T}
-
\sum_{k=0}^{r-1}\delta_{k,h,T}'\Pi_{r-k,T}$,
where an empty sum is zero. In particular,
\begin{align}
\delta_{0,h,T}'
&=
\sum_{\ell=0}^{h-1}
\beta_\ell'\Pi_{h-\ell,T}
=
\sum_{j=1}^{h}
\beta_{h-j}'\Pi_{j,T}
=
(\beta_{h,T}-\beta_h)'.
\label{eq:local-delta-recursion}
\end{align}
This proves \eqref{eq:local_beta_T}.

For $r=0,\ldots,s-1$,
\begin{equation}
u_{t-r}
=
y_{t-r}
-
\sum_{j=1}^{d}A_jy_{t-r-j}
-
\sum_{j\geq1}\Pi_{j,T}u_{t-r-j}.
\label{eq:local-proof-ut-identity}
\end{equation}
Successive substitution in \eqref{eq:local-proof-collected} gives
$y_{1,t+h}=(\beta_{h,T}',\eta_{h,p,T}')Y_{t,p}+e_{t,h,T}$, where
$p=d+s$ and
\begin{align}
\eta_{h,p,T}'
&=
[\eta_{h,d}',0_{1\times Ks}]
-
[\delta_{0,h,T}'(A_1,\ldots,A_d),
0_{1\times K(s-1)}]
\nonumber\\
&\quad
+
\sum_{r=1}^{s-1}
[
0_{1\times K(r-1)},
\delta_{r,h,T}',
-\delta_{r,h,T}'(A_1,\ldots,A_d),
0_{1\times K(s-r-1)}
].
\label{eq:local-eta-p}
\end{align}
Empty zero blocks and sums are omitted.

Moreover,
\begin{align}
e_{t,h,T}
&=
\nu_1'
\sum_{\ell=0}^{h-1}
\Psi_\ell u_{t+h-\ell}
+
\nu_1'
\sum_{j=1}^{\min\{s,h-1\}}
\sum_{\ell=0}^{h-j-1}
\Psi_\ell\Pi_{j,T}
u_{t+h-\ell-j}
+
\mathcal E_{t,h,T},
\end{align}
where
\begin{align}
\mathcal E_{t,h,T}
&=
\nu_1'
\sum_{j=s+1}^{h-1}
\sum_{\ell=0}^{h-j-1}
\Psi_\ell\Pi_{j,T}u_{t+h-\ell-j}
\nonumber\\
&\quad
+
\sum_{r\geq s}
\left(
\nu_1'
\sum_{\ell=0}^{h-1}
\Psi_\ell\Pi_{h+r-\ell,T}
-
\sum_{k=0}^{s-1}
\delta_{k,h,T}'\Pi_{r-k,T}
\right)u_{t-r}.
\label{eq:local-error-remainder}
\end{align}
The first sum is zero when $s\geq h-1$.

The recursion and these bounds give
$\|\delta_{r,h,T}\|=O(T^{-(r+1)\bar\zeta})$ for
$r=0,\ldots,s-1$,
$\sum_{j=s+1}^{h-1}\|\Pi_{j,T}\|
=O(T^{-(s+1)\bar\zeta})$, and
\begin{align}
\sum_{r\geq s}
\left\|
\nu_1'
\sum_{\ell=0}^{h-1}
\Psi_\ell\Pi_{h+r-\ell,T}
-
\sum_{k=0}^{s-1}
\delta_{k,h,T}'\Pi_{r-k,T}
\right\|
=
O(T^{-(s+1)\bar\zeta}).
\end{align}
By absolute summability and the moment condition on $u_t$,
$\mathcal E_{t,h,T}=O_p(T^{-(s+1)\bar\zeta})$.
This proves \eqref{eq:lalp_local} and
\eqref{eq:local-error-decomposition} for $s\geq1$.

If $s=0$, \eqref{eq:local-proof-reference-projection} and
\eqref{eq:local-filter-innovation} give
\begin{align}
y_{1,t+h}
&=
(\beta_h',\eta_{h,d}')Y_{t,d}
+
\nu_1'
\sum_{\ell=0}^{h-1}
\Psi_\ell u_{t+h-\ell}
+
\mathcal E_{t,h,T},
\end{align}
where
$\mathcal E_{t,h,T}
=
\nu_1'\sum_{\ell=0}^{h-1}\Psi_\ell
\sum_{j\geq1}\Pi_{j,T}u_{t+h-\ell-j}$.
By absolute summability and the moment condition on $u_t$,
$\mathcal E_{t,h,T}=O_p(T^{-\bar\zeta})$.
\end{proof}

\vspace{0.5\baselineskip}

\section{Proofs of Theorems \ref{thm:persistence-uniform-gmm} and
\ref{theorem:local_mis_lalp}}
\label{sec:uniform-lemmas}

We first prove Theorem~\ref{theorem:local_mis_lalp}.
Theorem~\ref{thm:persistence-uniform-gmm} follows by setting the local
ARMA coefficients to zero.

For a random vector $X$ and $a\geq1$, write
$\|X\|_{L^a}=(\mathbb E\|X\|_2^a)^{1/a}$.
For deterministic matrices, $\|\cdot\|_2$ denotes the spectral norm.
Let $\lesssim$ denote inequality up to a multiplicative constant. We use
\citet[Theorem~2]{andrews1988laws} for laws of large numbers,
\citet[Corollary~1]{dejong1997central} for triangular-array NED central
limit theorems, and
\citet[Theorem~2.1]{dejong2000consistency} for kernel long-run covariance
estimation. Their conditions follow from
Assumption~\ref{ass:stable-score}(ii), the geometric NED bounds, and the
uniform moment bounds established below.

The uniformity results in Theorem~\ref{thm:persistence-uniform-gmm} and \ref{theorem:local_mis_lalp} are established via the sequential approach used in
\citet{andrews2010asymptotic}, \citet{andrews2012estimation}, and
\citet{andrews2020generic}. For an arbitrary sequence of parameters in
$\mathcal P$, every subsequence admits a further subsequence along which
the persistence configuration falls into one of the preceding limiting
regimes. Since the same chi-square limit obtains along every such
subsequence, the sequential characterization of uniform convergence
yields the results like \eqref{eq:uniform-wald-convergence}.

Throughout, $\xi_t=(I_K-DL)y_t$, so that $B(L)\xi_t=\widetilde u_t$, and
$e_{t,h}=\nu_1'\sum_{\ell=0}^{h-1}\Psi_\ell u_{t+h-\ell}$, while
$e_{t,h,T}$ is given by
\eqref{eq:lalp_local}--\eqref{eq:local-error-decomposition}. Let $d$ satisfy
$A(L)=I_K-A_1L-\cdots-A_dL^d=B(L)(I_K-DL)$.
For Theorem~\ref{theorem:local_mis_lalp},
$s\geq0$, $q\geq p=d+s>1$ are fixed and
$\bar\zeta>1/(2s+2)$. Along any subsequence such that
$\rho_k\to\rho_k^\ast\in[-1,1]$ for every $k$, define
$\mathcal I_0=\{k:|\rho_k^\ast|<1\}$ and
$\mathcal I_1=\{k:|\rho_k^\ast|=1\}$,
$\Lambda_{0,T}=\operatorname{diag}(\tau_{k,T}:k\in\mathcal I_0)$,
$\Lambda_{1,T}=\operatorname{diag}(\tau_{k,T}:k\in\mathcal I_1)$,
$y_t^{(0)}=(y_{k,t})_{k\in\mathcal I_0}$, and
$y_t^{(1)}=(y_{k,t})_{k\in\mathcal I_1}$. Define
$Y_{t,p}^{\xi} = \bigl(\xi_t', \xi_{t-1}', \dots, \xi_{t-p+2}', (\Lambda_{0,T}^{-1/2}y_{t-p+1}^{(0)})'\bigr)'$
and define $Y_{t,q}^{\xi}$ analogously with $p$ replaced by $q$.
After permuting the last $K$ coordinates,
\begin{align}
H_pY_{t,p}
&=
\begin{pmatrix}
Y_{t,p}^{\xi}\\
\Lambda_{1,T}^{-1/2}y_{t-p+1}^{(1)}
\end{pmatrix},
&
H_qY_{t,q}
&=
\begin{pmatrix}
Y_{t,q}^{\xi}\\
\Lambda_{1,T}^{-1/2}y_{t-q+1}^{(1)}
\end{pmatrix}.
\label{eq:aux-rotated-partition}
\end{align}
Write $\widecheck\Sigma_{p,q}^{\xi} = T^{-1}\sum_{t=q}^{T-h} Y_{t,p}^{\xi}(Y_{t,q}^{\xi})'$. Let $\widecheck\Omega_{q,h}^{\xi}$ be the principal block of
$\widecheck\Omega_{q,h}$ corresponding to $Y_{t,q}^{\xi}$, and define
$\Sigma_{p,q,T}^{\xi} = \mathbb{E}\bigl[Y_{t,p}^{\xi}(Y_{t,q}^{\xi})'\bigr]$ and $\Omega_{q,h,T}^{\xi} = \operatorname{Var}\bigl(T^{-1/2}\sum_{t=q}^{T-h} Y_{t,q}^{\xi}e_{t,h}\bigr)$. Whenever $\max_k|\rho_k|\le1-c$, define
$\Sigma_{p,q,T} = \mathbb{E}\bigl[Y_{t,p}Y_{t,q}'\bigr]$ and $\Omega_{q,h,T} = \operatorname{Var}\bigl(T^{-1/2}\sum_{t=q}^{T-h} Y_{t,q}e_{t,h,T}\bigr)$.

\subsection{Proof of Theorem~\ref{theorem:local_mis_lalp}}
\label{proof:local_mis_lalp}
 
\begin{proof}
Suppose first that $s\geq1$. By the sequential characterization of uniform
convergence, it suffices to consider an arbitrary parameter sequence in
$\mathcal P$.
 
By Proposition~\ref{prop:local-lalp-representation},
$y_{1,t+h}=(\beta_{h,T}',\eta_{h,p,T}')Y_{t,p}+e_{t,h,T}$.
Equation~\eqref{gmm_iv} and $J_p(H_p')^{-1}=J_p$ give
\begin{align}
\sqrt T\left(\widehat\beta_h^{\operatorname{GMM}(p,q)}-\beta_{h,T}\right)
&=
J_p
\left(
\widehat\Sigma_{p,q}\widehat\Omega_{q,h}^{-1}\widehat\Sigma_{p,q}'
\right)^{-1}
\widehat\Sigma_{p,q}\widehat\Omega_{q,h}^{-1}
T^{-1/2}\sum_{t=q}^{T-h}Y_{t,q}e_{t,h,T}
\label{eq:local-proof-unrotated-expansion}
\\
&=
J_p
\left(
\widecheck\Sigma_{p,q}\widecheck\Omega_{q,h}^{-1}
\widecheck\Sigma_{p,q}'
\right)^{-1}
\widecheck\Sigma_{p,q}\widecheck\Omega_{q,h}^{-1}
T^{-1/2}\sum_{t=q}^{T-h}H_qY_{t,q}e_{t,h,T}.
\label{eq:local-proof-expansion}
\end{align}
 
Fix an arbitrary subsequence. By compactness of $[-1,1]^K$, pass to a
further subsequence such that $\rho_k\to\rho_k^\ast$ for every $k$, and use
$\mathcal I_0$, $\mathcal I_1$, $Y_{t,p}^{\xi}$, and $Y_{t,q}^{\xi}$
defined above. The permutation in \eqref{eq:aux-rotated-partition} leaves
the coordinates selected by $J_p$ unchanged. When
$\mathcal I_1=\varnothing$, the lower-right blocks below are absent.

Define
$\widecheck Q_{p,q,h}^{\xi}
=
\widecheck\Sigma_{p,q}^{\xi}
(\widecheck\Omega_{q,h}^{\xi})^{-1}
(\widecheck\Sigma_{p,q}^{\xi})'$.
Lemma~\ref{lem:stable-persistent-separation} gives
\begin{align}
\widecheck\Sigma_{p,q}
&=
\begin{pmatrix}
\widecheck\Sigma_{p,q}^{\xi} & o_p(1)\\
o_p(1) & O_p(1)
\end{pmatrix},
\label{eq:local-proof-sigma-block}
\\
\widecheck\Omega_{q,h}
&=
\begin{pmatrix}
\widecheck\Omega_{q,h}^{\xi} & o_p(1)\\
o_p(1) & O_p(1)
\end{pmatrix},
\label{eq:local-proof-omega-block}
\\
T^{-1/2}
\sum_{t=q}^{T-h}
H_qY_{t,q}e_{t,h,T}
&=
\begin{pmatrix}
T^{-1/2}\displaystyle\sum_{t=q}^{T-h}
Y_{t,q}^{\xi}e_{t,h}+o_p(1)
\\[1mm]
O_p(1)
\end{pmatrix},
\label{eq:local-proof-score-block}
\end{align}
and, as shown at the end of this proof,
\begin{align}
\widecheck\Omega_{q,h}^{-1}
&=
\begin{pmatrix}
(\widecheck\Omega_{q,h}^{\xi})^{-1} & o_p(1)\\
o_p(1) & O_p(1)
\end{pmatrix},
\label{eq:local-proof-omega-inverse}
\\
\left(
\widecheck\Sigma_{p,q}
\widecheck\Omega_{q,h}^{-1}
\widecheck\Sigma_{p,q}'
\right)^{-1}
&=
\begin{pmatrix}
(\widecheck Q_{p,q,h}^{\xi})^{-1} & o_p(1)\\
o_p(1) & O_p(1)
\end{pmatrix}.
\label{eq:local-proof-Q-inverse-block}
\end{align}
 
Given these results, $p>1$ implies that the coordinates selected by
$J_p$ are contained in $Y_{t,p}^{\xi}$. Hence
\begin{align}
\sqrt T\left(\widehat\beta_h^{\operatorname{GMM}(p,q)}-\beta_{h,T}\right)
&=
J_p
\begin{pmatrix}
(\widecheck Q_{p,q,h}^{\xi})^{-1}
\widecheck\Sigma_{p,q}^{\xi}
(\widecheck\Omega_{q,h}^{\xi})^{-1}
T^{-1/2}\displaystyle\sum_{t=q}^{T-h}Y_{t,q}^{\xi}e_{t,h}\\[1mm]
0
\end{pmatrix}
+o_p(1),
\label{eq:local-proof-selected-expansion}
\\
\widehat V_h^{\operatorname{GMM}(p,q)}
&=
J_p
\begin{pmatrix}
(\widecheck Q_{p,q,h}^{\xi})^{-1} & 0\\
0 & 0
\end{pmatrix}
J_p'
+o_p(1).
\label{eq:local-proof-selected-variance}
\end{align}

Let
$Q_{p,q,h,T}^{\xi}
=\Sigma_{p,q,T}^{\xi}
(\Omega_{q,h,T}^{\xi})^{-1}
(\Sigma_{p,q,T}^{\xi})'$.
Lemma~\ref{lem:stable-block-approximation} gives
$\widecheck\Sigma_{p,q}^{\xi}-\Sigma_{p,q,T}^{\xi}=o_p(1)$,
$\widecheck\Omega_{q,h}^{\xi}-\Omega_{q,h,T}^{\xi}=o_p(1)$,
\eqref{eq:stable-block-score-clt}, and, for all sufficiently large $T$,
$\|\Sigma_{p,q,T}^{\xi}\|_2+\|\Omega_{q,h,T}^{\xi}\|_2\le C$,
$\lambda_{\min}(\Sigma_{p,q,T}^{\xi})\ge c_\Sigma$, and
$\lambda_{\min}(\Omega_{q,h,T}^{\xi})\ge c_\Omega$.
The weighted-Gram eigenvalue bound%
\footnote{The weighted-Gram eigenvalue bound states the following fact. For a full-row-rank matrix $A$ and a symmetric positive definite matrix $B$,
$\lambda_{\min}(AB^{-1}A')
\ge \lambda_{\min}(A)^2/\lambda_{\max}(B)$ and
$\lambda_{\max}(AB^{-1}A')
\le \lambda_{\max}(A)^2/\lambda_{\min}(B)$.
See \citet[Theorem~4.2.2]{horn2013matrix}.}
therefore gives
$0<c\le\lambda_{\min}(Q_{p,q,h,T}^{\xi})
\le\lambda_{\max}(Q_{p,q,h,T}^{\xi})\le C$.
Thus
$\widecheck Q_{p,q,h}^{\xi}-Q_{p,q,h,T}^{\xi}=o_p(1)$ and
$(\widecheck Q_{p,q,h}^{\xi})^{-1}-(Q_{p,q,h,T}^{\xi})^{-1}=o_p(1)$.
Since $J_p$ selects coordinates from the upper-left block, Cauchy's
interlacing theorem%
\footnote{Cauchy's interlacing theorem states the following fact. If $B$ is a principal submatrix of a symmetric matrix $A$, then
$\lambda_{\min}(B)\ge\lambda_{\min}(A)$.
See \citet[Theorem~4.3.28]{horn2013matrix}.}
implies
$0<c\le
\lambda_{\min}\!\left[
J_p\bigl(\begin{smallmatrix}
(Q_{p,q,h,T}^{\xi})^{-1}&0\\0&0
\end{smallmatrix}\bigr)J_p'
\right]
\le
\lambda_{\max}\!\left[
J_p\bigl(\begin{smallmatrix}
(Q_{p,q,h,T}^{\xi})^{-1}&0\\0&0
\end{smallmatrix}\bigr)J_p'
\right]\le C$.

Using these limits in \eqref{eq:local-proof-selected-expansion} gives
\begin{align}
\begin{split}
    &\sqrt T\left(\widehat\beta_h^{\operatorname{GMM}(p,q)}-\beta_{h,T}\right)\\
=&
J_p
\begin{pmatrix}
(Q_{p,q,h,T}^{\xi})^{-1}
\Sigma_{p,q,T}^{\xi}
(\Omega_{q,h,T}^{\xi})^{-1/2}\\[1mm]
0
\end{pmatrix}
(\Omega_{q,h,T}^{\xi})^{-1/2}
T^{-1/2}\sum_{t=q}^{T-h}Y_{t,q}^{\xi}e_{t,h}
+o_p(1).
\end{split}
\label{eq:local-proof-beta-population-expansion}
\end{align}
By the definition of $Q_{p,q,h,T}^{\xi}$,
\begin{align}
\begin{split}
   & \operatorname{Var}\!\left[
J_p
\begin{pmatrix}
(Q_{p,q,h,T}^{\xi})^{-1}
\Sigma_{p,q,T}^{\xi}
(\Omega_{q,h,T}^{\xi})^{-1/2}\\
0
\end{pmatrix}
(\Omega_{q,h,T}^{\xi})^{-1/2}
T^{-1/2}\sum_{t=q}^{T-h}Y_{t,q}^{\xi}e_{t,h}
\right]\\
=&
J_p
\begin{pmatrix}
(Q_{p,q,h,T}^{\xi})^{-1} & 0\\
0 & 0
\end{pmatrix}
J_p'.
\end{split}
\label{eq:local-proof-beta-leading-covariance}
\end{align}
Premultiplying \eqref{eq:local-proof-beta-population-expansion} by
$\bigl[J_p\bigl(\begin{smallmatrix}
(Q_{p,q,h,T}^{\xi})^{-1}&0\\0&0
\end{smallmatrix}\bigr)J_p'\bigr]^{-1/2}$
therefore produces a coefficient matrix with orthonormal rows,%
\footnote{By \eqref{eq:local-proof-beta-leading-covariance},
\begin{align*}
&
\left[
J_p
\begin{pmatrix}
(Q_{p,q,h,T}^{\xi})^{-1} & 0\\
0 & 0
\end{pmatrix}
J_p'
\right]^{-1/2}
J_p
\begin{pmatrix}
(Q_{p,q,h,T}^{\xi})^{-1}
\Sigma_{p,q,T}^{\xi}
(\Omega_{q,h,T}^{\xi})^{-1/2}\\
0
\end{pmatrix}
\\
&\qquad\times
\left\{
\left[
J_p
\begin{pmatrix}
(Q_{p,q,h,T}^{\xi})^{-1} & 0\\
0 & 0
\end{pmatrix}
J_p'
\right]^{-1/2}
J_p
\begin{pmatrix}
(Q_{p,q,h,T}^{\xi})^{-1}
\Sigma_{p,q,T}^{\xi}
(\Omega_{q,h,T}^{\xi})^{-1/2}\\
0
\end{pmatrix}
\right\}'
=
I_K.
\end{align*}}
and hence with spectral norm one. 
By Lemma~\ref{lem:stable-block-approximation},
$(\Omega_{q,h,T}^{\xi})^{-1/2}T^{-1/2}\sum_tY_{t,q}^{\xi}e_{t,h}
\xrightarrow{d}N(0,I)$.
Along any subsequence on which the coefficient matrix converges, Slutsky's
theorem, the preceding eigenvalue bounds, and the block identity above give
\begin{align}
\left[
J_p
\begin{pmatrix}
(Q_{p,q,h,T}^{\xi})^{-1} & 0\\
0 & 0
\end{pmatrix}
J_p'
\right]^{-1/2}
\sqrt T\left(
\widehat\beta_h^{\operatorname{GMM}(p,q)}-\beta_{h,T}
\right)
&\xrightarrow{d}N(0,I_K).
\label{eq:local-proof-beta-standardized}
\end{align}
Similarly,
\begin{align}
\widehat V_h^{\operatorname{GMM}(p,q)}
&=
J_p
\begin{pmatrix}
(Q_{p,q,h,T}^{\xi})^{-1} & 0\\
0 & 0
\end{pmatrix}
J_p'
+o_p(1).
\label{eq:local-proof-variance-approximation}
\end{align}
Since $G$ has full row rank, the congruence eigenvalue bound%
\footnote{The congruence eigenvalue bound states the following fact. For a symmetric positive
definite matrix $V$ and a full-row-rank matrix $G$,
$\lambda_{\min}(GVG')
\ge \lambda_{\min}(V)\lambda_{\min}(GG')$ and
$\lambda_{\max}(GVG')
\le \lambda_{\max}(V)\lambda_{\max}(GG')$.
See \citet[Theorem~4.2.2]{horn2013matrix}.}
preserves the preceding eigenvalue bounds under premultiplication and
postmultiplication by $G$ and $G'$. Hence
\eqref{eq:local-proof-beta-standardized},
\eqref{eq:local-proof-variance-approximation}, and Slutsky's theorem give
$W_{T,h}^{\operatorname{local}}\xrightarrow{d}\chi_g^2$.

It remains to verify \eqref{eq:local-proof-omega-inverse} and
\eqref{eq:local-proof-Q-inverse-block}.

\noindent
\emph{Proof of \eqref{eq:local-proof-omega-inverse}.}
By Assumption~\ref{ass:rotated-nonsingularity} and Cauchy's interlacing
theorem, both diagonal blocks of \eqref{eq:local-proof-omega-block}
have eigenvalues bounded away from zero with probability approaching one,
and hence $O_p(1)$ inverses. Therefore, the Schur complement of the
lower-right block satisfies
\begin{align}
\widecheck\Omega_{q,h}^{\xi}
-
o_p(1)O_p(1)o_p(1)
=
\widecheck\Omega_{q,h}^{\xi}+o_p(1).
\end{align}
The block-inverse formula then gives
\eqref{eq:local-proof-omega-inverse}.
 
\medskip
\noindent
\emph{Proof of \eqref{eq:local-proof-Q-inverse-block}.}
Lemma~\ref{lem:rotated-product-bounds} and
Lemma~\ref{lem:hac-feasible} give
$\|\widecheck\Sigma_{p,q}\|_2+\|\widecheck\Omega_{q,h}\|_2=O_p(1)$.
Since $q-p$ is fixed, $H_pY_{t,p}$ is an exact uniformly bounded linear
transformation of $H_qY_{t,q}$. The singular-value inequality,%
\footnote{The singular-value inequality states the following fact. For conformable matrices $A$ and $B$, with $A$ of full row
rank,
$\lambda_{\min}(A)
\ge \lambda_{\min}(AB)/\|B\|_2$.
See \citet[Eq.~(7.3.14)]{horn2013matrix}.}
the weighted-Gram eigenvalue bound, and
Assumption~\ref{ass:rotated-nonsingularity} therefore give, with
probability approaching one,
\begin{align}
\lambda_{\min}
\left(
\widecheck\Sigma_{p,q}
\widecheck\Omega_{q,h}^{-1}
\widecheck\Sigma_{p,q}'
\right)
&\ge
\frac{
\lambda_{\min}
\left(
\widecheck\Sigma_{p,q}
\right)^2
}{
\lambda_{\max}
\left(
\widecheck\Omega_{q,h}
\right)
}
\ge
\frac{
C^{-2}
\lambda_{\min}
\left(
\widecheck\Sigma_{p,p}
\right)^2
}{
\lambda_{\max}
\left(
\widecheck\Omega_{q,h}
\right)
}
\ge
c
>
0.
\label{eq:local-proof-Q-lower}
\end{align}
Equations~\eqref{eq:local-proof-sigma-block} and
\eqref{eq:local-proof-omega-inverse} imply
\begin{align}
\widecheck\Sigma_{p,q}
\widecheck\Omega_{q,h}^{-1}
\widecheck\Sigma_{p,q}'
&=
\begin{pmatrix}
\widecheck Q_{p,q,h}^{\xi}
+
o_p(1)
&
o_p(1)
\\
o_p(1)
&
O_p(1)
\end{pmatrix}.
\label{eq:local-proof-Q-block}
\end{align}
By \eqref{eq:local-proof-Q-lower} and Cauchy's interlacing theorem, both
diagonal blocks in \eqref{eq:local-proof-Q-block} have minimum
eigenvalue bounded away from zero with probability approaching one.
Weyl's inequality%
\footnote{Weyl's inequality states the following fact. For symmetric matrices $A$ and $B$,
$\left|
\lambda_{\min}(A+B)-\lambda_{\min}(A)
\right|
\le \|B\|_2$.
See \citet[Corollary~4.3.15]{horn2013matrix}.}
then gives $(\widecheck Q_{p,q,h}^{\xi})^{-1}=O_p(1)$, and the
block-inverse formula applied to \eqref{eq:local-proof-Q-block} gives
\eqref{eq:local-proof-Q-inverse-block}.
 
\medskip
\noindent
It remains to consider $s=0$, where
Proposition~\ref{prop:local-lalp-representation} gives
$y_{1,t+h}=(\beta_h',\eta_{h,d}')Y_{t,d}+e_{t,h,T}$, so that the
preceding argument applies with $\beta_h$ in place of $\beta_{h,T}$.
Moreover, \eqref{eq:local_beta_T} and
\eqref{eq:local-filter-coefficient-bounds} give
$\sqrt T\|\beta_{h,T}-\beta_h\|_2
\le\sqrt T\sum_{j=1}^{h}\|\Pi_{j,T}\|_2\|\beta_{h-j}\|_2
\le CT^{1/2-\bar\zeta}=o(1)$,
since $\bar\zeta>1/2$ when $s=0$. Thus the same limit holds when the
Wald statistic is centered at $\beta_{h,T}$.
 
The parameter sequence and selected subsequence were arbitrary. Hence
every subsequence admits a further subsequence along which
$W_{T,h}^{\operatorname{local}}\xrightarrow{d}\chi_g^2$.
The sequential characterization of uniform convergence, continuity of
$F_{\chi_g^2}$, and P\'olya's theorem yield
\eqref{eq:local_mis_uniform}.
\end{proof}
 
\begin{remark}[Lag augmentation and local misspecification]
\label{rem:local-intuition}
By \eqref{eq:aux-local-error-gap},
\begin{align}
\begin{split}
    &T^{-1/2}\sum_t
H_qY_{t,q}\left(e_{t,h,T}-e_{t,h}\right)\\
=&
T^{-1/2}\sum_t
H_qY_{t,q}\nu_1'
\sum_{j=1}^{\min\{s,h-1\}}
\sum_{\ell=0}^{h-j-1}
\Psi_\ell\Pi_{j,T}u_{t+h-\ell-j}+
T^{-1/2}\sum_t
H_qY_{t,q}\mathcal E_{t,h,T}\\
=&
O_p(T^{-\bar\zeta})
+
O_p(T^{1/2-(s+1)\bar\zeta}).
\end{split}
\label{eq:aux-local-remark-score}
\end{align}
The first rate holds for every $\bar\zeta>0$ because
$\mathbb E(u_{t+h-\ell-j}\mid\mathcal F_t)=0$ for
$j\le\min\{s,h-1\}$ and $\ell\le h-j-1$, whereas the second is
$o_p(1)$ if and only if $(s+1)\bar\zeta>1/2$. The same decomposition
gives
$\|H_q(\widehat\Omega_{q,h}-\widetilde\Omega_{q,h})H_q'\|_2
=O_p(T^{-\bar\zeta})
+O_p(b_TT^{-\bar\zeta-1/2})
+O_p(b_TT^{-(s+1)\bar\zeta})
=o_p(1)$.
\end{remark}
 
\begin{remark}[Neither non-augmented nor lag-augmented VAR inference is robust to local misspecification, illustrated under the stationary AR(1) model with local ARMA(1,1) misspecification]
\label{rem:var-illustration}
Let $\{u_t\}$ be i.i.d.\ with $\mathbb E u_t=0$ and
$\mathbb E u_t^2=\sigma^2$. Consider
$y_t=\alpha y_{t-1}+\widetilde u_t$ and
$(1-\Phi_TL)\widetilde u_t=(1+\Theta_TL)u_t$, where
$|\alpha|<1$,
$\Phi_T=T^{-1/2}\Phi_1$,
$\Theta_T=T^{-1/2}\Theta_1$, and
$\Pi_T=\Phi_T+\Theta_T=T^{-1/2}\pi$ with
$\pi=\Phi_1+\Theta_1$. Thus
$\zeta_{AR}=\zeta_{MA}=1/2$. For all sufficiently large $T$,
$\widetilde u_t
=
u_t+\Pi_Tu_{t-1}
+\Pi_T\sum_{j\ge2}\Phi_T^{\,j-1}u_{t-j}$.
Since
$\Pi_T\Phi_T^{\,j-1}
=
T^{-j/2}\pi\Phi_1^{\,j-1}$,
$\|\Pi_T\sum_{j\ge2}\Phi_T^{\,j-1}u_{t-j}\|_{L^2}
=O(T^{-1})$.
Hence the local AR and MA departures enter the first-order calculations
only through $\Pi_T=\Phi_T+\Theta_T$.
 
By \eqref{eq:local_beta_T}, the true impulse response is exactly
$\beta_{h,T}
=
\alpha^h
+
T^{-1/2}\pi\alpha^{h-1}
+
\sum_{j=2}^{h}
T^{-j/2}\pi\Phi_1^{\,j-1}\alpha^{h-j}$,
where the sum is interpreted as zero for $h=1$. For fixed $h$, the
last sum is $O(T^{-1})$, since its leading term is
$T^{-1}\pi\Phi_1\alpha^{h-2}$ when $h\ge2$. Hence
$\sqrt T(\beta_{h,T}-\alpha^h)=\pi\alpha^{h-1}+O(T^{-1/2})$.
 
\noindent\textbf{Proof of \eqref{eq:var-illustration-nonaug}.} (Non-augmented VAR)
Using the preceding expansion of $\widetilde u_t$,
\begin{align}
\sqrt T(\widehat\alpha-\alpha)
&=
\frac{
T^{-1/2}\sum_t y_{t-1}u_t
+
\pi T^{-1}\sum_t y_{t-1}u_{t-1}
}{
T^{-1}\sum_t y_{t-1}^2
}
+
O_p(T^{-1/2}).
\label{eq:var-illustration-nonaug-expansion}
\end{align}
The sequence $\{y_{t-1}u_t\}$ is a martingale difference,
$T^{-1}\sum_t y_{t-1}u_{t-1}
=
\sigma^2+o_p(1)$, and
$T^{-1}\sum_t y_{t-1}^2
=
\sigma^2/(1-\alpha^2)+o_p(1)$.
The martingale CLT and the law of large numbers therefore give
$\sqrt T(\widehat\alpha-\alpha)
\xrightarrow{d}
N(\pi(1-\alpha^2),1-\alpha^2)$, and the delta method gives
$\sqrt T(\widehat\alpha^{\,h}-\alpha^h)
\xrightarrow{d}
N\bigl(h\alpha^{h-1}\pi(1-\alpha^2),
h^2\alpha^{2(h-1)}(1-\alpha^2)\bigr)$.
Subtracting $\sqrt T(\beta_{h,T}-\alpha^h)$ gives
\eqref{eq:var-illustration-nonaug}.
 
\noindent\textbf{Proof of \eqref{eq:var-illustration-aug}.} (Lag-augmented VAR)
From the local ARMA$(1,1)$ model,
$u_t=(1+\Theta_TL)^{-1}(1-\Phi_TL)(1-\alpha L)y_t$.
Expanding $(1+\Theta_TL)^{-1}$ gives
\begin{align}
y_t
&=
(\alpha+\Pi_T)y_{t-1}
-
\Pi_T(\alpha+\Theta_T)y_{t-2}
+
u_t
+
R_t.
\label{eq:var-illustration-ar2}
\end{align}
Here
$R_t=-\Pi_T(\alpha+\Theta_T)\sum_{j\ge3}(-\Theta_T)^{j-2}y_{t-j}$,
and $\|R_t\|_{L^2}=O(T^{-1})$ because $\Pi_T=O(T^{-1/2})$ and
$\Theta_T=O(T^{-1/2})$.

Let $(\phi_1^*,\phi_2^*)$ be the population coefficients from
projecting $y_t$ on $Y_{t,2}$, and let
$u_t^*=y_t-\phi_1^*y_{t-1}-\phi_2^*y_{t-2}$. Since
$\mathbb E[Y_{t,2}u_t]=0$, \eqref{eq:var-illustration-ar2} gives
$\phi_1^*=\alpha+\Pi_T
+e_1'(\mathbb E[Y_{t,2}Y_{t,2}'])^{-1}\mathbb E[Y_{t,2}R_t]$,
where $\mathbb E[Y_{t,2}R_t]=O(T^{-1})$ and
$(\mathbb E[Y_{t,2}Y_{t,2}'])^{-1}=O(1)$, so
$\sqrt T(\phi_1^*-\alpha)=\pi+O(T^{-1/2})$. The same two bounds give
$\|u_t^*-u_t\|_{L^2}=O(T^{-1})$ and hence
$T^{-1/2}\sum_tY_{t,2}(u_t^*-u_t)=O_p(T^{-1/2})$.

As $\Phi_T,\Theta_T\to0$,
$\sigma^2(\mathbb E[Y_{t,2}Y_{t,2}'])^{-1}
\to
\bigl(\begin{smallmatrix}1&-\alpha\\-\alpha&1\end{smallmatrix}\bigr)$,
so the martingale CLT gives
$\sqrt T(\widehat\phi_1-\phi_1^*)\xrightarrow{d}N(0,1)$ and
$\sqrt T(\widehat\phi_1-\alpha)\xrightarrow{d}N(\pi,1)$. The delta
method and
$\sqrt T(\beta_{h,T}-\alpha^h)=\pi\alpha^{h-1}+O(T^{-1/2})$ then give $\sqrt T(\widehat\phi_1^{\,h}-\alpha^h)\xrightarrow{d}
N\bigl(h\alpha^{h-1}\pi,\;h^2\alpha^{2(h-1)}\bigr)$ and hence
\eqref{eq:var-illustration-aug}:
\begin{align}
\sqrt T(\widehat\phi_1^{\,h}-\beta_{h,T})
&\xrightarrow{d}
N\bigl(\alpha^{h-1}\pi(h-1),\;h^2\alpha^{2(h-1)}\bigr).
\nonumber
\end{align}
\end{remark}
 
\begin{remark}[Relation to \citet{olea2024double}]
\label{rem:mpqw}
\citet{olea2024double} study a stationary local-to-SVAR model with an
MA$(\infty)$ departure with rate $T^{-\zeta}$. Their Proposition~3.1
shows that the omitted-variable bias of the LP impulse-response
estimator is of order $T^{-2\zeta}$, yielding conventional root-$T$
inference when $\zeta>1/4$. Their Proposition~3.2 characterizes the
order-$T^{-\zeta}$ bias of the VAR impulse-response estimator.

Our setting instead imposes the fixed-order factorization in
\eqref{eq:local_arma_error}, which implies
$\|\Pi_{j,T}\|=O(T^{-j\bar\zeta})$. This successive decay is what lag
augmentation exploits: with $p=d+s$, the approximation error is
$O_p(T^{-(s+1)\bar\zeta})$, so $s=1$ gives the $1/4$ threshold and
larger $s$ permits smaller $\bar\zeta$, allowing the misspecification to
vanish more slowly than under the $T^{-1/4}$ boundary. Theorem~\ref{theorem:local_mis_lalp}
also admits unit and near-unit roots.
\end{remark}

\subsection{Auxiliary lemmas}
\label{subsec:local-aux-lemmas}

\begin{lemma}[Local-filter approximation and persistence dichotomy]
\label{lem:local-filter-approximation}
\label{lem:rho-tau-dichotomy}
Suppose \eqref{localmis}--\eqref{eq:local_arma_error} and
Assumptions~\ref{ass:persistent} and \ref{ass:stable-score} hold. For
every fixed $j_0\geq1$, uniformly over $(D,B)\in\mathcal P$,
\begin{align}
&\sum_{j\geq j_0}\|\Pi_{j,T}\|_2
\le
CT^{-j_0\bar\zeta},
\label{eq:local-filter-coefficient-bounds}
\\
&\sup_{t\le T}\|\xi_t-B(L)^{-1}u_t\|_{L^{2r}}
\le
CT^{-\bar\zeta},
\label{eq:local-quasi-difference-bound}
\\
&\sup_{t\le T}\|\xi_t\|_{L^{4r+\delta}}
\le
C,
\label{eq:local-xi-moment}
\\
&\sup_{0\le t\le T}\|\Lambda_T^{-1/2}y_t\|_{L^{4r+\delta}}
\le
C.
\label{eq:scaled-level-moment}
\end{align}
Moreover, let $\tau_{k,T}=\sum_{\ell=0}^{T-1}\rho_k^{2\ell}$ and
suppose $\rho_k\to\rho_k^\ast\in[-1,1]$. If $|\rho_k^\ast|<1$, then
$\sup_T\tau_{k,T}<C<\infty$, whereas if $|\rho_k^\ast|=1$, then
$\tau_{k,T}\to\infty$.
\end{lemma}

\begin{lemma}[Uniform rotated moment bounds]
\label{lem:rotated-product-bounds}
Suppose \eqref{localmis}--\eqref{eq:local_arma_error} and
Assumptions~\ref{ass:persistent} and \ref{ass:stable-score} hold. For
any fixed integers $\ell,\ell'\ge0$ and $j$, uniformly over
$(D,B)\in\mathcal P$,
\begin{align}
&T^{-1}\sum_t
\xi_{t-\ell}y_{t-\ell'}'
\Lambda_T^{-1/2}
=
O_p(1),
\label{eq:level-level-bound-xi}
\\
&T^{-1}\sum_t
\Lambda_T^{-1/2}
y_{t-\ell}y_{t-\ell'}'
\Lambda_T^{-1/2}
=
O_p(1),
\label{eq:level-level-bound-y}
\\
&T^{-1}\sum_t
\xi_{t-\ell}y_{t-\ell'}'
\Lambda_T^{-1/2}
e_{t,h}e_{t+j,h}
=
O_p(1),
\label{eq:level-level-bound-xi-e}
\\
&T^{-1}\sum_t
\Lambda_T^{-1/2}
y_{t-\ell}y_{t-\ell'}'
\Lambda_T^{-1/2}
e_{t,h}e_{t+j,h}
=
O_p(1),
\label{eq:level-level-bound-y-e}
\\
&T^{-1/2}\sum_t
\xi_{t-\ell}e_{t,h}
=
O_p(1),
\label{eq:level-score-bound-xi}
\\
&T^{-1/2}\sum_t
\Lambda_T^{-1/2}y_{t-\ell}e_{t,h}
=
O_p(1).
\label{eq:level-score-bound-y}
\end{align}
Moreover,
\begin{align}
&\sum_{h<|j|\le b_T}
k(jb_T^{-1})
T^{-1}\sum_t
\xi_{t-\ell}e_{t,h}
y_{t-j-\ell'}'
\Lambda_T^{-1/2}e_{t-j,h}
=
O_p\left(b_TT^{-1/2}\right),
\label{eq:hac-level-tail-bound-xi}
\\
&\sum_{h<|j|\le b_T}
k(jb_T^{-1})
T^{-1}\sum_t
\Lambda_T^{-1/2}y_{t-\ell}e_{t,h}
y_{t-j-\ell'}'
\Lambda_T^{-1/2}e_{t-j,h}
=
O_p\left(b_TT^{-1/2}\right).
\label{eq:hac-level-tail-bound-y}
\end{align}
\end{lemma}

\begin{lemma}[Local-score equivalence and feasible HAC approximation]
\label{lem:local-score-equivalence}
\label{lem:hac-feasible}
Suppose \eqref{localmis}--\eqref{eq:local_arma_error} and
Assumptions~\ref{ass:persistent} and \ref{ass:stable-score} hold. Let
$s\geq0$, $h$ and $q\geq p=d+s>1$ be fixed, and suppose
$\bar\zeta>1/(2s+2)$. Then, uniformly over $(D,B)\in\mathcal P$,
\begin{align}
&T^{-1/2}
\sum_{t=q}^{T-h}
H_qY_{t,q}
\left(
e_{t,h,T}-e_{t,h}
\right)
=
o_p(1).
\label{eq:local-score-equivalence}
\end{align}
If, in addition, Assumption~\ref{ass:rotated-nonsingularity} holds,
then $H_q(\widehat\Omega_{q,h}-\widetilde\Omega_{q,h})H_q'=o_p(1)$,
where $\widetilde\Omega_{q,h,T}$ and $\widetilde\Omega_{q,h}$ are the
infeasible HAC matrices formed from $Y_{t,q}e_{t,h,T}$ and
$Y_{t,q}e_{t,h}$, respectively.
\end{lemma}

\begin{lemma}[Stable--persistent separation and stable-block approximation]
\label{lem:stable-persistent-separation}
\label{lem:stable-block-approximation}
Suppose \eqref{localmis}--\eqref{eq:local_arma_error} and
Assumptions~\ref{ass:persistent} and \ref{ass:stable-score} hold. Let
$h$ and $q\geq p$ be fixed, and consider any subsequence along
which $\rho_k\to\rho_k^\ast\in[-1,1]$ for every $k$.

\textup{(i)} It holds that
$\widecheck\Sigma_{p,q}^{\xi}-\Sigma_{p,q,T}^{\xi}=o_p(1)$,
$\|\Sigma_{p,q,T}^{\xi}\|_2+\|\Omega_{q,h,T}^{\xi}\|_2\le C$, and
\begin{align}
&\widecheck\Sigma_{p,q}
=
\begin{pmatrix}
\widecheck\Sigma_{p,q}^{\xi} & o_p(1)\\
o_p(1) & O_p(1)
\end{pmatrix}.
\label{eq:stable-persistent-sigma-separation}
\end{align}
If Assumption~\ref{ass:rotated-nonsingularity} holds, then
$\lambda_{\min}(\Sigma_{p,q,T}^{\xi})\ge c_\Sigma$ for sufficiently
large $T$.

\textup{(ii)} If $s\geq0$, $q\geq p=d+s>1$ and
$\bar\zeta>1/(2s+2)$, then
\begin{align}
&T^{-1/2}
\sum_{t=q}^{T-h}
H_qY_{t,q}e_{t,h,T}
=
\begin{pmatrix}
T^{-1/2}\displaystyle\sum_{t=q}^{T-h}
Y_{t,q}^{\xi}e_{t,h}+o_p(1)
\\[1mm]
O_p(1)
\end{pmatrix}.
\label{eq:stable-persistent-score-separation}
\end{align}

\textup{(iii)} If, in addition to the conditions in \textup{(ii)},
Assumption~\ref{ass:rotated-nonsingularity} holds, then
$\widecheck\Omega_{q,h}^{\xi}-\Omega_{q,h,T}^{\xi}=o_p(1)$,
$\lambda_{\min}(\Omega_{q,h,T}^{\xi})\ge c_\Omega$ for sufficiently
large $T$, and
\begin{align}
&\widecheck\Omega_{q,h}
=
\begin{pmatrix}
\widecheck\Omega_{q,h}^{\xi} & o_p(1)\\
o_p(1) & O_p(1)
\end{pmatrix},
\label{eq:stable-persistent-block-separation}
\\
&(\Omega_{q,h,T}^{\xi})^{-1/2}
T^{-1/2}
\sum_{t=q}^{T-h}
Y_{t,q}^{\xi}e_{t,h}
\xrightarrow{d}
N\left(0,I_{(q-1)K+|\mathcal I_0|}\right).
\label{eq:stable-block-score-clt}
\end{align}
\end{lemma}

\begin{lemma}[Stationary specialization]
\label{lem:stationary-regime}
Suppose \eqref{localmis}--\eqref{eq:local_arma_error} and
Assumptions~\ref{ass:persistent} and \ref{ass:stable-score} hold. Let
$h$ and $q\geq p$ be fixed and suppose
$\max_{1\le k\le K}|\rho_k|\le1-c$ for some fixed $c>0$. Then
\begin{align}
\widehat\Sigma_{p,q}-\Sigma_{p,q,T}
&=
o_p(1).
\end{align}
Moreover, $\|\Sigma_{p,q,T}\|_2\le C$. If Assumption~\ref{ass:rotated-nonsingularity} holds, then
$\lambda_{\min}(\Sigma_{p,p,T})\ge c_\Sigma$. If, in addition,
$s\geq0$, $q\geq p=d+s>1$ and $\bar\zeta>1/(2s+2)$, then
$\|\Omega_{q,h,T}\|_2\le C$,
$\lambda_{\min}(\Omega_{q,h,T})\ge c_\Omega$, and
\begin{align}
&\widehat\Omega_{q,h}-\Omega_{q,h,T}
=
o_p(1),
\label{eq:stationary-hac-limit}
\\
&\Omega_{q,h,T}^{-1/2}
T^{-1/2}
\sum_{t=q}^{T-h}
Y_{t,q}e_{t,h,T}
\xrightarrow{d}
N(0,I_{qK}).
\label{eq:stationary-score-clt}
\end{align}
\end{lemma}

\subsection{Proofs of auxiliary lemmas}
\label{subsec:local-aux-proofs}

\subsubsection{Proof of Lemma~\ref{lem:local-filter-approximation}}

\noindent
\emph{Proof of \eqref{eq:local-filter-coefficient-bounds}.}
For all sufficiently large $T$,
$(I_K-\Phi_{j,T}L)^{-1}
=\sum_{a\geq0}\Phi_{j,T}^{\,a}L^a$, $j=1,\ldots,m$.
Each term in $\Pi_{j,T}$ contains exactly $j$ local matrices, each with
norm at most $CT^{-\bar\zeta}$, and the number of such terms is polynomial
in $j$. Hence
$\|\Pi_{j,T}\|_2\le C(1+j)^{m+n}(CT^{-\bar\zeta})^j$, and summing the
geometric tail gives
$\sum_{j\geq j_0}\|\Pi_{j,T}\|_2\le CT^{-j_0\bar\zeta}$.

\medskip
\noindent
\emph{Proof of \eqref{eq:local-quasi-difference-bound},
\eqref{eq:local-xi-moment}, and the NED property.}
Since $A(L)=B(L)(I_K-DL)$ and $\xi_t=(I_K-DL)y_t$,
\begin{align}
\xi_t
&=
B(L)^{-1}u_t
+
B(L)^{-1}\sum_{j\geq1}\Pi_{j,T}u_{t-j}
=
\sum_{j\geq0}C_{j,T}u_{t-j}.
\label{eq:aux-local-xi-representation}
\end{align}
Uniform geometric summability of $B(L)^{-1}$, stability of the local
autoregressive inverse factors for all sufficiently large $T$, and fixed
$m,n$ give, for some $\varrho\in(0,1)$,
$\sup_T\|C_{j,T}\|_2\le C\varrho^j$ and
$\sup_T\sum_{j>\ell_0}\|C_{j,T}\|_2\le C\varrho^{\ell_0}$.
Minkowski's inequality and Assumption~\ref{ass:stable-score}(ii) then
give $\sup_{t\le T}\|\xi_t\|_{L^{4r+\delta}}\le C$, while truncating
\eqref{eq:aux-local-xi-representation} after $\ell_0$ lags gives the geometric
$L_{4r+\delta}$-NED bound. Finally,
$\sup_{t\le T}\|\xi_t-B(L)^{-1}u_t\|_{L^{2r}}
\le C\sum_{j\geq1}\|\Pi_{j,T}\|_2\le CT^{-\bar\zeta}$.

\medskip
\noindent
\emph{Proof of \eqref{eq:scaled-level-moment}.}
From $\xi_t=(I_K-DL)y_t$,
\begin{align}
y_t
&=
D^ty_0+\sum_{\ell=0}^{t-1}D^\ell\xi_{t-\ell}.
\label{eq:local-level-linear-representation}
\end{align}
For each $k$, rearranging \eqref{eq:aux-local-xi-representation} as a
weighted sum of the martingale differences and applying the
Burkholder--Rosenthal inequality%
\footnote{Burkholder--Rosenthal
inequality states the following fact. For deterministic coefficients $\{a_j\}$ and a
martingale-difference sequence $\{u_j\}$ with uniformly bounded
$(4r+\delta)$th moments,
$\mathbb E|\sum_j a_j'u_j|^{4r+\delta}
\le C(\sum_j\|a_j\|^2)^{(4r+\delta)/2}$.
See, e.g., \citet[Theorem~2.10]{hall1980martingale}.}
gives, by Assumption~\ref{ass:stable-score}(ii),
$\mathbb E|\sum_{\ell=0}^{t-1}\rho_k^\ell\xi_{k,t-\ell}|^{4r+\delta}
\le C\tau_{k,T}^{(4r+\delta)/2}$.
Thus
$\|\tau_{k,T}^{-1/2}\sum_{\ell=0}^{t-1}
\rho_k^\ell\xi_{k,t-\ell}\|_{L^{4r+\delta}}\le C$.
Since $D$ and $\Lambda_T$ are diagonal and $|\rho_k|\le1$,
$\|\Lambda_T^{-1/2}D^ty_0\|\le\|\Lambda_T^{-1/2}y_0\|$, so
\eqref{eq:local-level-linear-representation} and
Assumption~\ref{ass:stable-score}(iii) yield
$\sup_{0\le t\le T}\|\Lambda_T^{-1/2}y_t\|_{L^{4r+\delta}}\le C$.

\medskip
\noindent
\emph{Proof of the persistence dichotomy.}
If $|\rho_k^\ast|<1$, then $|\rho_k|\le1-c$ for all sufficiently large
$T$, so $\tau_{k,T}\le\{1-(1-c)^2\}^{-1}$. If
$|\rho_k^\ast|=1$, then, for every fixed $\ell_0$ and $T>\ell_0$,
$\tau_{k,T}\ge\sum_{\ell=0}^{\ell_0}\rho_k^{2\ell}\to \ell_0+1$; hence
$\tau_{k,T}\to\infty$.

\subsubsection{Proof of Lemma~\ref{lem:rotated-product-bounds}}

\noindent
\emph{Proof of \eqref{eq:level-level-bound-xi}--\eqref{eq:level-level-bound-y-e}.}
Since $e_{t,h}$ is a fixed linear combination of
$u_{t+1},\ldots,u_{t+h}$,
$\sup_{t\le T}\|e_{t,h}\|_{L^{4r+\delta}}\le C$.
Together with Lemma~\ref{lem:local-filter-approximation} and H\"older's
inequality,
\begin{align}
\begin{split}
\mathbb E\|\xi_{t-\ell}y_{t-\ell'}'\Lambda_T^{-1/2}\|
&\le C,\\
\mathbb E\|\Lambda_T^{-1/2}y_{t-\ell}y_{t-\ell'}'
\Lambda_T^{-1/2}\|
&\le C,\\
\mathbb E\|\xi_{t-\ell}y_{t-\ell'}'\Lambda_T^{-1/2}
e_{t,h}e_{t+j,h}\|
&\le C,\\
\mathbb E\|\Lambda_T^{-1/2}y_{t-\ell}y_{t-\ell'}'
\Lambda_T^{-1/2}e_{t,h}e_{t+j,h}\|
&\le C.
\end{split}
\end{align}
The triangle inequality and Markov's inequality therefore imply that
each corresponding sample average is $O_p(1)$, which gives
\eqref{eq:level-level-bound-xi}--\eqref{eq:level-level-bound-y-e}.

\medskip
\noindent
\emph{Proof of \eqref{eq:level-score-bound-xi} and
\eqref{eq:level-score-bound-y}.}
Both $\xi_{t-\ell}$ and $\Lambda_T^{-1/2}y_{t-\ell}$ are
$\mathcal F_t$-measurable and
$\mathbb E(e_{t,h}\mid\mathcal F_t)=0$.
For $t>s+h$,
\begin{align}
\mathbb E\!\left[
\xi_{t-\ell}e_{t,h}e_{s,h}\xi_{s-\ell}'
\right]
&=
\mathbb E\!\left[
\mathbb E\!\left(
\xi_{t-\ell}e_{t,h}\mid\mathcal F_t
\right)
e_{s,h}\xi_{s-\ell}'
\right]
=0,
\end{align}
and the same argument applies with
$\xi_{t-\ell}$ replaced by $\Lambda_T^{-1/2}y_{t-\ell}$.
Lemma~\ref{lem:local-filter-approximation} and H\"older's inequality
give uniformly bounded second moments. Since only $|t-s|\le h$ terms
remain,
\begin{align}
\mathbb E\left\|\sum_t\xi_{t-\ell}e_{t,h}\right\|^2
+
\mathbb E\left\|\sum_t\Lambda_T^{-1/2}y_{t-\ell}e_{t,h}\right\|^2
\le CT.
\end{align}
Chebyshev's inequality gives
\eqref{eq:level-score-bound-xi} and
\eqref{eq:level-score-bound-y}.

\medskip
\noindent
\emph{Proof of \eqref{eq:hac-level-tail-bound-xi} and
\eqref{eq:hac-level-tail-bound-y}.}
Fix $j>h$. Since $\mathbb E(e_{t,h}\mid\mathcal F_t)=0$ and every other
factor is $\mathcal F_t$-measurable, each summand has mean zero, and
summands dated more than $h$ apart are uncorrelated. Hence, by
\eqref{eq:local-xi-moment}, \eqref{eq:scaled-level-moment} and
H\"older's inequality,
\begin{align}
&\mathbb E
\left\|
T^{-1}\sum_t
\xi_{t-\ell}e_{t,h}
y_{t-j-\ell'}'\Lambda_T^{-1/2}e_{t-j,h}
\right\|^2
\le
CT^{-1}
\end{align}
uniformly over $h<j\le b_T$, and the same bound holds with
$\xi_{t-\ell}$ replaced by $\Lambda_T^{-1/2}y_{t-\ell}$. The
Cauchy--Schwarz inequality therefore bounds the corresponding first
moments by $CT^{-1/2}$, and the case $j<-h$ follows by reindexing and
transposition. Summing over $h<|j|\le b_T$ with bounded kernel weights
and applying Markov's inequality proves
\eqref{eq:hac-level-tail-bound-xi} and
\eqref{eq:hac-level-tail-bound-y}.

\subsubsection{Proof of Lemma~\ref{lem:local-score-equivalence}}

By \eqref{eq:local-error-decomposition},
\begin{align}
e_{t,h,T}-e_{t,h}
&=
\nu_1'
\sum_{j=1}^{\min\{s,h-1\}}
\sum_{\ell=0}^{h-j-1}
\Psi_\ell\Pi_{j,T}u_{t+h-\ell-j}
+
\mathcal E_{t,h,T}.
\label{eq:aux-local-error-gap}
\end{align}
For $s=0$,
$\sum_{j=1}^{\min\{s,h-1\}}
\sum_{\ell=0}^{h-j-1}
\Psi_\ell\Pi_{j,T}u_{t+h-\ell-j}=0$.
Moreover,
\eqref{eq:local-error-remainder} and
\eqref{eq:local-filter-coefficient-bounds} give
$\sup_{t\le T}\|\mathcal E_{t,h,T}\|_{L^{2r}}
\le CT^{-(s+1)\bar\zeta}$.
Also $\sup_{t\le T}\|H_qY_{t,q}\|_{L^{2r}}\le C$.
Since $H_qY_{t,q}$ is $\mathcal F_t$-measurable, the martingale moment
inequality and \eqref{eq:local-filter-coefficient-bounds} give
\begin{align}
T^{-1/2}\sum_t
H_qY_{t,q}\nu_1'
\sum_{j=1}^{\min\{s,h-1\}}
\sum_{\ell=0}^{h-j-1}
\Psi_\ell\Pi_{j,T}u_{t+h-\ell-j}
&=
O_p(T^{-\bar\zeta}),
\label{eq:aux-local-first-order-score}
\\
\left\|
T^{-1/2}\sum_tH_qY_{t,q}\mathcal E_{t,h,T}
\right\|_{L^1}
&\le CT^{1/2-(s+1)\bar\zeta}=o(1).
\label{eq:aux-local-second-order-score}
\end{align}
Markov's inequality proves \eqref{eq:local-score-equivalence}.

For the HAC approximation, Proposition~\ref{prop:local-lalp-representation}
gives, for $s\geq1$,
\begin{align}
(H_p')^{-1}
\left[
\widehat\theta_h^{\operatorname{LALP}}
-
\begin{pmatrix}
\beta_{h,T}\\
\eta_{h,p,T}
\end{pmatrix}
\right]
&=
\widecheck\Sigma_{p,p}^{-1}
T^{-1}\sum_tH_pY_{t,p}e_{t,h,T}
=
O_p(T^{-1/2}),
\label{eq:aux-local-ols-expansion}
\end{align}
by \eqref{eq:local-score-equivalence},
Lemma~\ref{lem:rotated-product-bounds}, and
Assumption~\ref{ass:rotated-nonsingularity}. For $s=0$, replace
$(\beta_{h,T},\eta_{h,p,T})$ by $(\beta_h,\eta_{h,d})$.
Expanding around $e_{t,h,T}$ gives
$\|H_q(\widehat\Omega_{q,h}-\widetilde\Omega_{q,h,T})H_q'\|_2
=O_p(b_TT^{-1/2})+O_p(b_TT^{-1})=o_p(1)$.
Using \eqref{eq:aux-local-error-gap} and the bound on
$\mathcal E_{t,h,T}$ gives
$\|H_q(\widetilde\Omega_{q,h,T}-\widetilde\Omega_{q,h})H_q'\|_2
=O_p(T^{-\bar\zeta})
+O_p(b_TT^{-\bar\zeta-1/2})
+O_p(b_TT^{-(s+1)\bar\zeta})=o_p(1)$.
Here $b_T=o(T^{1/2})$ and $(s+1)\bar\zeta>1/2$.

\subsubsection{Proof of Lemma~\ref{lem:stable-persistent-separation}}

By Lemma~\ref{lem:rho-tau-dichotomy},
$\|\Lambda_{1,T}^{-1/2}\|_2\to0$.
For $k\in\mathcal I_0$, the corresponding powers of $D$ decay
geometrically. Equation~\eqref{eq:local-level-linear-representation},
Assumption~\ref{ass:stable-score}(iii), and
Lemma~\ref{lem:local-filter-approximation} imply that the coordinates of
$Y_{t,p}^{\xi}$ and $Y_{t,q}^{\xi}$ are uniformly geometrically NED on
$\{u_t\}$, up to initial-condition terms whose normalized contribution is
$o_p(1)$.

Fix a coordinate $i$ of $Y_{t,p}^{\xi}$ and $k\in\mathcal I_1$; the same
argument applies with $p$ replaced by $q$. By
\eqref{eq:local-level-linear-representation},
\begin{align}
y_{k,t-b}
&=
\rho_k^{\,t-b}y_{k,0}
+
\sum_{\ell=0}^{t-b-1}\rho_k^\ell\xi_{k,t-b-\ell}.
\label{eq:aux-persistent-level-expansion}
\end{align}
Geometric NED and Assumption~\ref{ass:stable-score}(ii) give, for fixed
$b$, $\sup_{t,T}|\mathbb E[(Y_{t,p}^{\xi})_i y_{k,t-b}]|\le C$ and\\
$\mathbb E|T^{-1}\sum_t\{(Y_{t,p}^{\xi})_i y_{k,t-b}
-\mathbb E[(Y_{t,p}^{\xi})_i y_{k,t-b}]\}|^2
\le CT^{-1}\tau_{k,T}$. Hence
\begin{align}
T^{-1}\sum_t
(Y_{t,p}^{\xi})_i\,y_{k,t-b}\tau_{k,T}^{-1/2}
&=
O(\tau_{k,T}^{-1/2})+O_p(T^{-1/2})
=
o_p(1).
\label{eq:aux-stable-persistent-cross-generic}
\end{align}
The same bound holds for fixed-lag products multiplied by
$e_{t,h}e_{t-j,h}$.

\noindent
\emph{Proof of \textup{(i)}.}
The off-diagonal blocks in
\eqref{eq:stable-persistent-sigma-separation} are $o_p(1)$ by
\eqref{eq:aux-stable-persistent-cross-generic}; the lower-right block is
$O_p(1)$ by \eqref{eq:level-level-bound-y}.
For the upper-left block,
\citet[Theorem~2]{andrews1988laws} applies entrywise to
$Y_{t,p}^{\xi}(Y_{t,q}^{\xi})'$, giving
$\widecheck\Sigma_{p,q}^{\xi}-\Sigma_{p,q,T}^{\xi}=o_p(1)$.
The moment bounds give
$\|\Sigma_{p,q,T}^{\xi}\|_2+\|\Omega_{q,h,T}^{\xi}\|_2\le C$.

Under Assumption~\ref{ass:rotated-nonsingularity}, applying the same limit
with $q=p$, Cauchy's interlacing theorem, and Weyl's inequality gives
$\lambda_{\min}(\Sigma_{p,p,T}^{\xi})\ge c_\Sigma$ for all sufficiently
large $T$. Since $Y_{t,p}^{\xi}$ is an exact uniformly bounded linear
transformation of $Y_{t,q}^{\xi}$, the singular-value inequality gives,
after reducing $c_\Sigma$ if necessary,
$\lambda_{\min}(\Sigma_{p,q,T}^{\xi})\ge c_\Sigma$.

\medskip
\noindent
\emph{Proof of \textup{(ii)}.}
Lemma~\ref{lem:local-score-equivalence} gives
$T^{-1/2}\sum_tH_qY_{t,q}e_{t,h,T}
=T^{-1/2}\sum_tH_qY_{t,q}e_{t,h}+o_p(1)$.
Under \eqref{eq:aux-rotated-partition}, the upper block is
$T^{-1/2}\sum_tY_{t,q}^{\xi}e_{t,h}$ and the lower block is $O_p(1)$ by
\eqref{eq:level-score-bound-y}.

\medskip
\noindent
\emph{Proof of \textup{(iii)}.}
The off-diagonal blocks of $H_q\widetilde\Omega_{q,h}H_q'$ are
kernel-weighted sums over $|j|\le b_T$ of the terms below. For
$|j|\le h$, \eqref{eq:aux-stable-persistent-cross-generic} gives
\begin{align}
&T^{-1}\sum_t
Y_{t,q}^{\xi}e_{t,h}
(y_{t-j-q+1}^{(1)})'
\Lambda_{1,T}^{-1/2}e_{t-j,h}
=
o_p(1),
\end{align}
and for $h<|j|\le b_T$, \eqref{eq:hac-level-tail-bound-xi} and
\eqref{eq:hac-level-tail-bound-y} give
\begin{align}
&\sum_{h<|j|\le b_T}
k(jb_T^{-1})
T^{-1}\sum_t
Y_{t,q}^{\xi}e_{t,h}
(y_{t-j-q+1}^{(1)})'
\Lambda_{1,T}^{-1/2}e_{t-j,h}
=
O_p(b_TT^{-1/2})
=
o_p(1).
\end{align}
The lower-right block is $O_p(1)$ by
\eqref{eq:level-level-bound-y-e} and
\eqref{eq:hac-level-tail-bound-y}. Together with
Lemma~\ref{lem:hac-feasible}, these bounds prove
\eqref{eq:stable-persistent-block-separation}.

Let $\widetilde\Omega_{q,h}^{\xi}$ be the upper-left block of
$H_q\widetilde\Omega_{q,h}H_q'$. Then
$\widecheck\Omega_{q,h}^{\xi}-\Omega_{q,h,T}^{\xi}
=(\widecheck\Omega_{q,h}^{\xi}-\widetilde\Omega_{q,h}^{\xi})
+(\widetilde\Omega_{q,h}^{\xi}-\Omega_{q,h,T}^{\xi})$.
Lemma~\ref{lem:hac-feasible} gives
$\widecheck\Omega_{q,h}^{\xi}
-\widetilde\Omega_{q,h}^{\xi}=o_p(1)$.
Also
$\sup_{t\le T}\|Y_{t,q}^{\xi}e_{t,h}\|_{L^r}\le C$, and
$Y_{t,q}^{\xi}e_{t,h}$ is uniformly geometrically $L_r$-NED on
$\{u_t\}$. Hence
\citet[Theorem~2.1]{dejong2000consistency} gives
$\widetilde\Omega_{q,h}^{\xi}
-\Omega_{q,h,T}^{\xi}=o_p(1)$.

Cauchy's interlacing theorem gives
$\lambda_{\min}(\widecheck\Omega_{q,h}^{\xi})
\ge\lambda_{\min}(\widecheck\Omega_{q,h})$.
If $\lambda_{\min}(\Omega_{q,h,T}^{\xi})$ were not bounded away from zero,
there would exist a subsequence along which it converges to zero.
The preceding limit and Weyl's inequality would then contradict
Assumption~\ref{ass:rotated-nonsingularity}. Thus
$\lambda_{\min}(\Omega_{q,h,T}^{\xi})\ge c_\Omega$ for all sufficiently
large $T$.

Finally, $\mathbb E[Y_{t,q}^{\xi}e_{t,h}]=0$. For nonzero $\nu$ and
$0<\delta_0<r-2$, the lower bound on
$\lambda_{\min}(\Omega_{q,h,T}^{\xi})$ gives
\begin{align}
&T^{-1-\delta_0/2}\sum_{t=q}^{T-h}
\mathbb E
\left|
\nu'(\Omega_{q,h,T}^{\xi})^{-1/2}
Y_{t,q}^{\xi}e_{t,h}
\right|^{2+\delta_0}
\le
CT^{-\delta_0/2}\|\nu\|^{2+\delta_0}
\to0.
\end{align}
The geometric NED bounds and
Assumption~\ref{ass:stable-score}(ii) therefore satisfy
\citet[Corollary~1]{dejong1997central}. The Cram\'er--Wold device gives
\eqref{eq:stable-block-score-clt}.

\subsubsection{Proof of Lemma~\ref{lem:stationary-regime}}

Under $\max_k|\rho_k|\le1-c$, $\Lambda_T$, $H_p$, $H_q$, and their
inverses are uniformly bounded. Lemma~\ref{lem:local-filter-approximation}
and the stable level recursion imply uniform geometric NED for the
relevant level and quasi-difference processes, while the initial-condition
terms decay geometrically.

For the Gram matrix, \citet[Theorem~2]{andrews1988laws} applies entrywise,
so $\widehat\Sigma_{p,q}-\Sigma_{p,q,T}=o_p(1)$ and
$\|\Sigma_{p,q,T}\|_2\le C$. Under
Assumption~\ref{ass:rotated-nonsingularity}, the same result with $q=p$,
the boundedness of $H_p,H_p^{-1}$, and Weyl's inequality give
$\lambda_{\min}(\Sigma_{p,p,T})\ge c_\Sigma$ for all sufficiently large
$T$.

For \eqref{eq:stationary-hac-limit},
$\widehat\Omega_{q,h}-\Omega_{q,h,T}
=(\widehat\Omega_{q,h}-\widetilde\Omega_{q,h,T})
+(\widetilde\Omega_{q,h,T}-\Omega_{q,h,T})$. The proof of Lemma~\ref{lem:hac-feasible} gives
$\widehat\Omega_{q,h}-\widetilde\Omega_{q,h,T}=o_p(1)$, while
\citet[Theorem~2.1]{dejong2000consistency} gives
$\widetilde\Omega_{q,h,T}-\Omega_{q,h,T}=o_p(1)$. If
$\lambda_{\min}(\Omega_{q,h,T})$ were not uniformly bounded away from
zero, then \eqref{eq:stationary-hac-limit}, boundedness of
$H_q,H_q^{-1}$, and Weyl's inequality
would contradict Assumption~\ref{ass:rotated-nonsingularity}. Hence
$\lambda_{\min}(\Omega_{q,h,T})\ge c_\Omega$ for all sufficiently large
$T$.

The explicit terms in \eqref{eq:aux-local-error-gap} have zero
conditional mean, while
\begin{align}
\left\|
T^{-1/2}\sum_t
\mathbb E[Y_{t,q}e_{t,h,T}]
\right\|_2
&\le
CT^{1/2-(s+1)\bar\zeta}
=
o(1).
\end{align}
For nonzero $\nu\in\mathbb R^{qK}$ and $0<\delta_0<r-2$, the lower
bound on $\lambda_{\min}(\Omega_{q,h,T})$ gives
\begin{align}
&T^{-1-\delta_0/2}\sum_t
\mathbb E
\left|
\nu'
\Omega_{q,h,T}^{-1/2}
\left\{
Y_{t,q}e_{t,h,T}
-
\mathbb E[Y_{t,q}e_{t,h,T}]
\right\}
\right|^{2+\delta_0}
\le
CT^{-\delta_0/2}
\to0.
\end{align}
Geometric NED and Assumption~\ref{ass:stable-score}(ii) then satisfy
\citet[Corollary~1]{dejong1997central}. Together with the preceding
mean bound, the Cram\'er--Wold device gives
\eqref{eq:stationary-score-clt}.

\subsection{Proof of Theorem~\ref{thm:persistence-uniform-gmm}}
\label{sec:uniform-theorem-proof}

\begin{proof}
Let $s=p-d\geq0$. Since
$p\geq\max\{d,2\}$, $p=d+s>1$.
Set $\Phi_j=0$ and $\Theta_k=0$ for all $j$ and $k$ in
\eqref{eq:local_arma_error}, and take any
$\bar\zeta>1/(2s+2)$. Then
$\widetilde u_t=u_t$, $\Pi_{j,T}=0$ for every $j\geq1$,
$\beta_{h,T}=\beta_h$, and $e_{t,h,T}=e_{t,h}$.
Theorem~\ref{theorem:local_mis_lalp} therefore applies directly.
Hence $W_{T,h}\xrightarrow{d}\chi_g^2$ along every parameter sequence
in $\mathcal P$. The sequential characterization of uniform
convergence, continuity of $F_{\chi_g^2}$, and P\'olya's theorem yield
\eqref{eq:uniform-wald-convergence}.
\end{proof}

\vspace{0.5\baselineskip}

\section{Assumption~\ref{ass:rotated-nonsingularity} under local misspecification}
\label{app:local-rotated-nonsingularity}

This section illustrates Assumption~\ref{ass:rotated-nonsingularity}
with an AR$(1)$ model under local ARMA$(1,1)$ error misspecification.
Consider
\begin{align}
y_t
&=
\rho_Ty_{t-1}+\widetilde u_t.
\label{eq:local-nonsing-ar1}
\end{align}
Let $\rho_T\in[-1,1]$ and
$(1-\phi_TL)\widetilde u_t=(1+\theta_TL)u_t$, where
$\phi_T=T^{-\zeta_{AR}}\phi$,
$\theta_T=T^{-\zeta_{MA}}\theta$, and
$\bar\zeta=\min\{\zeta_{AR},\zeta_{MA}\}>1/4$.
Let $\{u_t\}_{t\in\mathbb Z}$ be i.i.d.\ with
$\mathbb E u_t=0$, $\mathbb E u_t^2=\sigma^2>0$, and sufficiently
many finite moments, and set $y_0=0$.

Let $\tau_T=\sum_{\ell=0}^{T-1}\rho_T^{2\ell}$ and
$\xi_t=(1-\rho_TL)y_t=\widetilde u_t$. By
\eqref{eq:rotated-vector},
$H_2Y_{t,2}=(\xi_t,\tau_T^{-1/2}y_{t-1})'$ and
$H_3Y_{t,3}=(\xi_t,\xi_{t-1},\tau_T^{-1/2}y_{t-2})'$;
$\widecheck\Sigma_{2,2}
=T^{-1}\sum_tH_2Y_{t,2}Y_{t,2}'H_2'$ and
$\widecheck\Omega_{3,h}
=\sum_{|j|\le b_T}K(j/b_T)T^{-1}\sum_t
(H_3Y_{t,3}\widehat e_{t,h})
(H_3Y_{t-j,3}\widehat e_{t-j,h})'$.
All stochastic orders below are uniform over $\rho_T\in[-1,1]$.

We show that, for every $\varepsilon>0$, there exists $c>0$ such that
\begin{align}
\liminf_{T\to\infty}
\inf_{\rho_T\in[-1,1]}
\mathbb P_{\rho_T}
\left\{
\lambda_{\min}(\widecheck\Sigma_{2,2})\ge c,\;
\lambda_{\min}(\widecheck\Omega_{3,h})\ge c
\right\}
&\ge
1-\varepsilon.
\label{eq:local-nonsing-conclusion}
\end{align}

Let
$y_t^\dagger=\rho_Ty_{t-1}^\dagger+u_t$ with $y_0^\dagger=0$.
Then
$H_2Y_{t,2}^\dagger=(u_t,\tau_T^{-1/2}y_{t-1}^\dagger)'$,
$H_3Y_{t,3}^\dagger
=(u_t,u_{t-1},\tau_T^{-1/2}y_{t-2}^\dagger)'$,
$\widecheck\Sigma_{2,2}^\dagger
=T^{-1}\sum_tH_2Y_{t,2}^\dagger
Y_{t,2}^{\dagger\prime}H_2'$, and
$\widecheck\Omega_{3,h}^\dagger
=\sum_{|j|\le b_T}K(j/b_T)T^{-1}\sum_t
(H_3Y_{t,3}^\dagger\widehat e_{t,h}^\dagger)
(H_3Y_{t-j,3}^\dagger\widehat e_{t-j,h}^\dagger)'$.

We establish \eqref{eq:local-nonsing-conclusion} from the following
three results. For every $\varepsilon>0$, there exists $\bar c>0$ such that
\begin{align}
\liminf_{T\to\infty}
\inf_{\rho_T\in[-1,1]}
\mathbb P_{\rho_T}
\left\{
\lambda_{\min}(\widecheck\Sigma_{2,2}^\dagger)\ge\bar c,\;
\lambda_{\min}(\widecheck\Omega_{3,h}^\dagger)\ge\bar c
\right\}
&\ge
1-\varepsilon.
\label{eq:local-nonsing-corr}
\end{align}
Moreover,
\begin{align}
\left\|
\widecheck\Sigma_{2,2}
-
\widecheck\Sigma_{2,2}^\dagger
\right\|_2
&=
O_p(T^{-\bar\zeta})
=
o_p(1).
\label{eq:local-nonsing-sigma-gap}
\end{align}
Finally,
\begin{align}
\left\|
\widecheck\Omega_{3,h}
-
\widecheck\Omega_{3,h}^\dagger
\right\|_2
&=
o_p(1).
\label{eq:local-nonsing-omega-gap}
\end{align}

\textbf{\textit{Proof of \eqref{eq:local-nonsing-conclusion}}.}
Fix $\varepsilon>0$ and let $\bar c>0$ satisfy
\eqref{eq:local-nonsing-corr}. By
\eqref{eq:local-nonsing-sigma-gap} and
\eqref{eq:local-nonsing-omega-gap},
$\|\widecheck\Sigma_{2,2}-\widecheck\Sigma_{2,2}^\dagger\|_2
\le\bar c/2$ and
$\|\widecheck\Omega_{3,h}-\widecheck\Omega_{3,h}^\dagger\|_2
\le\bar c/2$
with probability $1-o(1)$ uniformly over $\rho_T$.
Weyl's inequality gives
$\lambda_{\min}(\widecheck\Sigma_{2,2})
\ge\lambda_{\min}(\widecheck\Sigma_{2,2}^\dagger)
-\|\widecheck\Sigma_{2,2}-\widecheck\Sigma_{2,2}^\dagger\|_2$
and
$\lambda_{\min}(\widecheck\Omega_{3,h})
\ge\lambda_{\min}(\widecheck\Omega_{3,h}^\dagger)
-\|\widecheck\Omega_{3,h}-\widecheck\Omega_{3,h}^\dagger\|_2$.
Hence both minimum eigenvalues are at least $\bar c/2$ with
asymptotic probability at least $1-\varepsilon-o(1)$ uniformly over
$\rho_T$. Relabeling $\bar c/2$ as $c$ proves
\eqref{eq:local-nonsing-conclusion}.

It remains to establish
\eqref{eq:local-nonsing-corr}--\eqref{eq:local-nonsing-omega-gap}.

\textbf{\textit{Proof of \eqref{eq:local-nonsing-corr}}.}

\noindent\textbf{Lower bound for $\widecheck\Sigma_{2,2}^\dagger$.}
\begin{align}
\widecheck\Sigma_{2,2}^\dagger
&=
\begin{pmatrix}
T^{-1}\sum_tu_t^2
&
(T\tau_T^{1/2})^{-1}\sum_tu_ty_{t-1}^\dagger
\\
(T\tau_T^{1/2})^{-1}\sum_tu_ty_{t-1}^\dagger
&
(T\tau_T)^{-1}\sum_t(y_{t-1}^\dagger)^2
\end{pmatrix}.
\label{eq:local-nonsing-reference-sigma}
\end{align}
Since $\mathbb E u_t^4<\infty$,
$\operatorname{Var}(T^{-1}\sum_tu_t^2)=O(T^{-1})$, so
$T^{-1}\sum_tu_t^2\xrightarrow{p}\sigma^2$ by Chebyshev's inequality.
Since $u_ty_{t-1}^\dagger$ is a martingale difference and
$\mathbb E(y_{t-1}^\dagger)^2\le\sigma^2\tau_T$,
\begin{align}
\mathbb E\left[
\left(
\frac{1}{T\tau_T^{1/2}}
\sum_tu_ty_{t-1}^\dagger
\right)^2
\right]
&\le
\frac{C}{T}.
\end{align}
Hence
$(T\tau_T^{1/2})^{-1}\sum_tu_ty_{t-1}^\dagger
=O_p(T^{-1/2})$
uniformly over $\rho_T$.

Consider
$(T\tau_T)^{-1}\sum_t(y_{t-1}^\dagger)^2$.
Take an arbitrary sequence $\{\rho_T\}\subset[-1,1]$ and a subsequence
along which $T(1-|\rho_T|)$ converges in $[0,\infty]$, taking a further
subsequence with constant sign when necessary.

If $T(1-|\rho_T|)\to\infty$, then
$|\rho_T|^{2T}\to0$ and
$\tau_T(1-\rho_T^2)=1-\rho_T^{2T}\to1$.
Since
$\mathbb E(y_{t-1}^\dagger)^2
=\sigma^2\sum_{\ell=0}^{t-2}\rho_T^{2\ell}$,
direct summation gives
$(T\tau_T)^{-1}\sum_t\mathbb E(y_{t-1}^\dagger)^2\to\sigma^2$.
For $s\le t$, independence of future innovations gives
\begin{align}
\operatorname{Cov}
\left(
(y_{s-1}^\dagger)^2,
(y_{t-1}^\dagger)^2
\right)
&=
\rho_T^{2(t-s)}
\operatorname{Var}
\left(
(y_{s-1}^\dagger)^2
\right).
\end{align}
Since $\mathbb E u_t^4<\infty$,
$\operatorname{Var}((y_{s-1}^\dagger)^2)\le C\tau_T^2$.
Summing the geometric covariance bound gives
\begin{align}
\operatorname{Var}\left(
\frac{1}{T\tau_T}\sum_t(y_{t-1}^\dagger)^2
\right)
&\le
\frac{C}{T(1-\rho_T^2)}
=
o(1).
\end{align}
Thus
$(T\tau_T)^{-1}\sum_t(y_{t-1}^\dagger)^2
\xrightarrow{p}\sigma^2$
by Chebyshev's inequality.

If $T(1-|\rho_T|)\to a\in[0,\infty)$ and $\rho_T\to1$, then
$\tau_T/T\to\kappa_a=\int_0^1e^{-2ar}\,dr$, and
\citet{phillips1987towards} gives
\begin{align}
\frac{1}{T\tau_T}\sum_t(y_{t-1}^\dagger)^2
&\xrightarrow{d}
\frac{\sigma^2}{\kappa_a}
\int_0^1
\left(
\int_0^r e^{-a(r-s)}\,dW(s)
\right)^2dr
>0
\quad\text{a.s.}
\label{eq:local-nonsing-local-unity}
\end{align}
The limit is positive almost surely because the continuous Gaussian
process is not identically zero. If $\rho_T\to-1$, then
$(-1)^ty_t^\dagger
=(-\rho_T)(-1)^{t-1}y_{t-1}^\dagger+(-1)^tu_t$ and
$[(-1)^ty_t^\dagger]^2=(y_t^\dagger)^2$, so
\eqref{eq:local-nonsing-local-unity} holds with $-\rho_T\to1$.

Hence every sequence $\{\rho_T\}$ has a subsequence along which
$\widecheck\Sigma_{2,2}^\dagger
\xrightarrow{d}\operatorname{diag}(\sigma^2,L)$, where
$L=\sigma^2$ in the first case and $L>0$ almost surely in the second.
Fix $\delta>0$. If no such uniform lower bound existed, there would be
$c_k\downarrow0$ and $\{\rho_{T_k}\}$ such that
$\mathbb P_{\rho_{T_k}}
\{\lambda_{\min}(\widecheck\Sigma_{2,2}^\dagger)<c_k\}\ge\delta$,
contradicting
$\widecheck\Sigma_{2,2}^\dagger
\xrightarrow{d}\operatorname{diag}(\sigma^2,L)$ with $L>0$ almost surely. Hence there exists
$c_\Sigma>0$ such that
\begin{align}
\liminf_{T\to\infty}
\inf_{\rho_T\in[-1,1]}
\mathbb P_{\rho_T}
\left\{
\lambda_{\min}(\widecheck\Sigma_{2,2}^\dagger)\ge c_\Sigma
\right\}
&\ge
1-\delta.
\label{eq:local-nonsing-reference-lower}
\end{align}

\noindent\textbf{Lower bound for $\widecheck\Omega_{3,h}^\dagger$.}
Here
$e_{t,h}^\dagger=\sum_{\ell=1}^h\rho_T^{h-\ell}u_{t+\ell}$.
Equation~\eqref{eq:local-nonsing-reference-lower} and
Lemma~\ref{lem:rotated-product-bounds} give
\begin{align}
\begin{split}
    &\left\|
\widecheck\Omega_{3,h}^\dagger
-
\sum_{|j|\le b_T}K(j/b_T)T^{-1}\sum_t
(H_3Y_{t,3}^\dagger e_{t,h}^\dagger)
(H_3Y_{t-j,3}^\dagger e_{t-j,h}^\dagger)'
\right\|_2\\
=&
O_p(b_TT^{-1/2})+O_p(b_TT^{-1})
=
o_p(1).
\end{split}
\end{align}
For $|j|>h$, conditioning on the innovation history through $t$
gives zero mean, and Lemma~\ref{lem:rotated-product-bounds} gives
$O_p(b_TT^{-1/2})=o_p(1)$ after summing over these lags.
For $|j|\le h$, $K(j/b_T)\to1$. Expanding
$e_{t,h}^\dagger e_{t-j,h}^\dagger$ gives
$u_{t+\ell_1}u_{t-j+\ell_2}$,
$1\le\ell_1,\ell_2\le h$.
If $t+\ell_1\neq t-j+\ell_2$, the corresponding sample average is
$O_p(T^{-1/2})$ by the $L^2$ bound in
Lemma~\ref{lem:rotated-product-bounds} and Chebyshev's inequality.
If $t+\ell_1=t-j+\ell_2$, writing $u_s^2=\sigma^2+(u_s^2-\sigma^2)$ gives an
$O_p(T^{-1/2})$ centered term by Chebyshev's inequality. Therefore
\begin{align}
\widecheck\Omega_{3,h}^\dagger
&=
\frac{\sigma^2}{T}
\sum_s
\left(
\sum_{\ell=1}^h\rho_T^{h-\ell}H_3Y_{s-\ell,3}^\dagger
\right)
\left(
\sum_{\ell=1}^h\rho_T^{h-\ell}H_3Y_{s-\ell,3}^\dagger
\right)'
+
o_p(1).
\label{eq:local-nonsing-reference-hac-gram}
\end{align}

Since
$H_3Y_{s-\ell,3}^\dagger
=(u_{s-\ell},u_{s-\ell-1},\tau_T^{-1/2}y_{s-\ell-2}^\dagger)'$,
the upper-left $2\times2$ block in
\eqref{eq:local-nonsing-reference-hac-gram} converges uniformly by
Chebyshev's inequality to
$\sigma^4\bigl(\begin{smallmatrix}
\sum_{\ell=0}^{h-1}\rho_T^{2\ell}
&
\rho_T\sum_{\ell=0}^{h-2}\rho_T^{2\ell}
\\
\rho_T\sum_{\ell=0}^{h-2}\rho_T^{2\ell}
&
\sum_{\ell=0}^{h-1}\rho_T^{2\ell}
\end{smallmatrix}\bigr)$.
Its smallest eigenvalue divided by $\sigma^4$ is
$(1-|\rho_T|)\sum_{\ell=0}^{h-2}|\rho_T|^{2\ell}
+|\rho_T|^{2h-2}$, which is uniformly bounded away from zero.

For the third coordinate,
\begin{align}
\tau_T^{-1/2}\sum_{\ell=1}^h\rho_T^{h-\ell}y_{s-\ell-2}^\dagger
&=
\tau_T^{-1/2}
\left(\sum_{\ell=0}^{h-1}\rho_T^{2\ell}\right)
y_{s-h-2}^\dagger
+
\tau_T^{-1/2}\sum_{\ell=1}^{h-1}\rho_T^{h-\ell}
\sum_{k=0}^{h-\ell-1}\rho_T^ku_{s-\ell-2-k}.
\end{align}
The $L^2$ bound and Chebyshev's inequality give
$o_p(1)$ uniformly over $\rho_T$ for the sample cross-products of
$\tau_T^{-1/2}y_{s-h-2}^\dagger$ with
$\sum_{\ell=1}^h\rho_T^{h-\ell}u_{s-\ell}$,
$\sum_{\ell=1}^h\rho_T^{h-\ell}u_{s-\ell-1}$, and
$\tau_T^{-1/2}\sum_{\ell=1}^{h-1}\rho_T^{h-\ell}
\sum_{k=0}^{h-\ell-1}\rho_T^ku_{s-\ell-2-k}$.
The remaining contribution to the Schur complement is nonnegative.
Hence the Schur complement is bounded below by
$\sigma^2
\left(\sum_{\ell=0}^{h-1}\rho_T^{2\ell}\right)^2
(T\tau_T)^{-1}\sum_s(y_{s-h-2}^\dagger)^2-o_p(1)$.
Since $\sum_{\ell=0}^{h-1}\rho_T^{2\ell}\ge1$,
\eqref{eq:local-nonsing-reference-lower} gives a uniform positive
lower bound.

Therefore, for every $\delta>0$, there exists $c_\Omega>0$ such that
\begin{align}
\liminf_{T\to\infty}
\inf_{\rho_T\in[-1,1]}
\mathbb P_{\rho_T}
\left\{
\lambda_{\min}(\widecheck\Omega_{3,h}^\dagger)\ge c_\Omega
\right\}
&\ge
1-\delta.
\label{eq:local-nonsing-reference-omega-lower}
\end{align}
Taking $\delta=\varepsilon/2$ in
\eqref{eq:local-nonsing-reference-lower} and
\eqref{eq:local-nonsing-reference-omega-lower}, the union bound gives
\eqref{eq:local-nonsing-corr} with
$\bar c=\min\{c_\Sigma,c_\Omega\}$.

\textbf{\textit{Proof of \eqref{eq:local-nonsing-sigma-gap}}.}
For all sufficiently large $T$, $|\phi_T|<1$, and
\begin{align}
\xi_t-u_t
&=
(\phi_T+\theta_T)
\sum_{\ell=1}^{\infty}\phi_T^{\ell-1}u_{t-\ell}.
\label{eq:local-nonsing-xi-gap}
\end{align}
Hence
$\sup_{t\le T}\|\xi_t-u_t\|_{2r}\le CT^{-\bar\zeta}$.
Since $y_0=y_0^\dagger=0$,
$y_t-y_t^\dagger
=
(\phi_T+\theta_T)
\sum_{s=1}^t
\rho_T^{t-s}
\sum_{\ell=1}^{\infty}
\phi_T^{\ell-1}u_{s-\ell}$.
Moreover,
$|\phi_T+\theta_T|
(\sum_{s=0}^{t-1}\rho_T^{2s})^{1/2}
\sum_{\ell\ge0}|\phi_T|^\ell
\le CT^{-\bar\zeta}\tau_T^{1/2}$.
The weighted Rosenthal bound in
Lemma~\ref{lem:local-filter-approximation} therefore gives
$\sup_{t\le T}\|\tau_T^{-1/2}(y_t-y_t^\dagger)\|_{2r}
\le CT^{-\bar\zeta}$ and
\begin{align}
\sup_{t\le T}
\left\|
H_jY_{t,j}
-
H_jY_{t,j}^\dagger
\right\|_{2r}
&\le
CT^{-\bar\zeta},
\qquad j\in\{2,3\}.
\label{eq:local-nonsing-rotated-gap}
\end{align}

Using
\begin{align}
\begin{split}
&H_2Y_{t,2}Y_{t,2}'H_2'
-H_2Y_{t,2}^\dagger Y_{t,2}^{\dagger\prime}H_2' \\
&=
(H_2Y_{t,2}-H_2Y_{t,2}^\dagger)(H_2Y_{t,2}^\dagger)'
+
H_2Y_{t,2}^\dagger(H_2Y_{t,2}-H_2Y_{t,2}^\dagger)' \\
&\quad+
(H_2Y_{t,2}-H_2Y_{t,2}^\dagger)
(H_2Y_{t,2}-H_2Y_{t,2}^\dagger)',
\end{split}
\end{align}
By \eqref{eq:local-nonsing-rotated-gap} and the uniform moment bounds,
$\|\widecheck\Sigma_{2,2}-\widecheck\Sigma_{2,2}^\dagger\|_2
=O_p(T^{-\bar\zeta})=o_p(1)$.

\textbf{\textit{Proof of \eqref{eq:local-nonsing-omega-gap}}.}
Equations~\eqref{eq:local-nonsing-reference-lower} and
\eqref{eq:local-nonsing-sigma-gap} give
$\|\widecheck\Sigma_{2,2}^{-1}\|_2+
\|(\widecheck\Sigma_{2,2}^\dagger)^{-1}\|_2=O_p(1)$.

Define
$\widetilde\Omega_{3,h}
:=
\sum_{|j|\le b_T}K(j/b_T)T^{-1}\sum_t
(H_3Y_{t,3}e_{t,h,T})
(H_3Y_{t-j,3}e_{t-j,h,T})'$
and\\
$\widetilde\Omega_{3,h}^\dagger
:=
\sum_{|j|\le b_T}K(j/b_T)T^{-1}\sum_t
(H_3Y_{t,3}^\dagger e_{t,h}^\dagger)
(H_3Y_{t-j,3}^\dagger e_{t-j,h}^\dagger)'$.
Then
\begin{align}
\left\|
\widecheck\Omega_{3,h}
-
\widecheck\Omega_{3,h}^\dagger
\right\|_2
&\le
\left\|
\widecheck\Omega_{3,h}
-
\widetilde\Omega_{3,h}
\right\|_2
+
\left\|
\widetilde\Omega_{3,h}
-
\widetilde\Omega_{3,h}^\dagger
\right\|_2
+
\left\|
\widetilde\Omega_{3,h}^\dagger
-
\widecheck\Omega_{3,h}^\dagger
\right\|_2.
\label{eq:local-nonsing-omega-decomposition}
\end{align}
Proposition~\ref{prop:local-lalp-representation},
Lemma~\ref{lem:local-score-equivalence},
\eqref{eq:local-nonsing-reference-lower}, and
Lemma~\ref{lem:rotated-product-bounds} give
$\|\widecheck\Omega_{3,h}-\widetilde\Omega_{3,h}\|_2
+
\|\widetilde\Omega_{3,h}^\dagger-\widecheck\Omega_{3,h}^\dagger\|_2
=
O_p(b_TT^{-1/2})+O_p(b_TT^{-1})
=
o_p(1)$.

Moreover,
$H_3Y_{t,3}e_{t,h,T}
-H_3Y_{t,3}^\dagger e_{t,h}^\dagger
=
(H_3Y_{t,3}-H_3Y_{t,3}^\dagger)e_{t,h}^\dagger
+
H_3Y_{t,3}(e_{t,h,T}-e_{t,h}^\dagger)$.
By \eqref{eq:local-nonsing-rotated-gap} and the uniform moment bounds,
\begin{align}
\|H_3Y_{t,3}e_{t,h,T}
-H_3Y_{t,3}^\dagger e_{t,h}^\dagger\|_r
&=O(T^{-\bar\zeta})
\end{align}
uniformly over $\rho_T$.

For $|j|\le h$,
\begin{align}
\left\|
T^{-1}\sum_t
\left[
H_3Y_{t,3}e_{t,h,T}e_{t-j,h,T}(H_3Y_{t-j,3})'
-
H_3Y_{t,3}^\dagger e_{t,h}^\dagger
e_{t-j,h}^\dagger(H_3Y_{t-j,3}^\dagger)'
\right]
\right\|_2
&=
O_p(T^{-\bar\zeta}).
\end{align}

For $h<|j|\le b_T$,
\begin{align}
\begin{split}
&T^{-1}\sum_t
\left[
H_3Y_{t,3}e_{t,h,T}e_{t-j,h,T}(H_3Y_{t-j,3})'
-
H_3Y_{t,3}^\dagger e_{t,h}^\dagger
e_{t-j,h}^\dagger(H_3Y_{t-j,3}^\dagger)'
\right]
\\
&=
T^{-1}\sum_t
\left[
H_3Y_{t,3}e_{t,h,T}
-
H_3Y_{t,3}^\dagger e_{t,h}^\dagger
\right]
e_{t-j,h}^\dagger(H_3Y_{t-j,3}^\dagger)'
\\
&\quad+
T^{-1}\sum_t
H_3Y_{t,3}^\dagger e_{t,h}^\dagger
\left[
H_3Y_{t-j,3}e_{t-j,h,T}
-
H_3Y_{t-j,3}^\dagger e_{t-j,h}^\dagger
\right]'
\\
&\quad+
T^{-1}\sum_t
\left[
H_3Y_{t,3}e_{t,h,T}
-
H_3Y_{t,3}^\dagger e_{t,h}^\dagger
\right]
\left[
H_3Y_{t-j,3}e_{t-j,h,T}
-
H_3Y_{t-j,3}^\dagger e_{t-j,h}^\dagger
\right]'.
\end{split}
\end{align}
The three terms are respectively
$O_p(T^{-\bar\zeta-1/2})$,
$O_p(T^{-\bar\zeta-1/2})$, and
$O_p(T^{-2\bar\zeta})$,
uniformly over $h<|j|\le b_T$.
Hence
\begin{align}
\|\widecheck\Omega_{3,h}-\widecheck\Omega_{3,h}^\dagger\|_2
&=
O_p(T^{-\bar\zeta})
+O_p(b_TT^{-1/2})
+O_p(b_TT^{-\bar\zeta-1/2})
+O_p(b_TT^{-2\bar\zeta})
=o_p(1).
\label{eq:local-nonsing-hac-rate}
\end{align}
Since $b_T=o(T^{1/2})$ and $\bar\zeta>1/4$, this proves
\eqref{eq:local-nonsing-omega-gap}.

\vspace{0.5\baselineskip}

\section{Strict efficiency of optimal GMM}
\label{app:gmm_vs_ols}

This appendix specializes the moment-redundancy condition of
\cite{breusch1999redundancy} to a time-series regression with controls
and additional instruments. After partialling out the controls,
redundancy is characterized by equality between a contemporaneous
projection coefficient and its long-run counterpart.\footnote{This
also gives the information interpretation of \cite{hall2007information}:
an additional moment is informative if and only if the component of its
score that is long-run orthogonal to the baseline score has a nonzero
derivative with respect to the parameter.}

Consider
\begin{align}
\label{eq:lemma_model}
y_t
=
x_t'\beta+w_t'\gamma+u_t,
\end{align}
where $x_t\in\mathbb R^{d_x}$, $w_t\in\mathbb R^{d_w}$, and
$z_t\in\mathbb R^{d_z}$. We maintain
$\mathbb E[x_tu_t]=0$, $\mathbb E[w_tu_t]=0$,
$\mathbb E[z_tu_t]=0$, and nonsingularity of
$\Sigma_{ww}:=\mathbb E[w_tw_t']$. We assume stationarity and weak
dependence throughout.

Set
$y_t^\perp:=y_t-P(y_t\mid w_t)$,
$x_t^\perp:=x_t-P(x_t\mid w_t)$, and
$z_t^\perp:=z_t-P(z_t\mid w_t)$.
Since $\mathbb E[w_tu_t]=0$,
\begin{align}
\label{eq:residualized_model}
y_t^\perp
=
x_t^{\perp\prime}\beta+u_t.
\end{align}
Let
$Z_t^\perp:=(x_t^{\perp\prime},z_t^{\perp\prime})'$ and assume
$\Sigma_{xx}^\perp:=\mathbb E[x_t^\perp x_t^{\perp\prime}]$
is nonsingular. Under the stationarity, mixing, and moment conditions
of the main text,
$\Omega^\perp
:=
\sum_{\ell=-\infty}^{\infty}
\mathbb E[
Z_t^\perp u_tu_{t-\ell}Z_{t-\ell}^{\perp\prime}]$
exists, is positive definite, and admits a consistent estimator.
Partition
$\Omega^\perp
=
\begin{pmatrix}
\Omega_{xx}^\perp & \Omega_{xz}^\perp\\
\Omega_{xz}^{\perp\prime} & \Omega_{zz}^\perp
\end{pmatrix}$,
and let
$\Sigma_{xz}^\perp:=\mathbb E[x_t^\perp z_t^{\perp\prime}]$ and
$S^\perp
:=
\Omega_{zz}^\perp
-
\Omega_{xz}^{\perp\prime}
(\Omega_{xx}^\perp)^{-1}
\Omega_{xz}^\perp$.
Since $\Omega^\perp\succ0$, $S^\perp\succ0$.

Using $x_t^\perp$ alone gives the OLS estimator by
Frisch--Waugh--Lovell, with
$V^{\mathrm{OLS}}_\beta
=(\Sigma_{xx}^\perp)^{-1}
\Omega_{xx}^\perp
(\Sigma_{xx}^\perp)^{-1}$.
Using $Z_t^\perp$ with the optimal weight gives
$V^{\mathrm{GMM}}_\beta
=
\left(
\mathbb E[x_t^\perp Z_t^{\perp\prime}]
(\Omega^\perp)^{-1}
\mathbb E[Z_t^\perp x_t^{\perp\prime}]
\right)^{-1}$.
Define
\begin{align}
\label{eq:lemma_R}
R^\perp
:=
\Sigma_{xz}^\perp
-
\Sigma_{xx}^\perp
\bigl(\Omega_{xx}^\perp\bigr)^{-1}
\Omega_{xz}^\perp
\in
\mathbb R^{d_x\times d_z}.
\end{align}
The score
$z_t^\perp u_t-
\Omega_{xz}^{\perp\prime}
(\Omega_{xx}^\perp)^{-1}x_t^\perp u_t$
is long-run orthogonal to $x_t^\perp u_t$, and its derivative with
respect to $\beta$ is $-R^{\perp\prime}$. Hence the additional moments
are redundant if and only if $R^\perp=0$.

\begin{lemma}[Strict efficiency of optimal GMM for $\beta$]
\label{lem:gmm_vs_ols}
Under the conditions above,
\begin{align}
\label{eq:lemma_gain}
V^{\mathrm{OLS}}_\beta-V^{\mathrm{GMM}}_\beta
=
V^{\mathrm{OLS}}_\beta R^\perp
\left(
S^\perp+R^{\perp\prime}V^{\mathrm{OLS}}_\beta R^\perp
\right)^{-1}
R^{\perp\prime}V^{\mathrm{OLS}}_\beta
\succeq0,
\end{align}
with equality if and only if $R^\perp=0$. Moreover,
$\operatorname{rank}
(V^{\mathrm{OLS}}_\beta-V^{\mathrm{GMM}}_\beta)
=
\operatorname{rank}(R^\perp)$,
and the column space of
$V^{\mathrm{OLS}}_\beta-V^{\mathrm{GMM}}_\beta$
equals that of
$V^{\mathrm{OLS}}_\beta R^\perp$.
\end{lemma}

\begin{proof}
Let
$Q^{\mathrm{OLS}}
:=\Sigma_{xx}^\perp
(\Omega_{xx}^\perp)^{-1}
\Sigma_{xx}^\perp$, so that
$V^{\mathrm{OLS}}_\beta=(Q^{\mathrm{OLS}})^{-1}$.
Block inversion gives
\begin{align}
\bigl(\Omega^\perp\bigr)^{-1}
=
\begin{pmatrix}
\bigl(\Omega_{xx}^\perp\bigr)^{-1}
+
\bigl(\Omega_{xx}^\perp\bigr)^{-1}
\Omega_{xz}^\perp
\bigl(S^\perp\bigr)^{-1}
\Omega_{xz}^{\perp\prime}
\bigl(\Omega_{xx}^\perp\bigr)^{-1}
&
-\bigl(\Omega_{xx}^\perp\bigr)^{-1}
\Omega_{xz}^\perp
\bigl(S^\perp\bigr)^{-1}
\\[2pt]
-\bigl(S^\perp\bigr)^{-1}
\Omega_{xz}^{\perp\prime}
\bigl(\Omega_{xx}^\perp\bigr)^{-1}
&
\bigl(S^\perp\bigr)^{-1}
\end{pmatrix}.
\end{align}
Substituting
$\mathbb E[x_t^\perp Z_t^{\perp\prime}]
=(\Sigma_{xx}^\perp,\Sigma_{xz}^\perp)$
and using \eqref{eq:lemma_R} gives
$\mathbb E[x_t^\perp Z_t^{\perp\prime}]
(\Omega^\perp)^{-1}
\mathbb E[Z_t^\perp x_t^{\perp\prime}]
=
Q^{\mathrm{OLS}}
+
R^\perp(S^\perp)^{-1}R^{\perp\prime}$.
Therefore
$V^{\mathrm{GMM}}_\beta
=
\bigl(
Q^{\mathrm{OLS}}
+
R^\perp(S^\perp)^{-1}R^{\perp\prime}
\bigr)^{-1}$.

The Woodbury identity\footnote{The Woodbury identity states the following fact. For conformable matrices $A$, $U$, $C$, and
$V$, with $A$ and $C$ nonsingular,
$(A+UCV)^{-1}
=
A^{-1}
-
A^{-1}U
(C^{-1}+VA^{-1}U)^{-1}
VA^{-1}$,
provided the displayed inverses exist.
See \citet[Eq.~(0.7.4.1)]{horn2013matrix}.}
gives
\begin{align}
V^{\mathrm{GMM}}_\beta
=
V^{\mathrm{OLS}}_\beta
-
V^{\mathrm{OLS}}_\beta R^\perp
\left(
S^\perp+R^{\perp\prime}V^{\mathrm{OLS}}_\beta R^\perp
\right)^{-1}
R^{\perp\prime}V^{\mathrm{OLS}}_\beta,
\end{align}
which proves \eqref{eq:lemma_gain}.

Since $V^{\mathrm{OLS}}_\beta\succ0$ and $S^\perp\succ0$,
$S^\perp+R^{\perp\prime}V^{\mathrm{OLS}}_\beta R^\perp\succ0$.
Therefore
$V^{\mathrm{OLS}}_\beta-V^{\mathrm{GMM}}_\beta\succeq0$, with equality
if and only if
$R^{\perp\prime}V^{\mathrm{OLS}}_\beta=0$, equivalently
$R^\perp=0$.

Since
$\bigl(
S^\perp+R^{\perp\prime}V^{\mathrm{OLS}}_\beta R^\perp
\bigr)^{-1}\succ0$, $\operatorname{rank}
\bigl(
V^{\mathrm{OLS}}_\beta-V^{\mathrm{GMM}}_\beta
\bigr)=
\operatorname{rank}(R^\perp). $ Moreover, the column space of
$V^{\mathrm{OLS}}_\beta-V^{\mathrm{GMM}}_\beta$
equals that of
$V^{\mathrm{OLS}}_\beta R^\perp$.
\end{proof}

The condition $R^\perp=0$ is equivalently
\begin{align}
\label{eq:R_zero_equiv}
\bigl(\Sigma_{xx}^\perp\bigr)^{-1}\Sigma_{xz}^\perp
=
\bigl(\Omega_{xx}^\perp\bigr)^{-1}\Omega_{xz}^\perp.
\end{align}
The two sides of \eqref{eq:R_zero_equiv} are, respectively, the
contemporaneous projection coefficient of $z_t^\perp$ on $x_t^\perp$
and the long-run projection coefficient of $z_t^\perp u_t$ on
$x_t^\perp u_t$. Thus the additional moments are redundant precisely
when these coefficients coincide.

Several cases follow directly. If
$z_t^\perp=Ax_t^\perp$, then
$\Sigma_{xz}^\perp=\Sigma_{xx}^\perp A'$ and
$\Omega_{xz}^\perp=\Omega_{xx}^\perp A'$, so $R^\perp=0$.
If $Z_t^\perp u_t$ is serially uncorrelated and
$\mathbb E[u_t^2\mid Z_t^\perp]=\sigma^2$, then
$\Omega_{xx}^\perp=\sigma^2\Sigma_{xx}^\perp$ and
$\Omega_{xz}^\perp=\sigma^2\Sigma_{xz}^\perp$, again implying
$R^\perp=0$.
If $\Sigma_{xz}^\perp=0$, then
\eqref{eq:R_zero_equiv} holds if and only if
$\Omega_{xz}^\perp=0$.

Hence strict efficiency fails only when the contemporaneous and
long-run projection coefficients in \eqref{eq:R_zero_equiv} coincide.
Serial correlation or conditional heteroskedasticity can make
$R^\perp\neq0$, but neither is sufficient. In particular, if the
scores are serially uncorrelated and conditionally homoskedastic, then
$\Omega^\perp
=\sigma^2\mathbb E[Z_t^\perp Z_t^{\perp\prime}]$ and
$R^\perp=0$.

\vspace{0.5\baselineskip}

\section{Bootstrap implementation}
\label{app:bootstrap}

Asymptotic normal inference based on \eqref{eq:Vpq_hat} is convenient,
but finite-sample distortions may remain. Following
\citet{montiel2021local}, we therefore use a recursive wild bootstrap
adapted to GMM--LP. Fix a nonzero vector $\nu\in\mathbb R^K$ and consider inference on
$\nu'\beta_h$. Let
$\widehat s_{h,\nu}^{\mathrm{GMM}(p,q)}
:=
\left(
T^{-1}\nu'\widehat V_h^{\mathrm{GMM}(p,q)}\nu
\right)^{1/2}$.

The bootstrap proceeds as follows.
\begin{enumerate}[label=(\roman*)]

\item
Estimate the VAR($p$) by OLS,
$
y_t
=
\hat c+\widehat A_1y_{t-1}+\cdots+\widehat A_py_{t-p}
+\widehat u_t.
$
Let $\widehat{\mathbf A}$ denote the associated companion matrix and set
$\widehat\beta_h^{\mathrm{VAR}}
=
J_p(\widehat{\mathbf A}')^hJ_p'\nu_1$.
The bootstrap statistic is centered at
$\nu'\widehat\beta_h^{\mathrm{VAR}}$.

\item
For each $b=1,\ldots,B$, draw i.i.d.\
$\xi_t^{(b)}\sim N(0,1)$ and set
$u_t^{*(b)}=\xi_t^{(b)}\widehat u_t$.
Generate
$
y_t^{*(b)}
=
\hat c+\widehat A_1y_{t-1}^{*(b)}
+\cdots+\widehat A_py_{t-p}^{*(b)}
+u_t^{*(b)},
\qquad
t=p+1,\ldots,T.
$

\item
Re-estimate the horizon-$h$ GMM--LP from
$\{y_t^{*(b)}\}_{t=1}^T$ using the same $(p,q)$ and kernel as in the
original sample. The horizon-$h$ regression variables are demeaned
before residualizing with respect to the lag controls, as in the
original sample. The bandwidth is computed once from the original
sample, horizon by horizon, and held fixed across replications.

This gives
$\widehat\beta_h^{*(b),\mathrm{GMM}(p,q)}$ and
$\widehat V_h^{*(b),\mathrm{GMM}(p,q)}$, with
$\widehat s_{h,\nu}^{*(b),\mathrm{GMM}(p,q)}
:=
\left(
T^{-1}\nu'\widehat V_h^{*(b),\mathrm{GMM}(p,q)}\nu
\right)^{1/2}$.

\item
Form
$\widehat T_{h,\nu}^{*(b)}
=
\left[
\nu'\widehat\beta_h^{*(b),\mathrm{GMM}(p,q)}
-
\nu'\widehat\beta_h^{\mathrm{VAR}}
\right]
/
\widehat s_{h,\nu}^{*(b),\mathrm{GMM}(p,q)}$.
Let $\widehat c_{h,\nu}(\tau)$ be the empirical $\tau$-quantile of
$\{\widehat T_{h,\nu}^{*(b)}\}_{b=1}^B$. The equal-tailed
percentile-$t$ confidence interval is
$\mathrm{CI}_{h,\nu}^{\mathrm{boot}}
=
\Bigl[
\nu'\widehat\beta_h^{\mathrm{GMM}(p,q)}
-
\widehat s_{h,\nu}^{\mathrm{GMM}(p,q)}
\widehat c_{h,\nu}(1-\alpha/2),
\,
\nu'\widehat\beta_h^{\mathrm{GMM}(p,q)}
-
\widehat s_{h,\nu}^{\mathrm{GMM}(p,q)}
\widehat c_{h,\nu}(\alpha/2)
\Bigr]$.

\end{enumerate}

\section{90\% confidence intervals in the empirical application}
\label{app:IRF}

In Figure \ref{fig:application-epu-shock} we provide additional results on confidence intervals for the empirical application in Section \ref{sec:application}.

\begin{figure}[htbp]
\makebox[\linewidth][c]{
\includegraphics[width=1.1\textwidth]{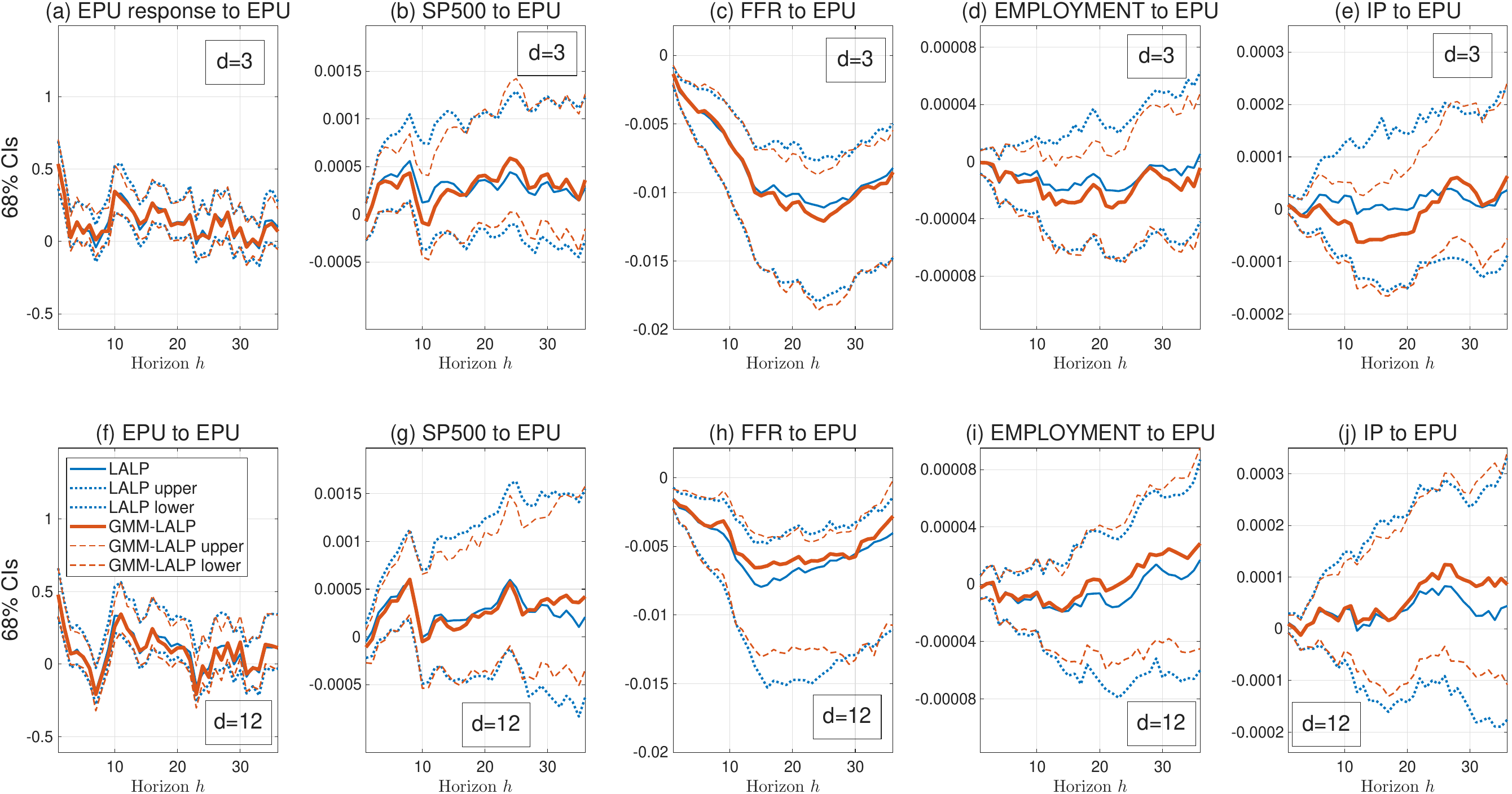}
}
\caption{Responses to an EPU innovation.}
\label{fig:application-epu-shock}
\begin{minipage}{1\textwidth}
\vspace{2pt}
\footnotesize
\textit{Notes:} The figure plots the impulse responses for EPU, SP500,
FFR, EMPLOYMENT, and IP to an EPU innovation with 90\% confidence intervals. Blue lines are LALP estimates and
bootstrap confidence intervals; orange lines are GMM-LALP estimates and bootstrap confidence intervals. The
model lag order is $d=3$ for the upper panels, and is $d=12$ for the lower panels. The horizons run from 1 to 36 months.
\end{minipage}
\end{figure}

\renewcommand{\refname}{References for the Appendices}
\putbib

\end{bibunit}

\end{document}